\pdfoutput=1
\documentclass[a4paper,USenglish,thm-restate]{lipics-v2021}
\hideLIPIcs

\usepackage{amssymb}
\usepackage{amsmath}
\usepackage{amsthm}
\usepackage{xcolor}
\usepackage[most]{tcolorbox}
\usepackage[T1]{fontenc}
\usepackage[normalem]{ulem}
\usepackage{array}
\usepackage{cite}

\newcommand{\cm}{\operatorname{cm}}
\newcommand{\cvm}{\operatorname{cvm}}
\newcommand{\solc}{\operatorname{sm}}
\newcommand{\src}{\mathrm{s}}
\newcommand{\tgt}{\mathrm{t}}
\newcommand{\rowsub}{\mathrm{Row}}
\newcommand{\colsub}{\mathrm{Col}}
\newcommand{\cand}{\mathrm{Cand}}

\newtheoremstyle{problemstyle}
  {0.5em}{0.5em}{\normalfont}{}{\bfseries}{.}{0.5em}{}
\theoremstyle{problemstyle}

\title{On the complexity of the single-move labeled token routing problem}

\titlerunning{Single-Move Labeled Token Routing}

\author{Nicolas Bousquet}{CNRS, Univ.\ Lyon, Universit\'e Lyon 1, INSA Lyon, LIRIS UMR CNRS 5205, F-69621 Villeurbanne, France}{nicolas.bousquet@cnrs.fr}{https://orcid.org/0000-0003-0170-0503}{}
\author{Remy {El Sabeh}}{David R. Cheriton School of Computer Science, University of Waterloo, Waterloo, ON, Canada}{relsabeh@uwaterloo.ca}{https://orcid.org/0000-0002-3786-4035}{}
\author{Amer E. Mouawad}{American University of Beirut, Beirut, Lebanon \and David R. Cheriton School of Computer Science, University of Waterloo, Waterloo, ON, Canada}{aa368@aub.edu.lb}{}{}
\author{Naomi Nishimura}{David R. Cheriton School of Computer Science, University of Waterloo, Waterloo, ON, Canada}{nishi@uwaterloo.ca}{}{}

\authorrunning{N. Bousquet, R. El Sabeh, A. E. Mouawad, and N. Nishimura}

\Copyright{Nicolas Bousquet, Remy El Sabeh, Amer E. Mouawad, and Naomi Nishimura}



\begin{document}
\nolinenumbers

\maketitle

\begin{abstract}
In neutral-atom quantum simulators and computers, loaded atoms must be moved
to desired target positions along paths of positions that are empty at the
time of each move, and a target position may be reserved for one species of atom. Motivated by this task, we introduce \textsc{Single-Move
Labeled Token Routing}: every source and every target vertex of a graph is
assigned a set of labels, and tokens occupy the sources. A solution consists
of a matching that assigns each source to a \emph{compatible} target (one whose label
set intersects its own), a route for each matched pair, and a \emph{movement
order} in which, when a token is moved, its route contains no other token.

When every source is
compatible with every target, the problem is solvable in polynomial
time~\cite{Calinescu2008Reconfigurations}; when each source is compatible
with exactly one target, it is \textsf{NP}-complete on grid
graphs~\cite{GeftHalperinNakar2024}. We prove that, with exactly one
compatible target per source, the problem remains \textsf{NP}-complete on
grids and on planar graphs of maximum degree four even when the routes of a
solution can be chosen pairwise edge-disjoint.

Our main results concern trees, where each compatible pair is joined by a
unique path, so a solution consists of a matching and a movement order; the
problem is known to be \textsf{NP}-complete even on trees of maximum degree
three~\cite{Geft2026Trees}. Inspired by the edge-disjointness of the routes
in our planar result, we study the problem on trees through two
parameters: the \emph{solution edge multiplicity} is
the largest number of routes of a solution that use a common edge, and the
\emph{candidate edge multiplicity} is the largest number of compatible
source-target pairs whose paths use a common edge. On trees of maximum degree
three, the problem is \textsf{W}[1]-hard when parameterized by a bound on the
solution edge multiplicity, even when a movement order is given as part of the input;
on trees of unbounded degree, it is \textsf{NP}-complete even when the
candidate edge multiplicity is at most eight. Our main algorithmic result
is a dynamic programming algorithm showing that on trees the problem is
fixed-parameter tractable when parameterized by the maximum degree together
with the candidate edge multiplicity, and also when parameterized by the
candidate vertex multiplicity alone, the largest number of compatible
source-target pairs whose paths contain a common vertex. Our hardness results
demonstrate that neither of the first two parameters can be omitted:
unless $\textsf{P}=\textsf{NP}$, the problem is not fixed-parameter tractable when parameterized by the
candidate edge multiplicity alone or by the maximum degree alone.
\end{abstract}

\section{Introduction}
\label{sec:intro}

Neutral-atom quantum simulators and computers use optical traps to hold
individual atoms. The
goal is to place atoms in the \emph{target traps}, a chosen subset of the
optical traps; the interactions that define the quantum system occur only
between nearby atoms, so the resulting arrangement of atoms, the \emph{target
configuration}, determines the structure of the system being simulated
or used for computation~\cite{Endres2016AtomAssembly,Barredo2016AtomAssembler,%
OhlDeMello2019DefectFree,Schymik2020EnhancedAssembly}. Loading atoms directly
into the target traps may not succeed, as the loading process does not guarantee
that every target trap receives an atom. To increase the chances of obtaining
the target configuration, one may instead load atoms into a larger set of traps
and then rearrange the atoms that were successfully loaded. Some atoms may
already occupy target traps and need not move, while other atoms are moved from
occupied traps to empty target traps, each along a path of traps that are
empty when the atom is moved. Computing the rearrangement, a task known as
\textsc{Atom Reconfiguration}, requires choosing which atoms should move to
which target traps, choosing routes for the moving atoms, and choosing a
movement order so that no atom blocks a route when that route is
used~\cite{Schymik2020EnhancedAssembly,ElSabeh2023ARO,Cimring2023RedRec,%
Afiouni2025BirdBatching}. Each additional move takes time and creates another
opportunity for errors or atom loss, which motivates keeping the number of
moves of each atom small.

\textsc{Atom Reconfiguration} admits a graph-theoretic formulation. Vertices
represent optical traps, source vertices are initially occupied by tokens
representing atoms, and target vertices correspond to positions in the target
configuration. A solution assigns sources to targets, chooses a route for each
assigned pair, and orders the corresponding moves. Every move takes a token from
its current vertex to an unoccupied vertex along a path no other vertex of which
is occupied, the target included~\cite{Calinescu2008Reconfigurations}. In
general, a token may be moved several times. We study the restriction in which
each token moves at most once.

\smallskip
\noindent\textbf{Atoms that are not interchangeable.}
In \textsc{Atom Reconfiguration}, atoms are interchangeable: any
loaded atom may be moved into any empty target trap. This is no longer the case
as soon as atoms of two \emph{species}, that is, of two chemical elements or of
two isotopes of one element, are loaded together, as in dual-species
settings, also called
mixed-species settings~\cite{Anand2024DualSpeciesRydberg,Fang2025Interleaved,Sheng2022MixedSpecies,Wei2026MixedSpecies}.
Mixtures of two species are loaded deliberately, for instance to form a
rubidium-caesium molecule from one atom of each species at every designated
site~\cite{Ruttley2023MergingTweezers,Ruttley2024MoleculeArrays}. In these
settings, a target trap may be designated for one species, so an atom of the
wrong species is not an acceptable occupant~\cite{Sheng2022MixedSpecies,Wei2026MixedSpecies}.
Thus, exchanging two atoms may turn a valid target configuration into an
invalid one.

Restricting, for each target trap, which atoms may occupy it gives rise
to the problem studied in this paper, \textsc{Single-Move Labeled Token Routing}, in which each source and
each target is assigned a set of labels, and a source may be matched to a target
if and only if their label sets intersect; in this case the source and the
target are said to be \emph{compatible}. Assigning the same label to every
source and every target yields \textsc{Atom Reconfiguration} restricted to at
most one move per atom, which
we call the \emph{all-matching case}; whether a rearrangement exists can
then be decided in polynomial
time~\cite{Calinescu2008Reconfigurations,ElSabeh2023ARO,Cooper2024ParameterizedAtoms}.
At the opposite extreme, which we call the \emph{unique-matching case},
every source is compatible with exactly one target and every target with
exactly one source. Between these two extremes, the set of compatible source-target pairs can be
arbitrary.

Our goal is to understand how the labels assigned to sources and targets,
together with the structure of the underlying graph, affect the complexity of
\textsc{Single-Move Labeled Token Routing}. We study several graph classes and
parameterizations, proving both negative and positive results that focus on one question: how much information must an algorithm carry across a cut of
the graph? Bounding the treewidth, or even the pathwidth, of the graph does not make
the problem tractable, since the problem is hard even on graphs of pathwidth
three~\cite{GeftHalperinNakar2024}. As we prove in this paper, bounding how many routes of a solution may
share an edge does not either, since the problem is hard on grids even when
the routes of a solution can be chosen pairwise edge-disjoint.
Our main algorithmic result concerns trees, where every edge is a cut:
bounding, for every edge, the number of compatible source-target pairs whose
path uses that edge still does not make the problem tractable, but adding a
bound on the maximum degree does, as does bounding, for every vertex, the
number of compatible pairs whose paths contain that vertex.

The paper is organized as follows. Section~\ref{sec:related} discusses related
work and Section~\ref{sec:our_contributions} presents an overview of our
results. Section~\ref{sec:prelim} introduces the necessary notation.
Section~\ref{sec:hardness} proves the hardness results,
Section~\ref{sec:positive} the algorithmic ones, and
Section~\ref{sec:conclusion} concludes with open problems.

\section{Related work}
\label{sec:related}

\textsc{Atom Reconfiguration}, the problem motivating our work, is the setting
in which every atom may move to every empty target trap, and in which one
typically seeks to fill every target trap while minimizing some measure of
the rearrangement. Atom reconfiguration has been carried out
experimentally for one-, two-, and three-dimensional target
configurations~\cite{Endres2016AtomAssembly,Barredo2016AtomAssembler,%
Barredo2018Synthetic3D,OhlDeMello2019DefectFree,Schymik2020EnhancedAssembly}.
On the algorithmic side, the assignment of atoms to target traps was first
formulated as a minimum-cost matching problem~\cite{Lee2017Hungarian}. Later
algorithms treat the routes and the movement order as parts of the problem
in their own right: after the targets and the routes are chosen, an order in
which to execute the moves must still be
found~\cite{Cimring2023RedRec,ElSabeh2023ARO,Afiouni2025BirdBatching}. All
of these algorithms assume interchangeable atoms and allow each atom to
move more than once. The measures being minimized increase with the number of
moves. Whenever the target traps can be filled at all, they can be filled
by moving every atom at most
once~\cite{Calinescu2008Reconfigurations}, a fact we reprove in our
terminology. The all-matching case of our problem is exactly
\textsc{Atom Reconfiguration} with as many atoms as target traps, restricted to at most
one move per atom. Deciding whether
every target trap can then be filled takes polynomial time
(Theorem~\ref{thm:allMatching}). Reconfiguration algorithms for two species
of atoms have also been
proposed~\cite{Tao2022DualSpeciesAlgorithm,Harle2026Atommovr}. The
parameterized complexity of atom reconfiguration is known for
interchangeable atoms and for the case in which each atom has its own
target: most tractability results for the former become hardness results
for the latter~\cite{Cooper2024ParameterizedAtoms}. In this paper, we study
what lies between these two cases: each atom may be compatible with an
arbitrary set of targets.

Other models of token movement differ from ours in what a single move is
and in which targets a token may be assigned to. We use the word
\emph{token} as a generic term for a moving object; the works below call
their objects pebbles, agents, robots, or coins. Our move, along a path no
vertex of which other than its first is occupied, was studied by C\u{a}linescu, Dumitrescu, and
Pach~\cite{Calinescu2008Reconfigurations} for interchangeable tokens, with
the objective of minimizing the number of moves. They prove that at most $k$ moves suffice, where $k$ is the number of
tokens, and that such a sequence can be chosen with no token moving twice.
They also prove that minimizing the number of moves is \textsf{NP}-hard even on
grids. A sequence of moves in which no token moves twice is a solution of our
all-matching case, so their setting and ours differ only in the objective: they
minimize the number of moves, whereas we decide feasibility. In \textsc{Pebble Motion}, a move takes a token along
a single edge instead, and feasibility is decidable in linear
time~\cite{Kornhauser1984PebbleMotion,Auletta1999PebbleTrees,%
Yu2013LinearPebble,Goraly2010MulticolorPebble}. In other models, several
tokens move along edges at the same
time~\cite{YuLaValle2016OptimalMPP,Stern2019MAPFDefinitions,Ma2022GraphMAPF,%
Kedia2022PCMAPF,Surynek2019CapacityMAPF,Chong2024TAPFPTC}, or two tokens are
swapped along the edge joining
them~\cite{Yamanaka2015SwappingLabeled,Miltzow2016Approximation,%
Alon1994RoutingMatchings,Kawahara2019PermutationRouting}. The assignment of
tokens to targets takes one of three forms: it is fixed, each token having
its own target; it is given by \emph{colors}, each token and each target
carrying one color and a token being assignable to every target of its
color~\cite{Goraly2010MulticolorPebble,MaKoenig2016TAPF,Bonnet2017Complexity,%
Geft2026Trees}; or each token carries a set of targets to which it may be
assigned~\cite{Bonnet2017Complexity,DBLP:journals/comgeo/AbellanasBHORT06}.
Label sets express exactly the third form, which contains the other two.
To our knowledge, arbitrary compatibility relations have not been studied
before for moves along paths of unoccupied vertices.

In our problem, every token moves at most once; in the robotics
literature, a sequence of moves with this property is called
\emph{monotone}~\cite{StilmanKuffner2008}. Three results on monotone sequences of
moves relate directly to ours. First, when both the
matching and the routes are given as part of the input, feasibility is
decided in polynomial time by an acyclicity test on a precedence
digraph~\cite{DBLP:conf/icra/Buckley89,%
DBLP:journals/comgeo/AbellanasBHORT06}; we state this as
Observation~\ref{obs:precedence-acyclicity} and use it in
Section~\ref{sec:positive}. Second, when the matching is given but the
routes are not, the problem is hard. Geft, Halperin, and
Nakar~\cite{GeftHalperinNakar2024} study \textsc{Monotone Sliding
Reconfiguration}, whose discrete form is exactly our unique-matching case,
in which every source is compatible with exactly one target and every
target with exactly one source. They prove it \textsf{NP}-complete on
$3\times w$ grid graphs, and \textsf{NP}-hard even when every token but one
has a specified route. Section~\ref{sec:planar} shows that the hardness
persists on instances that, when solvable, admit a solution with pairwise
edge-disjoint routes. Third, in
work independent of and concurrent with ours, Geft~\cite{Geft2026Trees}
proves that it is \textsf{NP}-hard to decide whether a \textsc{Pebble
Motion} instance with two colors on a tree of maximum degree three can be
solved with as many moves as the sum of the distances from the sources to
the targets. In our terms, a color is a label set of size one, so such an
instance uses two distinct labels. Geft further proves that a solution
attaining this bound can be reordered so that each token traverses the
path from its source to its target in consecutive moves, during which no
other token occupies that path; such a
traversal is exactly one move of our model. Hence,
\textsc{Single-Move Labeled Token Routing} is \textsf{NP}-complete on trees
of maximum degree three even when only two distinct labels are used.

\section{Our contributions}
\label{sec:our_contributions}

We study \emph{feasibility} throughout: deciding whether the target
configuration can be reached at all. The statements below use the notation
of Section~\ref{sec:prelim}: an instance $I=(G,S,T,L)$ consists of a graph
$G$ with $n$ vertices and $m$ edges, a set $S$ of sources, a set $T$ of
targets, and a labeling $L$, which assigns a set of labels to each source
and to each target, and $|I|$ denotes the length of its encoding.

We begin with the previously studied all-matching case, in which every source is
compatible with every target. An all-matching instance is solvable if and
only if every component of the graph contains as many sources as
targets~\cite{Calinescu2008Reconfigurations}. We provide a direct proof in
our terminology; the proof simultaneously constructs a matching, a route for
each matched pair, and an order of the moves by repeatedly removing leaves
from a spanning tree of the graph.
\begin{restatable}{theorem}{allMatching}
\label{thm:allMatching}
\textsc{Single-Move Labeled Token Routing} on all-matching instances can be
solved, and a solution constructed, in time $\mathcal{O}(n^2+m)$, and every connected all-matching instance
is a yes-instance.
\end{restatable}

At the other extreme, in the unique-matching case, every source is
compatible with exactly one target and every target with exactly one source.
This case is \textsf{NP}-complete on $3\times w$ grid
graphs~\cite{GeftHalperinNakar2024}, which have pathwidth at most three; in the instances produced by their reduction, every solution has an edge that lies on at least one third of the routes (Section~\ref{sec:planar}).
Our first two results give a construction whose instances, whenever they are
solvable, admit a solution with pairwise edge-disjoint routes; the
\emph{solution edge multiplicity} of a solution is the largest number of its
routes that use a common edge, so such a solution has solution edge
multiplicity one. Neither bounding the pathwidth nor
bounding the solution edge multiplicity of the sought solution therefore
makes the problem tractable. Our construction also lowers the maximum degree from four to
three, at the price of allowing some edges to be shared by two routes of a
solution.

\begin{restatable}{theorem}{uniqueMatchingPlanarHardness}
\label{thm:unique_matching_planar_hardness}
\textsc{Single-Move Labeled Token Routing} on unique-matching instances is
\textsf{NP}-complete on planar graphs of maximum degree four, even on
instances that, when solvable, admit a solution with pairwise edge-disjoint
routes, and on planar graphs of maximum degree three, even on instances
that, when solvable, admit a solution with solution edge multiplicity at
most two.
\end{restatable}

Grids and grid graphs are the graph classes closest to atom reconfiguration:
one- and two-dimensional target configurations place target traps at
regularly spaced positions, and moves between neighboring traps induce a
grid-like graph structure. The same reduction transfers to them with minor
modifications.

\begin{restatable}{theorem}{gridAndSubgridHardness}
\label{thm:grid_and_subgrid_hardness}
\textsc{Single-Move Labeled Token Routing} on unique-matching instances is
\textsf{NP}-complete on grids and on grid graphs, even on instances that,
when solvable, admit a solution with pairwise edge-disjoint routes.
\end{restatable}

On trees, unique-matching instances are solvable in polynomial time: the
route of every pair is forced, and whether those routes can be ordered is
decided by the acyclicity test described in Section~\ref{sec:related}. Trees are therefore where
the label constraints are worth examining closely. By the independent and
concurrent result discussed in Section~\ref{sec:related}, the problem is
\textsf{NP}-complete on trees of maximum degree three even with two
distinct labels~\cite{Geft2026Trees}. We study the parameterized complexity
of the problem on trees, taking as parameter an upper bound $\mu$, given as
part of the input, on the solution edge multiplicity of the sought solution.
With this parameter, the problem is not fixed-parameter tractable unless
$\textsf{FPT}=\textsf{W}[1]$, even when a movement order is given as part of
the input.

\begin{restatable}{theorem}{solutionMultiplicityHardness}
\label{thm:sol_multiplicity_hardness}
On trees of maximum degree three, \textsc{Single-Move Labeled Token Routing}
is \textsf{NP}-complete, and, for a nonnegative integer $\mu$ given as part
of the input, deciding whether there is a solution with solution edge
multiplicity at most $\mu$ is \textsf{W}[1]-hard when parameterized by
$\mu$; both results hold even when a movement order is given as part of
the input.
\end{restatable}

The solution edge multiplicity counts only the routes of a solution.
The \emph{candidate edge multiplicity} is determined by the instance alone:
for each edge, one counts the compatible source-target pairs whose unique
paths in the tree use that edge, and the candidate edge multiplicity of the
instance is the maximum of these counts over all edges. Our last hardness
result shows that bounding the candidate edge multiplicity does not make the
problem tractable either: the problem remains \textsf{NP}-complete when the
candidate edge multiplicity is bounded by a constant, so unless
$\textsf{P}=\textsf{NP}$, it admits no algorithm whose running time is
polynomial in the input size for each fixed value of this parameter.

\begin{restatable}{theorem}{candidateMultiplicityHardness}
\label{thm:para_np_tree}
\textsc{Single-Move Labeled Token Routing} is \textsf{NP}-complete on trees
with candidate edge multiplicity at most eight, even when every source is
compatible with at most three targets and every target with at most three
sources.
\end{restatable}

Next, we turn to algorithmic results. The first concerns paths, where the linear
order of the vertices forces the matching: the $i$th source must be matched
to the $i$th target, and it remains to test the compatibility of these
pairs.

\begin{restatable}{theorem}{pathTime}
\label{thm:path_time}
\textsc{Single-Move Labeled Token Routing} on a path can be solved in time
$\mathcal{O}(|I|)$.
\end{restatable}

Our remaining results are parameterized. Since the hardness results above
involve choosing a matching together with a movement order, it is natural to
parameterize by the number $k$ of tokens, which bounds the search space for both
choices: there are at most $k!$ matchings and at most $k!$ movement orders per
matching. For
each guessed matching and movement order, the remaining question of whether
suitable routes exist can be answered efficiently.

\begin{restatable}{theorem}{tokenNumberFpt}
\label{thm:token_number_fpt}
\textsc{Single-Move Labeled Token Routing} is fixed-parameter tractable when
parameterized by the number $k$ of tokens.
\end{restatable}

Finally, we return to trees for our main algorithmic results.
Theorem~\ref{thm:sol_multiplicity_hardness} shows that bounding the solution
edge multiplicity of the sought solution, even together with the maximum
degree, does not make the
problem tractable, and Theorem~\ref{thm:para_np_tree} shows the same for the
candidate edge multiplicity alone. Bounding the candidate edge multiplicity
together with the maximum degree, however, does. Each edge of a tree
separates the tree into two sides, and the pairs counted for that edge are
exactly the compatible pairs with a source on one side and a target on the
other. When the maximum degree is also bounded, only a bounded number of
compatible pairs cross the edges incident with any one vertex, so the
problem can be solved by dynamic programming over the tree.

\begin{restatable}{theorem}{treeCandidateMultiplicityFpt}
\label{thm:tree_candidate_multiplicity_fpt}
\textsc{Single-Move Labeled Token Routing} on a tree instance $I=(G,S,T,L)$ can
be solved in time $2^{\mathcal{O}(\Delta\cdot \cm(I)^2)}\cdot |I|^{\mathcal{O}(1)}$,
where $\Delta$ is the maximum degree of $G$ and $\cm(I)$ is the candidate edge
multiplicity of $I$. Consequently, the problem is fixed-parameter tractable when
parameterized by $\Delta+\cm(I)$.
\end{restatable}

Neither parameter of Theorem~\ref{thm:tree_candidate_multiplicity_fpt} can be
omitted: by Theorem~\ref{thm:para_np_tree} the problem is
\textsf{NP}-complete on trees whose candidate edge multiplicity is constant,
so parameterizing by the candidate edge multiplicity alone does not yield
fixed-parameter tractability unless $\textsf{P}=\textsf{NP}$, and by
Theorem~\ref{thm:sol_multiplicity_hardness} the problem is
\textsf{NP}-complete on trees of maximum degree three, so the same holds for
the maximum degree alone. Conversely, by
Theorem~\ref{thm:tree_candidate_multiplicity_fpt}, no \textsf{NP}-hardness
reduction to instances on trees can produce only instances of constant maximum
degree and constant candidate edge multiplicity, unless
$\textsf{P}=\textsf{NP}$.

The construction of Theorem~\ref{thm:para_np_tree} has a single vertex of
unbounded degree, and that vertex lies on the paths of unboundedly many
compatible pairs. Counting the compatible pairs at vertices instead of edges
gives the \emph{candidate vertex multiplicity}, the maximum over all vertices
of the number of compatible source-target pairs whose unique paths contain
that vertex, and the same dynamic programming algorithm shows that the
problem is fixed-parameter tractable when parameterized by this quantity
alone.

\begin{restatable}{theorem}{treeCandidateVertexMultiplicityFpt}
\label{thm:tree_cvm_fpt}
\textsc{Single-Move Labeled Token Routing} on a tree instance $I=(G,S,T,L)$
can be solved in time $2^{\mathcal{O}(\cvm(I)^2)}\cdot|I|^{\mathcal{O}(1)}$,
where $\cvm(I)$ is the candidate vertex multiplicity of $I$. Consequently,
the problem is fixed-parameter tractable when parameterized by $\cvm(I)$.
\end{restatable}

\section{Preliminaries}
\label{sec:prelim}

\subsection{Graph terminology}
We use standard graph-theoretic terminology; see, e.g., Diestel~\cite{Diestel}.
All graphs are finite, simple, and undirected unless stated otherwise. For a
graph $G$ we write $V(G)$ and $E(G)$ for its vertex set and edge set, and we set
$n=|V(G)|$ and $m=|E(G)|$; for a path $P$ we write $V(P)$ and $E(P)$. For a
positive integer $i$ we write $[i]=\{1,\ldots,i\}$, and for an instance
$I$ of a decision problem we write $|I|$ for the length of its encoding.

A \emph{grid} is a graph whose vertex set is the Cartesian product of two
finite nonempty sets of consecutive integers, two vertices being adjacent if
and only if they are at Euclidean distance one. Its \emph{width} $w$ and
\emph{height} $h$ are the sizes of the first and second of these sets, and we
call it an \emph{$h\times w$ grid}. A \emph{grid graph} is a subgraph of a
grid, and an \emph{$h\times w$ grid graph} is a subgraph of an $h\times w$
grid. The \emph{treewidth} of a graph measures how far the graph is from being a
tree, and its \emph{pathwidth} how far it is from being a path; we refer the reader to
Diestel~\cite{Diestel} for the formal definitions.

Let $G$ be a planar graph. A \emph{plane drawing} of $G$ consists of an
injective map assigning to each vertex of $G$ a point of $\mathbb{R}^2$ and, to
each edge $uv\in E(G)$, a simple curve from the point representing $u$ to the
point representing $v$, such that the interior of each edge curve is disjoint
from the points representing vertices and from the interiors of all other edge
curves. An \emph{orthogonal planar representation} of $G$ is a plane drawing in
which every edge is represented by a simple polygonal path all of whose segments
are horizontal or vertical. An orthogonal planar representation is an
\emph{orthogonal grid drawing} if the points representing vertices and the
endpoints of all segments representing edges have integer coordinates. Every
planar graph of maximum degree at most four admits an orthogonal grid drawing
of polynomial size that can be computed in polynomial
time~\cite{Tamassia1987,TamassiaTollisVitter2000}.

\subsection{Instances, moves, and solutions}
Let $G$ be a graph and let $S,T\subseteq V(G)$ with $|S|=|T|$. The vertices of
$S$ are \emph{source vertices}, the vertices of $T$ are \emph{target vertices},
and the vertices of $V(G)\setminus(S\cup T)$ are \emph{ordinary vertices}. At
any point in time, a vertex is \emph{occupied} if it contains a token and
\emph{unoccupied} otherwise. Initially each source holds one token, so
precisely the vertices of $S$ are occupied. In figures, source vertices are
drawn as filled discs, target vertices as hollow squares, ordinary vertices as
hollow discs, and vertices that are both sources and targets as filled discs
enclosed in squares.

Let $\Lambda$ be a finite set of \emph{labels}. A \emph{labeling} is a pair
$L=(L_S,L_T)$ of functions $L_S\colon S\to 2^{\Lambda}$ and
$L_T\colon T\to 2^{\Lambda}$; the sets $L_S(s)$ and $L_T(t)$ are the \emph{label
sets} of the source $s$ and of the target $t$.
A source $s\in S$ and a target $t\in T$ are \emph{$L$-compatible} if
$L_S(s)\cap L_T(t)\neq\emptyset$; when $L$ is clear from context we simply say
\emph{compatible}. The set of compatible pairs $(s,t)\in S\times T$ is the
\emph{compatibility relation} of $L$. A label $\lambda\in\Lambda$ is \emph{private} if it belongs to the
label set of exactly one source and of exactly one target. An
\emph{$L$-compatible matching} is a set $M\subseteq S\times T$ in which every
source of $S$ and every target of $T$ appears in exactly one pair, and every
pair of which is $L$-compatible.

For a pair $(s,t)\in S\times T$, a \emph{route} from $s$ to $t$ (an \emph{$s$-$t$
route}) is a simple $s$-$t$ path in $G$, possibly \emph{trivial}, that is,
consisting of the single vertex $s=t$; the word
\emph{route} designates the path along which a token is moved in a
solution. A \emph{route
family} for a compatible matching $M$ is a family
$\Pi=\{P_{s,t}:(s,t)\in M\}$ in which each $P_{s,t}$ is an $s$-$t$ route; for a
pair $q=(s,t)$ of $M$ we also write $P_q$ for $P_{s,t}$, $\src(q)=s$, and
$\tgt(q)=t$. A \emph{movement
order} for $M$ is an ordering of the pairs of $M$.

At a given point in time, a route $P_{s,t}$ is \emph{clear} if no vertex of
$V(P_{s,t})\setminus\{s\}$ is occupied at that time; in particular, the
route of a pair with $s=t$ is always clear. A \emph{move} takes the
token occupying $s$ to $t$ along a clear $s$-$t$ route; after the move, $t$ is
occupied and, if $s\neq t$, $s$ is unoccupied. In
this paper every token is moved at most once, whence the name of the problem.
Let $\pi$ be a movement order for $M$ and let $\Pi$ be a route family for $M$.
The move of a pair $q=(s,t)$ of $M$ is \emph{feasible} if $P_q$ is clear once
all pairs preceding $q$ in $\pi$ have been moved, at which point the occupied
vertices other than $s$ are precisely the targets of the pairs before $q$ and
the sources of the pairs after $q$. The order $\pi$ is \emph{feasible for
$\Pi$} if the move of every pair of $M$ is feasible. A \emph{solution} consists of a compatible matching
$M$, a route family $\Pi$ for $M$, and a movement order feasible for $\Pi$.

\begin{tcolorbox}[colback=white, colframe=black, arc=2mm, boxrule=0.3mm,
  left=2mm, right=2mm, top=1mm, bottom=1mm]
\textsc{Single-Move Labeled Token Routing}\\[0.5ex]
\textbf{Input:} A graph $G$, sets $S,T\subseteq V(G)$ with $|S|=|T|$, one token on each vertex of $S$, a finite
set $\Lambda$ of labels, and a labeling $L=(L_S,L_T)$ with
$L_S\colon S\to 2^{\Lambda}$ and $L_T\colon T\to 2^{\Lambda}$.\\[0.5ex]
\textbf{Question:} Do there exist a compatible matching
$M\subseteq S\times T$, a route family $\Pi=\{P_{s,t}:(s,t)\in M\}$ for $M$, and
a movement order feasible for $\Pi$?
\end{tcolorbox}

We write an instance as $I=(G,S,T,L)$, leaving the label set $\Lambda$
implicit, and call it a \emph{tree instance} if $G$ is a tree.

The problem belongs to \textsf{NP}: a compatible matching, a route family for
it, and a movement order form a certificate of size polynomial in $|I|$, and
simulating the moves in that order, while maintaining the set of occupied
vertices, checks in polynomial time that every move is feasible. For each
\textsf{NP}-completeness result below we therefore prove only hardness.

\subsection{Special cases and precedence digraphs}
An instance is \emph{all-matching} if every source is compatible with every
target, and \emph{unique-matching} if each source is compatible with exactly
one target and each target with exactly one source.

Let $M$ be a compatible matching and let $\Pi$ be a route family for $M$.
The \emph{precedence digraph} $D(M,\Pi)$ has vertex set $M$. For
distinct pairs $p=(s_p,t_p)$ and $q=(s_q,t_q)$ of $M$, an arc $(p,q)$
means that the move of $p$ must precede the move of $q$, and is added
in two cases. If $s_p\in V(P_q)$, then the token at $s_p$ must leave
before the route $P_q$ is used. If $t_q\in V(P_p)$, then the token
moving to $t_q$ may arrive only after the route $P_p$ has been used.

The characterization below is the acyclicity test referenced in
Section~\ref{sec:related}; we prove it in our notation because the dynamic
programming algorithm of Section~\ref{sec:positive} relies on it.

\begin{observation}\label{obs:precedence-acyclicity}
Let $M$ be a compatible matching and let $\Pi$ be a route family for $M$. There
exists a feasible movement order for $\Pi$ if and only if $D(M,\Pi)$ is acyclic.
\end{observation}

\begin{proof}
Suppose that $\pi$ is a feasible movement order for $\Pi$, and let $(p,q)$ be
an arc of $D(M,\Pi)$. If $s_p\in V(P_q)$, then $s_p$ is a vertex of $P_q$
other than $s_q$, occupied until $p$ is moved, so $p$ precedes $q$. If
$t_q\in V(P_p)$, then either $t_q=s_p$, and the previous case applies, or
$t_q$ is a vertex of $P_p$ other than $s_p$, occupied once $q$ is moved, so
again $p$ precedes $q$. Hence, every arc goes from an earlier to a later pair
of $\pi$, and $D(M,\Pi)$ has no directed cycle.

Conversely, suppose that $D(M,\Pi)$ is acyclic and let $\pi=(q_1,\ldots,q_k)$
be an ordering of the pairs of $M$ in which every arc of $D(M,\Pi)$ goes from
an earlier pair to a later pair; such a topological ordering exists
because $D(M,\Pi)$ is acyclic. Consider a pair $q_i$ of this ordering
and write $s_i$ for its source. If the route of $q_i$ contains the source of
some pair $q_j$ with $j\neq i$ other than at $s_i$, then $(q_j,q_i)$ is an arc of the
precedence digraph, so $j<i$ and that source has already been vacated. If the
route of $q_i$ contains the target of some pair $q_j$ with $j\neq i$ other than at $s_i$,
then $(q_i,q_j)$ is an arc, so $i<j$ and that target is not yet occupied.
Hence, when the token at $s_i$ is moved, every vertex of the route of $q_i$
other than $s_i$ is unoccupied: the route is clear. Since this holds for
every $i$, the movement order $\pi$ is feasible.
\end{proof}

\subsection{Parameterized complexity and multiplicity}
We use standard parameterized complexity
terminology~\cite{cygan2015parameterized}. A parameterized problem is
\emph{fixed-parameter tractable} if there is a computable function $f$ such
that every instance $I$ with parameter $\kappa$ can be solved in time
$f(\kappa)\cdot|I|^{\mathcal{O}(1)}$; the class of all fixed-parameter
tractable problems is denoted $\textsf{FPT}$. A \emph{parameterized
reduction} from a parameterized problem $Q$ to a parameterized problem $Q'$
is an algorithm that maps every instance $(I,\kappa)$ of $Q$ to an instance
$(I',\kappa')$ of $Q'$ that is a yes-instance if and only if $(I,\kappa)$
is, and that does so in time $f(\kappa)\cdot|I|^{\mathcal{O}(1)}$ and with
$\kappa'\le f(\kappa)$, for some computable function $f$. A parameterized
problem is \emph{$\textsf{W}[1]$-hard} if every problem of the class
$\textsf{W}[1]$ admits a parameterized reduction to it; unless
$\textsf{FPT}=\textsf{W}[1]$, no $\textsf{W}[1]$-hard problem is
fixed-parameter tractable~\cite{cygan2015parameterized}.

The first of our three multiplicity measures is defined for route families in
arbitrary graphs. For a route family $\Pi$, the \emph{solution edge multiplicity}
of $\Pi$ is
$\solc(\Pi)=\max_{e\in E(G)}|\{P\in\Pi:e\in E(P)\}|$. The solution edge
multiplicity of a solution is that of its route family. In
Theorem~\ref{thm:sol_multiplicity_hardness} we parameterize by an upper
bound on this quantity: the input contains, in addition to an instance $I$,
a nonnegative integer $\mu$, the question is whether $I$ admits a solution
with solution edge multiplicity at most $\mu$, and the parameter is $\mu$.

The second measure is defined only for tree instances $I=(G,S,T,L)$, in which
every pair $(s,t)\in S\times T$ is joined by a unique path $P_{s,t}$. For an
edge $e\in E(G)$, the \emph{candidate edge multiplicity} of $e$ is
$\cm_I(e)=|\{(s,t)\in S\times T: s\text{ and }t\text{ are compatible and }e\in
E(P_{s,t})\}|$, and the \emph{candidate edge multiplicity} of $I$ is
$\cm(I)=\max_{e\in E(G)}\cm_I(e)$. The third measure counts the same pairs at
vertices instead of edges. For a vertex $v\in V(G)$, the \emph{candidate
vertex multiplicity} of $v$ is $\cvm_I(v)=|\{(s,t)\in S\times T: s\text{ and
}t\text{ are compatible and }v\in V(P_{s,t})\}|$, and the \emph{candidate
vertex multiplicity} of $I$ is $\cvm(I)=\max_{v\in V(G)}\cvm_I(v)$. The name
\emph{multiplicity} is borrowed from the literature on coordinating the movement of robots
along paths given as part of the input, where the \emph{vertex multiplicity}
of an instance is the maximum number of these paths through a single
vertex~\cite{DBLP:conf/atal/AbrahamsenGHU23}.

For every route family $\Pi$ of a solution of a tree instance $I$ we have
$\solc(\Pi)\le \cm(I)$, since every route of a solution is the unique path of some
compatible pair. The ratio between the two can nevertheless be linear in the
number $k$ of tokens. Consider an all-matching instance on a star whose leaves
are $k$ sources and $k$ targets, with $S\cap T=\emptyset$. Every source is
compatible with every target, and the path from a source leaf to any target
leaf uses the edge incident with that source leaf. That edge therefore lies on
the paths of $k$ compatible pairs, and $\cm(I)=k$. In every solution each source
is matched with exactly one target, so every edge lies on exactly one route
of the solution, and every solution has solution edge multiplicity one.

\subsection{The all-matching case}
We close the preliminaries with the all-matching case, the starting point for
our study of compatibility constraints. When the matching is unrestricted, the
matching, the routes, and the movement order can be chosen simultaneously, by a
procedure that repeatedly deletes a leaf of a spanning tree. The two nontrivial
cases are a leaf that is a target, which is filled from a closest remaining
source and whose move is placed \emph{before} all moves produced later, and a
leaf that is a source, which is emptied into a closest remaining target and
whose move is placed \emph{after} all moves produced later; in each case the
minimality of the choice is what makes the route clear at the time it is
used. The result below is due to C\u{a}linescu, Dumitrescu, and
Pach~\cite{Calinescu2008Reconfigurations}, and a stronger form, in which the
matching is additionally required to minimize the total distance traveled
by the tokens, is also known~\cite{ElSabeh2023ARO}.

\allMatching*

\begin{proof}
Let $I$ be an all-matching instance on a graph $G$ with source set $S$ and
target set $T$, where $|S|=|T|$.

If $G$ is disconnected, then every route lies in a single connected component,
so every component must contain the same number of sources as targets; if this
fails, the instance is a no-instance, and otherwise it suffices to solve the
instance independently in each component and concatenate the resulting movement
orders. We may therefore assume that $G$ is connected.

Compute an arbitrary spanning tree $F$ of $G$. Since every route in $F$ is also
a route in $G$, it is enough to construct a solution whose routes use only
edges of $F$. We describe a recursive procedure on the current tree $F$ together
with the current source set $S$ and target set $T$, maintaining the invariant
that $F$ is connected and $|S|=|T|$. Let $v$ be a leaf of $F$.

\begin{enumerate}
\item \textbf{Case $v\notin S\cup T$.} The leaf $v$ is irrelevant, as no route
between two remaining vertices uses it. Delete $v$ from $F$ and continue.

\item \textbf{Case $v\in S\cap T$.} Match $v$ to itself and use the trivial
route. Since $v$ is a leaf, it is not an internal vertex of any route between
the remaining vertices, and its move may be placed anywhere in the movement
order, since its route contains no vertex other than $v$ itself. Delete $v$
from $F$, remove it from $S$ and from $T$, and continue.

\item \textbf{Case $v\in T\setminus S$.} Choose a source $s\in S$ closest to $v$
in $F$, match $s$ to $v$, and let $P$ be the unique $s$-$v$ path in $F$. Place
this move before the moves obtained recursively; at the time it is executed,
the occupied vertices are precisely the current sources. By the choice of $s$,
no internal vertex of $P$ belongs to $S$, and $v\notin S$, so $s$ is the only
occupied vertex of $P$ and $P$ is clear. Delete $v$ from $F$, remove $s$ from
$S$ and $v$ from $T$, and continue.

\item \textbf{Case $v\in S\setminus T$.} Choose a target $t\in T$ closest to $v$
in $F$, match $v$ to $t$, and let $P$ be the unique $v$-$t$ path in $F$. By the
choice of $t$, no internal vertex of $P$ belongs to $T$. Place this move after
the moves obtained recursively, delete $v$ from $F$, remove $v$ from $S$ and $t$
from $T$, and continue. After the recursive moves have been executed, the
occupied vertices of the remaining tree are precisely the remaining targets; in
particular $t$ is still unoccupied and every internal vertex of $P$ is
unoccupied, so $P$ is clear at that time and the token at $v$ can be moved
to $t$.
\end{enumerate}

Each reduction removes at least one vertex. The case analysis shows how to
extend a feasible movement order of the reduced instance to one of the original
instance: in the case $v\in T\setminus S$ one prepends the move into $v$, and in
the case $v\in S\setminus T$ one appends the move out of $v$. By induction on
$|V(F)|$, the original instance admits a feasible movement order.

A spanning tree of $G$ can be computed in time $\mathcal{O}(n+m)$. The recursive
procedure performs at most $n$ leaf removals, and with a straightforward
implementation each step finds the closest remaining source or target by a graph
search in time $\mathcal{O}(n)$. Listing the routes takes total time
$\mathcal{O}(n^2)$, since at most $n$ moves are produced and each route has
length at most $n-1$. Hence, the whole solution is constructed in time
$\mathcal{O}(n^2+m)$.
\end{proof}

\section{Hardness under structural and compatibility constraints}
\label{sec:hardness}

We now prove the hardness results stated in Section~\ref{sec:our_contributions}.
Their purpose is to separate several possible sources of difficulty in
\textsc{Single-Move Labeled Token Routing}. Section~\ref{sec:planar} shows
that the problem is hard on planar graphs of bounded degree and on grids,
even when the compatible matching is unique and the instances, when
solvable, admit a solution with pairwise edge-disjoint routes. Hardness
therefore arises neither from the need to choose a matching nor from any
need for two routes to share an edge. We then turn to trees, where the
route of every compatible pair is forced. Section~\ref{sec:tree-hardness}
shows that on trees of maximum degree three, choosing a matching is
\textsf{W}[1]-hard when parameterized by an input bound on the solution edge
multiplicity, even when a movement order is given as part of the input, and
Section~\ref{sec:cm-hardness} shows that on trees of unbounded degree, the
problem is \textsf{NP}-complete even when the candidate edge multiplicity is
at most eight.

\subsection{Reducing the degree of ordinary vertices}
\label{sec:degree-reduction}
For ease of understanding, we first describe each construction using ordinary
vertices of high degree and then reduce the degree by the operation defined
next, which replaces such a vertex by a tree of ordinary vertices of maximum
degree three. The resulting instance is a yes-instance if and only if the
original one is.

Let $u$ be an ordinary vertex of a graph $G$ with $d\ge 4$ neighbors, and let
$F_u$ be a tree of maximum degree three with $d$ leaves and at most $2d$
vertices. The \emph{expansion} of $u$ \emph{by} $F_u$ replaces $u$ by $F_u$: it
deletes $u$, adds a disjoint copy of $F_u$ whose vertices are all ordinary, and
joins the $d$ leaves of the copy to the $d$ neighbors of $u$, one leaf to each
neighbor. The resulting graph is denoted $G-u+F_u$.

\begin{lemma}\label{lem:subcubic-expansion}
Let $(G,S,T,L)$ be an instance of \textsc{Single-Move Labeled Token Routing},
let $u\in V(G)$ be an ordinary vertex of degree at least four, and let $G'=G-u+F_u$
be obtained by expanding $u$ by a tree $F_u$ of maximum degree three with as
many leaves as $u$ has neighbors. Then $(G,S,T,L)$ is a
yes-instance if and only if $(G',S,T,L)$ is a yes-instance. Moreover, $G'$ is
a tree if $G$ is a tree, and, if $G$ is planar, the leaves can be attached
to the neighbors of $u$ so that $G'$ is planar.
\end{lemma}

\begin{proof}
Every vertex of $F_u$ is ordinary, and is therefore permanently unoccupied.
Consider a solution of $(G,S,T,L)$. Since $u$ is ordinary, every route
that uses $u$ contains a subpath of the form $v,u,w$, where $v$ and $w$ are
neighbors of $u$ in $G$. Replace this subpath by the path from $v$ to $w$ that
enters $F_u$ through the leaf adjacent to $v$, follows the unique path in $F_u$
to the leaf adjacent to $w$, and then proceeds to $w$. This transforms
every route into a
route of $G'$. Since the newly introduced vertices are permanently unoccupied,
every route that was clear in $G$ remains clear in $G'$. Hence, the same
matching and movement order, together with the transformed routes, form a
solution of $(G',S,T,L)$.

Conversely, consider a solution of $(G',S,T,L)$, with route family $\Pi'$. Contract
$F_u$ to a single vertex $u$. Every route in $\Pi'$ then becomes a walk in
$G$. Since the original route is simple, the resulting walk can repeat only
$u$. For each such walk, retain the subwalk from its source to its first
occurrence of $u$ and the subwalk from its last occurrence of $u$ to its
target, identifying these two occurrences as a single occurrence of $u$.
The resulting walk repeats no vertex and is therefore a simple route with
the same endpoints. Its
internal vertices, apart from the permanently unoccupied vertex $u$, form a
subset of the internal vertices of the original route. Therefore, every route
that was clear in $G'$ yields a clear route in $G$. Hence, the same matching and
movement order, together with the resulting routes, form a solution of
$(G,S,T,L)$.

If $G$ is a tree, then $G'$ is obtained from the forest $G-u$ by attaching
each of its components to a distinct leaf of $F_u$, so $G'$ is
a tree.

For planarity, fix a plane drawing of $G$ and a closed disc around $u$ so
small that it contains no vertex other than $u$, meets no edge that is not
incident with $u$, and meets each edge incident with $u$ in a single curve from
$u$ to the boundary of the disc. Let $p_1,\ldots,p_d$ be the points at which
these edges cross the boundary of the disc, in cyclic order, and let
$v_1,\ldots,v_d$ be the corresponding neighbors of $u$. Draw $F_u$ inside the
disc without crossings, with its $d$ leaves at the points $p_1,\ldots,p_d$, and
join the leaf at $p_i$ to $v_i$ by the part of the edge $uv_i$ outside the disc,
for each $i\in[d]$. This yields a plane drawing of $G'$
(Figure~\ref{fig:high_degree}).
\end{proof}

\begin{figure}
    \centering
    \includegraphics[width=0.8\linewidth]{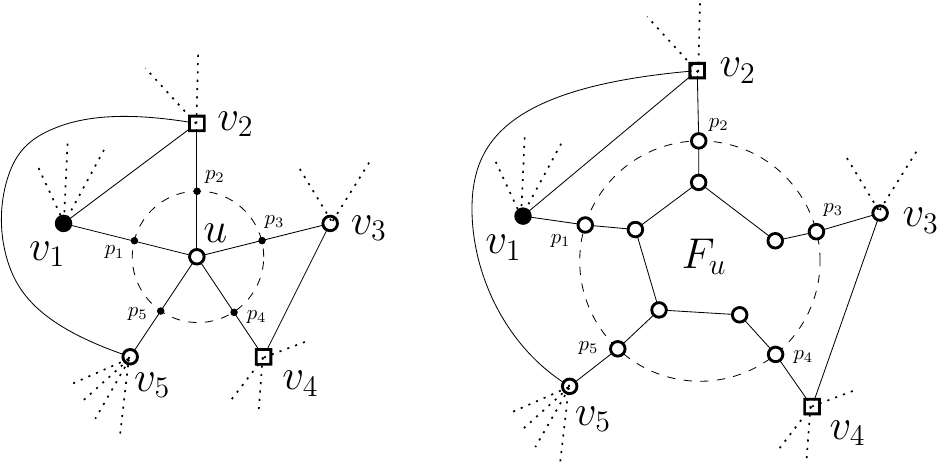}
    \caption{The expansion of an ordinary vertex $u$ of degree at
    least four. On the left, $u$ has five neighbors $v_1,\ldots,v_5$, and the
    edge $uv_i$ crosses the boundary of the dashed disc around $u$ at the
    point $p_i$. On the right, the part of the drawing inside the disc is
    replaced by a tree $F_u$ of maximum degree three, whose leaves lie at the
    points $p_1,\ldots,p_5$, and the leaf at $p_i$ is joined to $v_i$ by the
    part of $uv_i$ outside the disc; for readability, the disc is enlarged on
    the right and the rest of the drawing is stretched accordingly. None of the
    vertices of $F_u$ is a source or a target. Dotted edge segments indicate
    that the graph continues beyond the neighbors of $u$: the rest of the
    graph, which the expansion leaves unchanged, is not drawn.}
    \label{fig:high_degree}
\end{figure}

\subsection{Unique matching with edge-disjoint routes on planar graphs and grids}
\label{sec:planar}

On unique-matching instances the matching is fixed, and Geft, Halperin, and
Nakar~\cite{GeftHalperinNakar2024} show that choosing the routes is hard
on $3\times w$ grid graphs, even when every token but one has a prescribed
route. When the routes are not prescribed, every solution of their instances
has an edge that lies on at least $k/3$ of the routes, where $k$ is the
number of tokens: all
sources lie to the left of all targets, so every route crosses any vertical
line between them, and only three edges cross such a line, so one of them lies
on at least a third of the routes. We show that hardness does not require
routes that share edges: every instance constructed in this subsection that
is a yes-instance admits a solution whose routes are pairwise edge-disjoint,
that is, a solution with solution edge multiplicity one.

\begin{figure}
    \centering
    \includegraphics[width=0.5\linewidth]{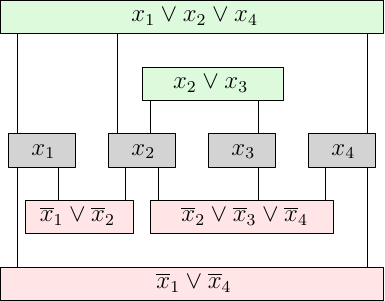}
    \caption{A rectilinear representation of the \textsc{Planar Monotone Rectilinear 3-SAT} instance $\varphi=(x_1\lor x_2\lor x_4)\land(x_2\lor x_3)\land(\overline{x}_1\lor\overline{x}_2)\land(\overline{x}_2\lor\overline{x}_3\lor\overline{x}_4)\land(\overline{x}_1\lor\overline{x}_4)$. The positive clauses are drawn above the variable line, and the negative clauses below it. A clause box is joined by a vertical segment to each of its variables, so one clause may reach several variables.}
    \label{fig:3sat}
\end{figure}

This subsection proves Theorems~\ref{thm:unique_matching_planar_hardness}
and~\ref{thm:grid_and_subgrid_hardness}, both by modifying a single
reduction that we first present on general graphs. On general graphs the
reduction yields a weaker statement than the known hardness on $3\times w$
grid graphs, but its construction and correctness proof are the basis of
both modifications, so we give them in full.

\subsubsection{The reduction on general graphs}
\label{sec:planar-general}

We reduce from \textsc{Planar Monotone Rectilinear 3-SAT}; we first define
its instances.
A clause of a CNF formula is \emph{positive} if none of its literals are
negated variables and \emph{negative} if all of them are.
A formula is \emph{monotone} if each of its clauses is positive or negative.
A \emph{rectilinear representation} of a monotone formula $\varphi$ consists
of an axis-parallel rectangle for each variable and for each clause of
$\varphi$, called the \emph{box} of that variable or clause, and of a
vertical segment for each clause $C$ and each variable occurring in $C$,
which joins the box of $C$ to the box of that variable, such that the boxes
are pairwise disjoint, the boxes of the variables all meet one horizontal
line, called the \emph{variable line}, the box of every positive clause
lies above that line and the box of every negative clause below it, and
every vertical segment meets the two boxes it joins only at its endpoints
and meets no other box and no other vertical segment
(Figure~\ref{fig:3sat}). We identify each variable and each clause with its
box.

An instance of \textsc{Planar Monotone Rectilinear 3-SAT} is a monotone CNF
formula whose clauses contain at most three literals, together with a
rectilinear representation of the formula; the question is whether the formula is
satisfiable. This problem is \textsf{NP}-complete~\cite{deBergKhosravi2012}.

Let $(\varphi,\mathcal{R})$ be such an instance, so that $\varphi$ is a
monotone CNF formula whose clauses contain at most three literals and
$\mathcal{R}$ is a rectilinear representation of $\varphi$. Throughout the
construction, $x$ denotes a variable of $\varphi$ and $C$ a clause of
$\varphi$. We assume that every clause contains two or three literals, that no two
literals of a clause have the same variable, and that every variable occurs
in some clause. All three assumptions can be made without loss of
generality. A clause with a single literal can only be satisfied by setting its
variable accordingly. Doing so, removing every clause that becomes satisfied,
deleting the literal that becomes false from every remaining clause, and repeating
this while some clause has a single literal yields, in polynomial time, a
formula that has no such clause and is satisfiable if and only if $\varphi$ is,
unless a clause becomes empty, in which case $\varphi$ is unsatisfiable; its
rectilinear representation is obtained from $\mathcal{R}$ by deleting the
boxes and segments of the removed clauses and literals. In a
monotone formula a variable occurring twice in a clause occurs twice as the
same literal, so removing the repeated literal leaves an equivalent formula with the
same rectilinear representation. A variable that occurs in no clause can be
deleted together with its box. An \emph{occurrence} is a pair consisting of a
clause and one of the literals of that clause, called the \emph{literal of the
occurrence}; the \emph{variable of the occurrence} is the variable of its
literal, and we say that the occurrence is an occurrence \emph{of} that
variable and \emph{in} that clause. Occurrences in distinct clauses are
distinct even when their literals are equal, so the number of occurrences is
the sum of the sizes of the clauses. We identify a clause $C$ with the set of
its occurrences and write $o\in C$. An occurrence $o$ is \emph{positive} if
its literal is its variable and \emph{negative} if its literal is the
negation of its variable, and $o$ is \emph{true} or \emph{false} under a
truth assignment according to the value of its literal.

To represent $\varphi$ as a graph, we create tokens and gadgets for each variable and each clause. A variable gadget contains two paths for the \emph{assignment token}, one corresponding to setting the variable to \texttt{true} (the \emph{true path} for the variable), and the other to setting the variable to \texttt{false} (the \emph{false path} for the variable). In a clause gadget, there are at most three possible paths for the \emph{clause token}, one corresponding to each occurrence in the clause.

Additional tokens tie the paths available to a clause token to the paths
taken by the assignment tokens. For each occurrence, a \emph{blocking
token} has its source in the variable gadget, on the true path for a
negative occurrence and on the false path for a positive occurrence, and
its target on the path of the clause gadget associated with that
occurrence. When the assignment token moves along the true path or the
false path, every blocking token whose source lies on that path must move
out of its way, into the clause gadget of its occurrence, where it blocks
the path of that occurrence. For each occurrence, a \emph{delay token} has
its source in the clause gadget, on every path from the source of the
clause token to its target, and its target in the variable gadget, on the
same path as the source of the blocking token. The clause token therefore
moves only after the assignment tokens that pass the targets of its delay
tokens, by which time the blocking tokens on the paths of those
assignment tokens have reached the clause gadget.

\smallskip
\noindent\textbf{Construction.}
For a token $z$, we write $\src(z)$ and $\tgt(z)$ for its source and its
target. Throughout this subsection the letters $a$, $b$, $c$, and $d$, with
subscripts, denote tokens.

We make use of the following types of tokens, each named by the initial of
its role: an \emph{assignment token} $a_x$ for each variable $x$, a
\emph{blocking token} $b_o$ and a \emph{delay token} $d_o$ for each
occurrence $o$, and a \emph{clause token} $c_C$ for each clause $C$. When
the occurrences of one clause are numbered $o_1,o_2,\ldots$, we write $b_i$
and $d_i$ for the blocking token and the delay token of $o_i$.

Figure~\ref{fig:unique_m_hardness} shows the part of the construction arising
from one clause with three occurrences and from the three variables occurring in it. We describe the variable gadgets, the clause gadgets, and the occurrence edges in turn.

\begin{figure}
    \centering
    \includegraphics[width=0.6\linewidth]{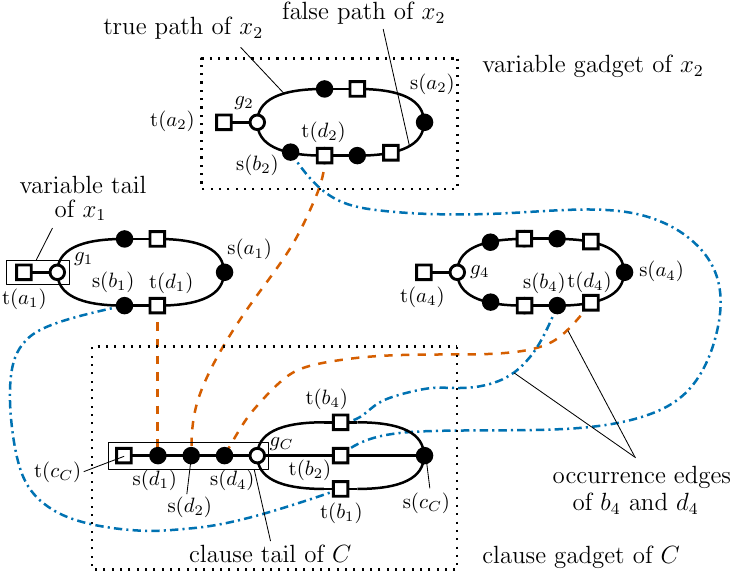}
    \caption{Part of the construction for the positive clause $C=x_1\lor x_2\lor
x_4$ of $\varphi$ (Figure~\ref{fig:3sat}); for $i\in\{1,2,4\}$, $a_i$ and
$g_i$ denote the assignment token and the junction of the gadget of $x_i$,
and $b_i$ and $d_i$ the blocking and delay tokens of the occurrence of $x_i$
in $C$. In every variable gadget, the true path is drawn above and the
false path below. The dotted rectangles single out the gadget of $x_2$ and
the gadget of $C$. The blue dash-dotted curves and the orange dashed curves
are the occurrence edges of the blocking tokens and of the delay tokens,
respectively. The occurrence edges of other clauses are omitted, but the
sources and targets that those clauses create on the true and false paths
are shown.}
    \label{fig:unique_m_hardness}
\end{figure}
\begin{description}
\item[Variable gadgets.] For each variable $x$, we create the vertices
$\src(a_x)$ and $\tgt(a_x)$, an ordinary vertex $g_x$, called the
\emph{junction} of the gadget, and, for each occurrence $o$ of $x$, the
vertices $\src(b_o)$ and $\tgt(d_o)$. The \emph{true path} is a path from
$g_x$ to $\src(a_x)$ whose internal vertices are $\src(b_o)$ followed by
$\tgt(d_o)$ for each negative occurrence $o$ of $x$, the pairs of distinct
occurrences appearing in an arbitrary order; the \emph{false path} is a
path from $g_x$ to $\src(a_x)$ defined in the same way from the positive
occurrences of $x$; the \emph{variable tail} is the edge $g_x\tgt(a_x)$.
Sources and targets thus alternate along the true path and along the
false path.

\item[Clause gadgets.] For each clause $C$, we create the vertices
$\src(c_C)$ and $\tgt(c_C)$, an ordinary vertex $g_C$, called the
\emph{junction} of the gadget, and, for each occurrence $o\in C$, the
vertices $\tgt(b_o)$ and $\src(d_o)$. The \emph{literal path} of $o$ is
the path $\src(c_C),\tgt(b_o),g_C$, and the \emph{clause tail} is a path
from $g_C$ to $\tgt(c_C)$ whose internal vertices are the vertices
$\src(d_o)$ with $o\in C$, in an arbitrary order.

\item[Occurrence edges.] For each occurrence $o$, we add the edge
$\src(b_o)\tgt(b_o)$ and the edge $\src(d_o)\tgt(d_o)$, the
\emph{occurrence edges} of $b_o$ and of $d_o$.

\item[Labels.] For each token, we assign a private label to its source and
its target, so that the compatible matching of the constructed instance is
unique.
\end{description}

\smallskip
In this reduction, the variable gadgets and clause gadgets are pairwise vertex-disjoint, and each occurrence edge meets the two gadgets it joins only at its endpoints. Let $J=(G,S,T,L)$ denote the constructed instance of \textsc{Single-Move Labeled Token Routing}. Each variable, clause, and occurrence contributes a constant number of vertices and edges, so $J$ can be computed in polynomial time.
We next present the properties of the instance $J$ that will be used in
the correctness proof. The first observation shows that the occurrence edge of a blocking or
delay token is used by no other token; since these edges are
the only edges joining two gadgets, the assignment tokens and the clause
tokens are confined to their own gadgets, which leaves each of them the
few routes described in the second observation.

\begin{observation}\label{obs:occurrence-edge-isolation}
Let $o$ be an occurrence and let $z$ be $b_o$ or $d_o$. In every solution
of $J$, no route of a token other than $z$ contains both $\src(z)$ and
$\tgt(z)$; in particular, no such route uses the occurrence edge of $z$.
\end{observation}

\begin{proof}
Fix an occurrence $o$ and let $z$ be $b_o$ or $d_o$. At every point of a
movement order, the vertex $\src(z)$ is occupied if $z$ has not yet moved,
and the vertex $\tgt(z)$ is occupied otherwise, so at least one of the two
is occupied at every point. Neither vertex is the source or the target of
a token other than $z$, so neither is the first vertex of the route of
another token. A route of another token containing both therefore
contains, at the time it is used, an occupied vertex other than its own
first vertex, and is not clear. The occurrence edge of $z$ joins these
two vertices, so no route of another token uses it.
\end{proof}

\begin{observation}\label{obs:variable-clause-route-forms}
In every solution of $J$, the route of each assignment token $a_x$ consists of either the true path or the false path of the variable gadget for $x$, followed by its variable tail. Similarly, the route of each clause token $c_C$ consists of one literal path of the clause gadget for $C$, followed by its clause tail.
\end{observation}

\begin{proof}
By Observation~\ref{obs:occurrence-edge-isolation}, neither an assignment token nor a clause token can use an occurrence edge. Since the gadgets are otherwise vertex-disjoint, the route of an assignment token remains inside its variable gadget, and the route of a clause token remains inside its clause gadget. The structure of the routes follows from the structure of the two gadgets.
\end{proof}

Confining the assignment tokens and the clause tokens to their gadgets also
settles the solution edge multiplicity of $J$.

\begin{observation}\label{obs:edge-disjoint-routes}
In every solution of $J$, the routes of the tokens are pairwise
edge-disjoint. Consequently, every solution of $J$ has solution edge
multiplicity one.
\end{observation}

\begin{proof}
By Observation~\ref{obs:variable-clause-route-forms}, the route of an assignment
token uses only edges of its variable gadget, and the route of a clause token
only edges of its clause gadget; the gadgets are pairwise vertex-disjoint, so
these routes are pairwise edge-disjoint. By
Observation~\ref{obs:occurrence-edge-isolation}, the occurrence edge of a
blocking or delay token $z$ is used by no token other than $z$. The token $z$
cannot use any other edge, since every route from $\src(z)$ to $\tgt(z)$ other
than its occurrence edge would have to leave the gadget containing $\src(z)$
through the occurrence edge of another token. Hence, every edge belongs to at
most one route.
\end{proof}

The remaining property concerns the order in which the tokens move. Consider
a solution of $J$ with route family $\Pi$. By
Observation~\ref{obs:variable-clause-route-forms}, every assignment token
$a_x$ uses either the true path or the false path of its gadget, and we let
$\tau_\Pi$ be the truth assignment that sets $x$ to \texttt{true} if $a_x$
uses its true path and to \texttt{false} if $a_x$ uses its false path. For
the reverse direction of the reduction, we will show that $\tau_\Pi$
satisfies $\varphi$; the observation below gives the precedence constraints
from which this follows.

\begin{observation}\label{obs:forced-literal-order}
Let $o$ be an occurrence in a clause $C$. In every solution of $J$ with
route family $\Pi$, if $o$ is false under $\tau_\Pi$, then $\tgt(b_o)$ is
occupied when $c_C$ moves.
\end{observation}

\begin{proof}
Let $x$ be the variable of $o$. By construction, $\src(b_o)$ and $\tgt(d_o)$
lie on the false path of the variable gadget for $x$ if $o$ is positive and
on its true path if $o$ is negative, so both lie on the route of $a_x$
precisely when $o$ is false under $\tau_\Pi$. Suppose that $o$ is false. The
source $\src(b_o)$ must be vacated before $a_x$ moves, so $b_o$ moves before
$a_x$. Moving $d_o$ before $a_x$ would occupy $\tgt(d_o)$, an internal
vertex of the route of $a_x$, so $a_x$ moves before $d_o$. By
Observation~\ref{obs:variable-clause-route-forms}, the route of $c_C$
contains the clause tail, and $\src(d_o)$ lies on this tail, so $d_o$ moves
before $c_C$. Therefore, $b_o$ moves before $c_C$, and since each token moves
at most once, $\tgt(b_o)$ is occupied when $c_C$ moves.
\end{proof}

We now have all the ingredients needed to establish the correctness of the reduction.

\begin{lemma}
\label{lem:unique_matching_hardness}
\textsc{Single-Move Labeled Token Routing} on unique-matching instances is
\textsf{NP}-complete, even on instances in which the routes of every solution
are pairwise edge-disjoint.
\end{lemma}

\begin{proof}
We show that $\varphi$ is satisfiable if and only if $J$ is a yes-instance.

\smallskip
\noindent
($\Rightarrow$) Suppose that $\varphi$ has a satisfying assignment $\tau$. First, for every
occurrence $o$ that is false under $\tau$, move the blocking token $b_o$
along its occurrence edge to $\tgt(b_o)$. This edge has no internal vertex,
and $\tgt(b_o)$ is unoccupied, since no token has yet reached a target, so
every such move is feasible.

Next, for every variable $x$, move $a_x$ along its true path if
$\tau(x)=\texttt{true}$ and along its false path if
$\tau(x)=\texttt{false}$, followed by its variable tail. The only sources on
the chosen path are those of the blocking tokens of the occurrences of $x$
that are false under $\tau$, and all of these sources were vacated in the
first step. No source of another assignment token, of a clause token, or of
a delay token lies on the chosen path, since the gadgets are pairwise
vertex-disjoint and the sources of clause and delay tokens lie in clause
gadgets. Moreover, no delay token has yet moved, so every target of a delay
token on the chosen path is unoccupied, and so is $\tgt(a_x)$. Hence, the
route of every assignment token is clear.

We next move every delay token along its occurrence edge, which has no
internal vertex. Its target is unoccupied, since no other token is assigned
to it, and since every assignment token has already reached its target,
occupying it does not block any route that remains to be used.

Then, consider a clause $C$. Since $\tau$ satisfies $C$, there exists an
occurrence $o\in C$ that is true under $\tau$. Only the blocking tokens of
false occurrences moved in the first step, so $b_o$ has not moved, and
$\tgt(b_o)$ is unoccupied. Every delay token has moved, so every source of a
delay token on the clause tail of $C$ has been vacated, and $\tgt(c_C)$ is
unoccupied. Thus, $c_C$ can move along the literal path of $o$ and then
along the clause tail to $\tgt(c_C)$.

Finally, move every remaining blocking token along its occurrence edge,
which has no internal vertex; its target is unoccupied, since no other token
is assigned to it. Every token has moved along a clear route, so the unique
compatible matching of $J$, the routes along which the tokens moved, and
this movement order form a solution of $J$.

\smallskip
\noindent
($\Leftarrow$)
Suppose that $J$ is a yes-instance, and let $\Pi$ be the route family of a
solution. We show that $\tau_\Pi$, the truth assignment that sets $x$ to
\texttt{true} if $a_x$ uses its true path and to \texttt{false} if $a_x$
uses its false path, satisfies $\varphi$. Consider a clause $C$
and suppose, for a contradiction, that every occurrence in $C$ is false under
$\tau_\Pi$. By Observation~\ref{obs:forced-literal-order}, the target of the
blocking token of each occurrence in $C$ is occupied when $c_C$ moves. By
Observation~\ref{obs:variable-clause-route-forms}, the route of $c_C$ uses a
literal path of the clause gadget for $C$, and each literal path contains
the target of the blocking token of its occurrence. Hence, whichever literal
path the route of $c_C$ uses, the route contains a vertex other than its
source that is occupied when $c_C$ moves, so no route available to $c_C$ is
clear at that time. This contradicts the feasibility of the movement order,
so $\tau_\Pi$ satisfies $C$; hence $\tau_\Pi$ satisfies every clause of
$\varphi$.

By Observation~\ref{obs:edge-disjoint-routes}, the routes of every solution
of $J$ are pairwise edge-disjoint. Together with membership in \textsf{NP},
established in Section~\ref{sec:prelim}, this proves the lemma.
\end{proof}

\subsubsection{A planar realization}
\label{sec:planar-realization}

The graph of $J$ need not be planar. We now modify $J$, in polynomial time,
into an instance $J_{\mathrm{pl}}=(G_{\mathrm{pl}},S,T,L)$ whose graph is
planar. The modification changes the clause gadgets and fixes the order of
the vertices along the true and false paths of the variable gadgets;
Figure~\ref{fig:planar_realization} shows the resulting graph for the
formula of Figure~\ref{fig:3sat}, drawn as described in the proof of
Lemma~\ref{lem:planar} below.

Both the new clause gadgets and the order of the vertices along the paths of
the variable gadgets are determined by the rectilinear representation
$\mathcal{R}$ given with $\varphi$. The \emph{segment} of an occurrence $o$
is the vertical segment of $\mathcal{R}$ that joins the box of the clause of
$o$ to the box of the variable of $o$. Its endpoint on the box of the clause
is the \emph{clause port} of $o$, which lies on the side of that box that
faces the variable line, and its endpoint on the box of the variable is the
\emph{variable port} of $o$, which lies on the upper side of that box if $o$
is positive and on its lower side if $o$ is negative. Distinct occurrences
have distinct segments, since no two literals of a clause have the same
variable, and the segments of distinct occurrences are disjoint. Hence, the
clause ports of a clause are distinct points of one horizontal side of its
box, and so are the variable ports of the positive occurrences of a
variable, and those of its negative occurrences. We order the ports on a
side from left to right, that is, by increasing $x$-coordinate. We number
the occurrences of each clause $C$ as $o_1,\ldots,o_{|C|}$ in the order of
their clause ports, writing $b_i$ and $d_i$ for the tokens of $o_i$ as in
Section~\ref{sec:planar-general}, and we order the positive occurrences of
each variable, and separately its negative occurrences, in the order of
their variable ports.

The correctness of $J$ rests on two properties
(Observations~\ref{obs:variable-clause-route-forms}
and~\ref{obs:forced-literal-order}): the route of an assignment token is
its true path or its false path, followed by its variable tail, and the
route of a clause token $c_C$ contains $\tgt(b_o)$ and $\src(d_o)$ for some
occurrence $o\in C$, which forces $b_o$ to block the route of $c_C$ when $o$
is false. The modification preserves both properties. Its purpose is to
allow a plane drawing of the resulting graph, which we now describe.

We call a vertex of a gadget an \emph{anchor} if an occurrence edge ends at
it. In the drawing, each gadget lies inside its box, and the two occurrence
edges of each occurrence $o$ run alongside the segment of $o$, that of $b_o$
to its left and that of $d_o$ to its right, crossing the side of each of the
two boxes immediately to the left and to the right of the port of $o$ on
that side. The occurrence edges entering a box therefore arrive at each of
its sides in the order of the ports on that side, the edge of $b_o$
immediately followed by that of $d_o$. Inside the box, each of these edges
is completed by a straight line segment to its anchor. These segments cross
neither each other nor the gadget if the anchors of the edges arriving at a
side lie on a line parallel to that side, in the order in which the edges
arrive, and no point of the gadget lies strictly between this line and that
side: the segments then join points of two parallel lines in the same
order, and the gadget has no point between these lines.

\begin{figure}[htp]
\centering
\includegraphics[width=\linewidth]{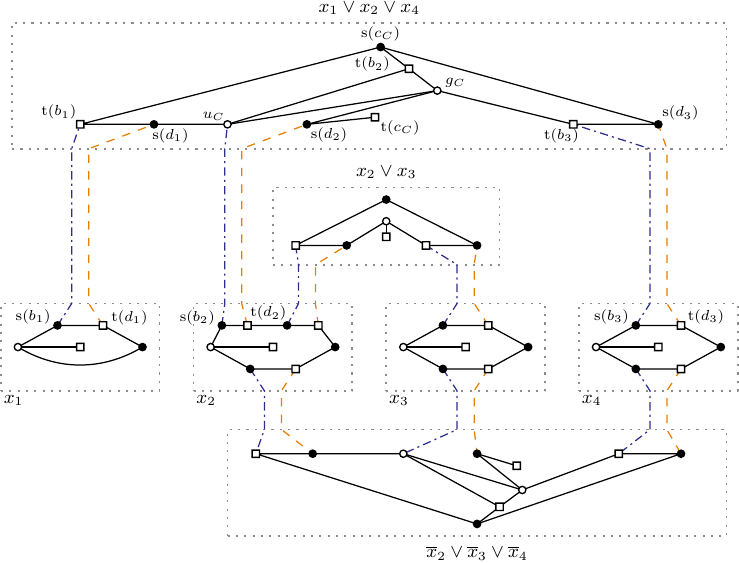}
\caption{The graph $G_{\mathrm{pl}}$ for the formula of Figure~\ref{fig:3sat},
drawn as in the proof of Lemma~\ref{lem:planar} and restricted to the four
variable gadgets and to the clause gadgets of $(x_1\lor x_2\lor x_4)$,
$(x_2\lor x_3)$, and $(\overline{x}_2\lor\overline{x}_3\lor\overline{x}_4)$;
the gadgets of the other two clauses, their occurrence edges, and the
vertices that their occurrences add to the variable gadgets are omitted.
The dotted rectangles are the boxes of $\mathcal{R}$, and the dash-dotted
blue curves and dashed orange curves are the occurrence edges of the
blocking tokens and of the delay tokens, as in
Figure~\ref{fig:unique_m_hardness}. In the box of the clause
$C=x_1\lor x_2\lor x_4$, whose occurrences $o_1$, $o_2$, and $o_3$ are those
in $x_1$, $x_2$, and $x_4$, the vertices are labeled with the tokens
indexed by occurrence, and so are the sources $\src(b_i)$ and targets
$\tgt(d_i)$ of its occurrence edges in the boxes of $x_1$, $x_2$, and $x_4$; $u_C$ is the crossing vertex of $C$, at which the occurrence edge
of $b_2$ ends. In each variable gadget, the false path, which carries the
vertices of the positive occurrences, is drawn above the variable line and
the true path below it, whereas Figure~\ref{fig:unique_m_hardness} draws
the true path above.}
\label{fig:planar_realization}
\end{figure}

In a variable gadget, the anchors on the upper side of the box are the
vertices $\src(b_o)$ and $\tgt(d_o)$ of the positive occurrences $o$, all on
the false path, and those on the lower side are the corresponding vertices
of the negative occurrences, all on the true path. The variable gadget can
therefore be drawn as required once these vertices appear along the false
path in the order of the positive occurrences, and along the true path in
the order of the negative ones; this is how the modification fixes the
order of a variable gadget's vertices.

In the gadget of a clause $C$, the anchors must lie on the line in the order
$\tgt(b_1)$, $\src(d_1)$, \ldots, $\tgt(b_{|C|})$, $\src(d_{|C|})$. The
clause gadget of
$J$ cannot be drawn in this way, since every delay source lies on the clause
tail and every blocking target on a literal path. The modification therefore
places $\src(d_o)$ next to $\tgt(b_o)$ on the literal path of $o$ for the
first and the last occurrence $o$ of each clause. For a clause with three
literals, one of the three literal paths lies inside the cycle formed by the
other two in every plane drawing of the gadget. In a drawing in which the
anchors lie on the line as required, this inner path is the literal path of
the second occurrence: the
blocking targets of the first and the last occurrence lie on the line, and
no point of the gadget lies between the line and the side of the box, so
these targets lie on the outer face. The occurrence edge of the blocking
token of the second occurrence, whose source lies outside the box,
therefore has to cross one of the other two literal paths. The construction
makes the crossing point an ordinary vertex of the gadget, which becomes the
anchor of this occurrence edge, and it keeps the delay source of the second
occurrence on the clause tail, which every route of the clause token uses.
The proof of Lemma~\ref{lem:planar} gives the drawings of both gadgets. The
construction is as follows.

\begin{description}
\item[Variable gadgets.] The variable gadget of $x$ is that of $J$, with
the order of the vertices along its true and false paths fixed as follows.
Along the false path, from $g_x$ to $\src(a_x)$, the vertices $\src(b_o)$
and $\tgt(d_o)$ of the positive occurrences $o$ of $x$ appear in the order
of these occurrences, with
$\src(b_o)$ immediately followed by $\tgt(d_o)$ for each $o$, as in $J$.
The order along the true path is defined symmetrically, with the negative
occurrences of $x$ in place of the positive ones.

\item[Clause gadgets.] The clause gadget of $C$ has the vertices
$\src(c_C)$, $\tgt(c_C)$, and $g_C$ and, for each $o\in C$, the vertices
$\tgt(b_o)$ and $\src(d_o)$, as in $J$, but its paths differ. If $|C|=2$,
the literal path of $o_1$ is $\src(c_C),\tgt(b_1),\src(d_1),g_C$,
the literal path of $o_2$ is $\src(c_C),\src(d_2),\tgt(b_2),g_C$, and the
clause tail is the edge $g_C\tgt(c_C)$; the gadget consists of the
two literal paths and the clause tail. If $|C|=3$, the gadget has one
further ordinary vertex $u_C$, the \emph{crossing vertex} of $C$
(Figure~\ref{fig:planar_realization}). The
literal path of $o_1$ is $\src(c_C),\tgt(b_1),\src(d_1),u_C,g_C$, the
literal path of $o_2$ is $\src(c_C),\tgt(b_2),g_C$, the literal path of
$o_3$ is $\src(c_C),\src(d_3),\tgt(b_3),g_C$, and the clause tail is
$g_C,\src(d_2),\tgt(c_C)$; the gadget consists of the three literal paths,
the clause tail, and the edge $u_C\tgt(b_2)$. Thus, $u_C$ lies on the
literal path of $o_1$, between $\src(d_1)$ and $g_C$, and is adjacent to
$\tgt(b_2)$; it is the anchor of the occurrence edge of $b_2$, defined in
the next item.

\item[Occurrence edges.] For each occurrence $o$, the occurrence edge of
$d_o$ is $\src(d_o)\tgt(d_o)$, as in $J$. The occurrence edge of $b_o$ is
$\src(b_o)\tgt(b_o)$, unless $o$ is the second occurrence of a clause $C$
with three literals, in which case it is $\src(b_o)u_C$. We also define,
for each blocking or delay token $z$, its \emph{occurrence path}, a path
from $\src(z)$ to $\tgt(z)$. For every token $z$ other than the blocking
token $b_2$ of a clause with three literals, the occurrence path of $z$ is
its occurrence edge alone. For the blocking token $b_2$ of a clause $C$
with three literals, whose occurrence edge ends at $u_C$ rather than at
$\tgt(b_2)$, the occurrence path is $\src(b_2),u_C,\tgt(b_2)$, the
occurrence edge followed by the edge $u_C\tgt(b_2)$ of the clause gadget.

\item[Labels.] The labels are those of $J$.
\end{description}

This completes the construction of $J_{\mathrm{pl}}$. We first describe its basic properties.

\begin{lemma}\label{lem:planar-construction-properties}
The instance $J_{\mathrm{pl}}$ can be computed in polynomial time, and its
compatible matching is unique. The gadgets of $G_{\mathrm{pl}}$ are pairwise
vertex-disjoint, and every edge of $G_{\mathrm{pl}}$ with endpoints in two
gadgets is an occurrence edge, which joins the variable gadget and the
clause gadget of its occurrence. The maximum degree of $G_{\mathrm{pl}}$ is
four, and its vertices of degree four are the crossing vertices and the
junctions of the clauses with three literals, all of which are ordinary.
\end{lemma}

\begin{proof}
Compared with $J$, each clause with three literals contributes one extra
vertex, its crossing vertex, and a constant number of extra edges, and the
orders along the paths are obtained by sorting the ports of $\mathcal{R}$
from left to right, so $J_{\mathrm{pl}}$ can be computed in polynomial time.
The labels are those of $J$, so the compatible matching remains unique.
The gadgets of $J$ are pairwise vertex-disjoint, and each crossing vertex
is added to one clause gadget, so the gadgets of $G_{\mathrm{pl}}$ are
pairwise vertex-disjoint as well. Every edge of $G_{\mathrm{pl}}$ is an
edge of a gadget or an occurrence edge, and the occurrence edge of $d_o$
joins $\src(d_o)$, in the clause gadget of $o$, to $\tgt(d_o)$, in the
variable gadget of $o$, while that of $b_o$ joins $\src(b_o)$, in the
variable gadget of $o$, to $\tgt(b_o)$ or to the crossing vertex, in the
clause gadget of $o$.

For the degrees, let $C$ be a clause with three literals. The crossing
vertex $u_C$ is adjacent to $\src(d_1)$, $g_C$, and $\tgt(b_2)$ in the
gadget and to $\src(b_2)$ by an occurrence edge, and $g_C$ is adjacent to
the last internal vertex of each of the three literal paths and to
$\src(d_2)$; both have degree four, and both are ordinary. Each vertex of
$G_{\mathrm{pl}}$ other than these has degree at most three. The vertices
$\src(c_C)$, $\tgt(c_C)$, $\src(a_x)$, $\tgt(a_x)$, $g_x$, and the junction
of a clause with two literals are incident with no occurrence edge and with
at most three edges of their gadgets. The vertex $\tgt(b_2)$ of a clause
$C$ with three literals is incident with no occurrence edge and with
exactly three edges of its gadget. Each remaining source or target of a
blocking or delay token is incident with two edges of its gadget and one
occurrence edge.
\end{proof}

We now show that $G_{\mathrm{pl}}$ is planar by giving the drawing
described before the construction.

\begin{lemma}\label{lem:planar}
The graph $G_{\mathrm{pl}}$ is planar.
\end{lemma}

\begin{proof}
We draw each gadget inside its box so that, for each side of
the box on which ports lie, the anchors of the occurrence edges arriving
at that side lie on a horizontal line segment parallel to that side, in
the left-to-right order of the ports of these edges, the anchor of the
edge of $b_o$ immediately followed by that of the edge of $d_o$, and that
no point of the gadget lies strictly between the line through this segment
and that side of the box. We first show that such
drawings of the gadgets yield a plane drawing of $G_{\mathrm{pl}}$, and
then that they exist.

Outside the boxes, the segments of distinct occurrences are disjoint, and
the segment of $o$ meets only the two boxes of $o$, at the ports of $o$.
We draw the occurrence edges of $b_o$ and $d_o$ between the two boxes as
vertical line segments immediately to the left and to the right of the
segment of $o$, so close to it that line segments of distinct occurrences
are disjoint and each ends on the box sides next to the ports of $o$.
Along each side, these endpoints thus occur in the order of the ports,
that of $b_o$ immediately before that of $d_o$, which is the order of the
anchors on the horizontal line segment facing that side. We complete each
occurrence edge by a straight line segment from its endpoint on the side
to its anchor. These joining line segments are pairwise disjoint, since
they join points of two parallel lines in the same left-to-right order,
and they meet the gadget only at the anchors, since every point of a
joining line segment other than its anchor lies strictly between the line
through the anchors and the side of the box, where the gadget has no point.
Together with the pairwise
disjointness of the boxes and of the vertical line segments, this makes
the drawing of $G_{\mathrm{pl}}$ plane.

It remains to draw the gadgets. In the variable gadget of $x$, the anchors
on the upper side are the vertices $\src(b_o)$ and $\tgt(d_o)$ of the
positive occurrences $o$ of $x$, which by construction appear along the
false path in the required order, and those on the lower side are the
corresponding vertices of the negative occurrences along the true path.
We draw the false path as a polygonal line above the variable line with
its internal vertices on a horizontal line segment, the true path
symmetrically below, and the variable tail along the variable line between
them; $g_x$ and $\src(a_x)$ lie on the variable line, so no point of the
gadget lies above the line through the first segment or below the line
through the second. In the
clause gadget of a positive clause $C$, the anchors are
$\tgt(b_1),\src(d_1),\tgt(b_2),\src(d_2)$ if $|C|=2$ and
$\tgt(b_1),\src(d_1),u_C,\src(d_2),\tgt(b_3),\src(d_3)$ if $|C|=3$, in
the required order, since the occurrence edge of $b_2$ ends at $u_C$ when
$|C|=3$. Figure~\ref{fig:planar_realization} shows a plane drawing of each
of the two gadgets with these vertices on a horizontal line segment in
this order and no point of the gadget below the line through this segment: $\src(c_C)$ is joined to
the two outermost anchors, $g_C$ lies between them and is joined to its
neighbors on the literal paths of $o_1$ and $o_{|C|}$, the clause tail
lies below $g_C$, and, for $|C|=3$, the literal path of $o_2$ runs from
$\src(c_C)$ through $\tgt(b_2)$ to $g_C$ with the edge $u_C\tgt(b_2)$
beside it. For a negative clause, we reflect this drawing across a
horizontal line.
\end{proof}

We next show that the two properties of $J$ stated above hold in
$J_{\mathrm{pl}}$.

\begin{lemma}\label{lem:planar-routes}
In every solution of $J_{\mathrm{pl}}$, the route of each assignment token
$a_x$ consists of either the true path or the false path of its variable
gadget, followed by its variable tail, and the route of each clause token
$c_C$ lies in the clause gadget of $C$, ends with the clause tail, and
contains $\tgt(b_o)$ and $\src(d_o)$ for some occurrence $o\in C$.
\end{lemma}

\begin{proof}
The proof of Observation~\ref{obs:occurrence-edge-isolation} uses only that
$\src(z)$ and $\tgt(z)$ are the source and the target of the token $z$
alone, so the observation holds in $J_{\mathrm{pl}}$: no route of a token
other than $z$ contains both $\src(z)$ and $\tgt(z)$. The occurrence edges
fall into two classes. An occurrence edge of the form $\src(z)\tgt(z)$ is
used by no route of an assignment token or of a clause token, by the
observation. The remaining occurrence edges are the edges $\src(b_2)u_C$
of the clauses $C$ with three literals, which the observation does not
exclude, since $u_C$ is not the target of $b_2$. Let $G_{\mathrm{pl}}'$ be
obtained from $G_{\mathrm{pl}}$ by deleting the occurrence edges of the
first class; every route of an assignment token or of a clause token is a
path in $G_{\mathrm{pl}}'$. By Lemma~\ref{lem:planar-construction-properties},
the only edges of $G_{\mathrm{pl}}'$ joining two gadgets are the edges
$\src(b_2)u_C$, and such an edge is the only edge of $G_{\mathrm{pl}}'$
with exactly one endpoint in the clause gadget of $C$, so it is a bridge
of $G_{\mathrm{pl}}'$, whose deletion separates the clause gadget of $C$
from the rest of its component. A route of an assignment token has both
endpoints in one variable gadget, and a route of a clause token has both
endpoints in one clause gadget, so in both cases the two endpoints lie on
the same side of each such bridge. A path that uses a bridge once ends on
the other side of it, and a path that uses it twice is not simple; hence
neither route uses an edge $\src(b_2)u_C$. Therefore, neither route uses an
edge joining two gadgets, and each lies inside the gadget of its token.

In the variable gadget of $x$, the only edge incident with $\tgt(a_x)$ is
the variable tail, so the route of $a_x$ ends with the tail. The rest of
the route is a simple path from $\src(a_x)$ to $g_x$ that avoids
$\tgt(a_x)$, that is, a simple path in the cycle formed by the true path
and the false path. The only such paths are the two arcs of this cycle
between $\src(a_x)$ and $g_x$, which are the true path and the false path.

In the clause gadget of $C$, the only edge incident with $\tgt(c_C)$ is
the last edge of the clause tail, and, if $|C|=3$, the only neighbors of
$\src(d_2)$ in the gadget are $g_C$ and $\tgt(c_C)$. Hence, the route of $c_C$ ends with
the clause tail. The second vertex of the route is a neighbor of
$\src(c_C)$, that is, $\tgt(b_1)$, $\src(d_{|C|})$, or, if $|C|=3$,
$\tgt(b_2)$. If it is $\tgt(b_1)$, then the third vertex is $\src(d_1)$,
the only other neighbor of $\tgt(b_1)$ in the gadget, so the route
contains $\tgt(b_1)$
and $\src(d_1)$. If it is $\src(d_{|C|})$, then the third vertex is
$\tgt(b_{|C|})$, the only other neighbor of $\src(d_{|C|})$ in the gadget, so the
route contains $\tgt(b_{|C|})$ and $\src(d_{|C|})$. If it is $\tgt(b_2)$, then
the route contains $\tgt(b_2)$ and, on the clause tail, $\src(d_2)$.
\end{proof}

As for $J$, a solution of $J_{\mathrm{pl}}$ with route family $\Pi$
determines the truth assignment $\tau_\Pi$ that sets $x$ to \texttt{true}
if $a_x$ uses its true path and to \texttt{false} if $a_x$ uses its false
path; by Lemma~\ref{lem:planar-routes}, $\tau_\Pi$ is well defined. In $J$,
every route of a clause token $c_C$ contains the delay source of every
occurrence in $C$, and this is what forces the blocking token of every false
occurrence in $C$ to move before $c_C$
(Observation~\ref{obs:forced-literal-order}). In $J_{\mathrm{pl}}$, the route
of $c_C$ contains only some of these delay sources, so the next lemma
establishes this precedence only for the occurrences whose delay source lies
on the route of $c_C$.

\begin{lemma}\label{lem:planar-precedence}
Let $o$ be an occurrence in a clause $C$, and let $x$ be the variable of
$o$. In every solution of $J_{\mathrm{pl}}$ with route family $\Pi$ in
which $o$ is false under $\tau_\Pi$ and the route of $c_C$ contains
$\src(d_o)$, the token $b_o$ moves before $a_x$, $a_x$ moves before $d_o$,
$d_o$ moves before $c_C$, and $\tgt(b_o)$ is occupied when $c_C$ moves.
\end{lemma}

\begin{proof}
By construction, $\src(b_o)$ and $\tgt(d_o)$ lie on the false path of the
variable gadget of $x$ if $o$ is positive and on its true path if $o$ is
negative. Since $o$ is false under $\tau_\Pi$, both vertices lie on the
route of $a_x$ by Lemma~\ref{lem:planar-routes}. The source $\src(b_o)$
must be vacated before $a_x$ moves, so $b_o$ moves before $a_x$. Moving
$d_o$ before $a_x$ would occupy $\tgt(d_o)$, an internal vertex of the
route of $a_x$, so $a_x$ moves before $d_o$. The source $\src(d_o)$ is an
internal vertex of the route of $c_C$, so $d_o$ moves before $c_C$. Hence
$b_o$ moves before $c_C$, and since each token moves at most once,
$\tgt(b_o)$ is occupied when $c_C$ moves.
\end{proof}

\uniqueMatchingPlanarHardness*

\begin{proof}
By Lemma~\ref{lem:planar-construction-properties}, $J_{\mathrm{pl}}$ is
computed in polynomial time, is a unique-matching instance, and has maximum
degree four, and by Lemma~\ref{lem:planar} its graph is planar. We first
show that $\varphi$ is satisfiable if and only if $J_{\mathrm{pl}}$ is a
yes-instance.

\smallskip
\noindent
($\Rightarrow$) Suppose that $\varphi$ has a satisfying assignment $\tau$.
We move the tokens in the five stages of the proof of
Lemma~\ref{lem:unique_matching_hardness}: first the blocking tokens of the
occurrences that are false under $\tau$, then the assignment tokens, the
delay tokens, the clause tokens, and finally the remaining blocking tokens.
Every blocking or delay token moves along its occurrence path, $a_x$ moves
along its true path if $\tau(x)=\texttt{true}$ and along its false path
otherwise, followed by its variable tail, and $c_C$ moves along the literal
path of an occurrence in $C$ that is true under $\tau$, followed by the
clause tail. The occurrence path of a blocking or delay token has no
internal vertex other than, possibly, a crossing vertex, which is ordinary,
and its target is neither a source nor the target of another token, so the
route is clear when the token moves. The internal vertices of the route of
$a_x$ are $g_x$ and the vertices $\src(b_o)$ and $\tgt(d_o)$ of the
occurrences $o$ of $x$ on the path along which $a_x$ moves, which are the
occurrences of $x$ that are false under $\tau$; their blocking tokens moved
in the first stage, their delay tokens have not yet moved, and $\tgt(a_x)$
is unoccupied, so the route is clear. The internal vertices of the route of
$c_C$ are $g_C$, possibly $u_C$, the vertex $\tgt(b_o)$ of the occurrence
$o$ whose literal path $c_C$ uses, and delay sources; $b_o$ has not moved
because $o$ is true, every delay token has moved, and $\tgt(c_C)$ is
unoccupied, so the route is clear. Hence, $J_{\mathrm{pl}}$ is a
yes-instance.

\smallskip
\noindent
($\Leftarrow$) Suppose that $J_{\mathrm{pl}}$ is a yes-instance, let $\Pi$
be the route family of a solution, and consider a clause $C$. By
Lemma~\ref{lem:planar-routes}, the route of $c_C$ contains $\tgt(b_o)$ and
$\src(d_o)$ for some occurrence $o\in C$. If $o$ were false under
$\tau_\Pi$, then by Lemma~\ref{lem:planar-precedence} $\tgt(b_o)$ would be
occupied when $c_C$ moves, and the route of $c_C$ would not be clear.
Hence, $o$ is true under $\tau_\Pi$, so $\tau_\Pi$ satisfies $C$. Since
this holds for every clause, $\tau_\Pi$ satisfies $\varphi$.

\smallskip
\noindent\textbf{Bounding the maximum degree.}
By Lemma~\ref{lem:planar-construction-properties}, the vertices of degree
four of $G_{\mathrm{pl}}$ are the crossing vertices and the junctions of
the clauses with three literals, and all of them are ordinary. We apply
Lemma~\ref{lem:subcubic-expansion} to each of them, expanding it by a tree
with four leaves and two internal vertices. The resulting instance is a
yes-instance if and only if $J_{\mathrm{pl}}$ is, its graph is planar and
has maximum degree three, and its compatible matching remains unique, since
only ordinary vertices are introduced.

\smallskip
\noindent\textbf{Bounding the solution edge multiplicity.}
Consider the solution constructed from a satisfying assignment in the
forward direction. The routes of the assignment tokens and of the
clause tokens lie in pairwise vertex-disjoint gadgets and use only edges of
true paths, false paths, variable tails, literal paths, and clause tails.
The occurrence path of a blocking or delay token consists of its occurrence
edge and, for the second occurrence of a clause $C$ with three literals,
the edge $u_C\tgt(b_2)$, which lies on no literal path and on no clause
tail. Hence, the routes of this solution are pairwise edge-disjoint, so
every constructed instance that is a yes-instance admits a solution with
solution edge multiplicity one.

For the instance of maximum degree three obtained by the degree reduction,
modify the solution above as follows: every route that passes through an
expanded vertex now passes through the tree that replaced it, entering it
through the leaf adjacent to the preceding vertex of the route and leaving
it through the leaf adjacent to the following one. The vertices of the tree
are ordinary, so these routes are clear whenever the original ones were,
and the result is a solution. At a crossing vertex $u_C$, the routes of
$b_2$ and of $c_C$ may both pass, entering and leaving the tree through four
distinct leaves, so they share at most the edge joining the two internal
vertices of the tree; at an expanded junction only the clause token passes.
Hence, every constructed instance of maximum degree three that is a
yes-instance admits a solution with solution edge multiplicity at most two.
Together with membership in \textsf{NP}, established in
Section~\ref{sec:prelim}, this proves the theorem.
\end{proof}

\subsubsection{Grids and grid graphs}
\label{sec:planar-grids}

To prove Theorem~\ref{thm:grid_and_subgrid_hardness}, we construct, from
the planar instance $J_{\mathrm{pl}}$ of
Section~\ref{sec:planar-realization}, an instance whose graph is an
induced subgraph of a grid. A \emph{grid point} is a point of the plane
with integer coordinates, so that the vertices of a grid are grid points
and the grid with vertex set the grid points of an axis-parallel
rectangle is the grid \emph{induced} by that rectangle; a \emph{grid line}
is a horizontal or vertical line through grid points. Since
$G_{\mathrm{pl}}$ is planar and has maximum degree four, it admits an
orthogonal grid drawing, in which every vertex is a grid point and every
edge is drawn as a path along grid lines. The grid points on such a
drawing induce a subdivision of $G_{\mathrm{pl}}$, since every bend and
every intermediate grid point on the drawing of an edge becomes a vertex.
We therefore first show that subdividing edges preserves feasibility.

\begin{observation}\label{obs:edge-subdivision}
Let $I=(G,S,T,L)$ be an instance of \textsc{Single-Move Labeled Token
Routing}, and let $G'$ be obtained from $G$ by subdividing any number of
edges, every subdivision vertex being an ordinary vertex of the instance
$I'=(G',S,T,L)$. Then $I$ is a yes-instance if and only if $I'$ is a
yes-instance.
\end{observation}

\begin{proof}
Suppose first that $I$ is a yes-instance, and let $M$, $\Pi$, and $\pi$ be a compatible matching, route family, and feasible movement order, respectively, that constitute a solution. Replace every edge of every route in $\Pi$ by the corresponding subdivided path in $G'$. Since every subdivision vertex is ordinary, it is permanently unoccupied. Hence, $\pi$ is feasible for the resulting route family, and $I'$ is a yes-instance.

Conversely, suppose that $I'$ is a yes-instance. Contract every subdivided path to the corresponding edge of $G$. This transforms every route into a simple path with the same endpoints. At the time a token moves, every internal vertex of the resulting route was also an internal vertex of its route in $G'$ and is therefore unoccupied. Hence, the resulting matching, route family, and movement order form a solution of $I$.
\end{proof}

\gridAndSubgridHardness*

\begin{proof}
For hardness on grid graphs, consider the planar instance
$J_{\mathrm{pl}}=(G_{\mathrm{pl}},S,T,L)$ of
Section~\ref{sec:planar-realization}. We use $G_{\mathrm{pl}}$ rather than
the graph of maximum degree three obtained by the degree reduction, since
we need a solution with pairwise edge-disjoint routes. Compute an
orthogonal grid drawing of $G_{\mathrm{pl}}$ of polynomial size, and
multiply all coordinates by two, so that two grid points of the drawing
are at distance one only if they are consecutive on the drawing of one
edge. Let $G_\square$ be the graph whose vertices are the grid points
lying on the drawing and whose edges join consecutive grid points along
the drawing of an edge of $G_{\mathrm{pl}}$.

By construction, $G_\square$ is a subdivision of $G_{\mathrm{pl}}$ of
polynomial size. It is also an induced subgraph of the grid: a grid edge
joining two vertices of $G_\square$ joins two grid points of the drawing
at distance one, which by the doubling are consecutive on the drawing of
one edge, so the grid edge is an edge of $G_\square$.

By Observation~\ref{obs:edge-subdivision}, $(G_{\mathrm{pl}},S,T,L)$ is a
yes-instance if and only if $(G_\square,S,T,L)$ is a yes-instance. Since
the labels are unchanged, every source in $G_\square$ remains compatible
with exactly one target. If $(G_\square,S,T,L)$ is a yes-instance, then so
is $(G_{\mathrm{pl}},S,T,L)$, which therefore admits a solution with
pairwise edge-disjoint routes; replacing every edge of each of its routes
by the path of $G_\square$ that subdivides it yields a solution of
$(G_\square,S,T,L)$ whose routes are again pairwise edge-disjoint, since
the paths that subdivide distinct edges are edge-disjoint. This proves
\textsf{NP}-hardness on grid graphs.

To obtain hardness on grids, let $H$ be the grid induced by the smallest axis-parallel rectangle containing the drawing. For every vertex $u\in V(H)\setminus V(G_\square)$, make $u$ both the source and the target of an additional token, and assign a fresh private label to $u$ as a source and as a target. The only simple route of this token is the one-vertex path at $u$, so $u$ remains occupied throughout every movement order.

Consequently, no token of the instance on $G_\square$ can use a vertex of $H\setminus V(G_\square)$ as an internal vertex of its route. Moreover, by the doubling, $H$ contains no edge with both endpoints in $V(G_\square)$ that does not already belong to $G_\square$. Therefore, adding the vertices and edges of $H\setminus G_\square$ creates no new route for any original token.

A solution for the instance on $G_\square$ extends to the padded instance
on $H$ by matching each added source with itself, using the one-vertex
route. Conversely, removing these pairs from a solution on $H$ yields a
solution on $G_\square$. The added tokens have private labels, and the
original sources retain their unique compatible targets. The added routes
use no edge, so pairwise edge-disjoint routes remain pairwise edge-disjoint.
Hence, \textsc{Single-Move Labeled Token Routing} on unique-matching
instances is \textsf{NP}-complete on grids, even on instances that, when
solvable, admit a solution with pairwise edge-disjoint routes.
\end{proof}

\subsection{Solution edge multiplicity on trees of maximum degree three}
\label{sec:tree-hardness}
We next turn to trees. On a tree, every compatible source-target pair is
connected by a unique path, so there is no choice of routes. Thus, any
hardness on trees must arise from the interaction between choosing a
compatible matching and ordering the resulting moves. The result of
Geft~\cite{Geft2026Trees} recalled in Section~\ref{sec:related} shows that
these choices alone suffice to make \textsc{Single-Move Labeled Token
Routing} \textsf{NP}-hard on instances whose tree has maximum degree three.

We further ask whether the complexity of such instances can be explained by
the amount of overlap among the routes of a solution. Our next result shows
that \textsc{Single-Move Labeled Token Routing}
is \textsf{W}[1]-hard when parameterized by an upper bound $\mu$, given as
part of the input, on the solution edge multiplicity, even
on trees of maximum degree three, and even when a movement order is given as
part of the input.

We prove the result by a parameterized reduction from \textsc{Grid Tiling}, in which we view an ordered pair $(i,j)\in[K]^2$ as a \emph{cell} of a $K\times K$ grid, to be filled with a \emph{tile}, that is, a pair of values in $[N]^2$. Formally, an instance of \textsc{Grid Tiling} consists of two positive integers $K$ and $N$ together with a set $\mathcal{T}_{i,j}\subseteq[N]^2$ for each cell $(i,j)\in[K]^2$; the tiles in $\mathcal{T}_{i,j}$ are the \emph{legal tiles for cell} $(i,j)$.

A solution is an assignment of a value in $[N]$ to each row and to each column such that the two values assigned to the row and the column of a cell form a legal tile for it, that is, if $a$ is assigned to row $i$ and $b$ is assigned to column $j$, then $(a,b)\in\mathcal{T}_{i,j}$. The problem is \textsf{NP}-hard, and \textsf{W}[1]-hard when parameterized by $K$~\cite{Marx2007GT,cygan2015parameterized}.

To encode the \textsc{Grid Tiling} instance, we construct an instance of
\textsc{Single-Move Labeled Token Routing} whose tree consists of a
\emph{tile subtree}, $2K$ \emph{choice subtrees}, one for each row and one
for each column, two \emph{storage paths}, and a \emph{cleanup path},
connected as described below and shown in Figure~\ref{fig:tree_w1_hard}. The choice subtrees are
used to choose a value for each row and each column, while the tile subtree is
used to verify that the chosen values specify a legal tile for each cell.
Each legal tile $(a,b)\in \mathcal{T}_{i,j}$ is associated with two
\emph{tile tokens}: a \emph{row tile token}, initially stored in the choice
subtree for row $i$, and a \emph{column tile token}, initially stored in the
choice subtree for column $j$. Choosing the tile $(a,b)$ for cell $(i,j)$
consists of moving its column tile token to the tile subtree, and its row
tile token into the target that the column tile token leaves free. The
choice subtree of a row or column contains one \emph{sealing token} for
each value, and moving it to a storage path, which contains targets for
the sealing tokens, chooses that value. The cleanup path contains the
sources of \emph{cleanup tokens}, used to fill targets in the choice
subtrees.

\smallskip
\noindent\textbf{Tile subtree.}
The tile subtree is a star with $K^2$ leaves, each a target vertex with a
distinct label set $\{\mathrm{tile}(i,j)\}$ for $(i,j)\in[K]^2$; the target
with label set $\{\mathrm{tile}(i,j)\}$ is the \emph{target for cell
$(i,j)$}. The sources whose label sets contain $\mathrm{tile}(i,j)$ are the
column tile token sources, introduced below.

\smallskip
\noindent\textbf{Choice subtrees.}
We construct $2K$ choice subtrees, one for each row and one for each column.
Each choice subtree contains $N$ \emph{value subtrees}, one for each of the
$N$ possible values. We denote the choice subtree for row $i$ by $\rowsub_i$
and its value subtree for value $a$ by $\rowsub_i(a)$. Similarly, we denote
the choice subtree for column $j$ by $\colsub_j$ and its value subtree for
value $b$ by $\colsub_j(b)$.

The construction must ensure that the tile tokens that leave a choice
subtree are those of a single value subtree, which then represents the value
chosen for that row or column. To that end, each value subtree is connected
to the rest of the tree through a single vertex, its \emph{attachment
vertex}, which is both a source and a target. At the attachment vertex of
$\rowsub_i(a)$, the label set of the source and the label set of the target
are both $\{\mathrm{seal}(\rowsub_i(a))\}$, and at the attachment vertex of
$\colsub_j(b)$, they are both $\{\mathrm{seal}(\colsub_j(b))\}$. The
token initially placed on an attachment vertex is the \emph{sealing token}
of its value subtree. Since the source and the target at an attachment
vertex share their label, a solution has two ways of matching a sealing
token. Leaving it at its own source keeps the attachment vertex occupied, so
that no tile token can leave the value subtree and each tile token must instead
move to a \emph{local target} inside the subtree; we then call the value
subtree \emph{sealed}. Matching the sealing token with a target outside the
choice subtree instead frees the attachment vertex; we call the value
subtree \emph{unsealed}. The storage paths described below supply only one
such outside target per choice subtree, so at most one value subtree per
choice subtree can be unsealed.

The remaining vertices of a value subtree form a path of \emph{tile token
sources}, one end of which is adjacent to the attachment vertex, together
with one \emph{local target} vertex for each tile token source, adjacent
to it as a leaf. In $\rowsub_i(a)$, there is one tile token source for each $j\in[K]$
and each $b\in[N]$ with $(a,b)\in \mathcal{T}_{i,j}$, and the token
initially placed on the source for $j$ and $b$ is the row tile token for
$(a,b)$. In $\colsub_j(b)$, there is one tile token source for each
$i\in[K]$ and each $a\in[N]$ with $(a,b)\in \mathcal{T}_{i,j}$, and the
token initially placed on the source for $i$ and $a$ is the column tile
token for $(a,b)$.

Because of the differing roles of column tile tokens and row tile tokens,
the labels of vertices in value subtrees differ depending on whether they
belong to a choice subtree for a column or for a row. In $\colsub_j(b)$, the
tile token source for $i$ and $a$ has label set
$\{\mathrm{replacement}(i,j,a,b),\mathrm{tile}(i,j)\}$ and the local target
has label set $\{\mathrm{replacement}(i,j,a,b)\}$. Thus, a column tile token
can move either to the target for its cell in the tile subtree or to its own
local target. Moreover, the only sources compatible with the local target of
the column tile token for $(a,b)$ are the source of that token itself and
the source of the row tile token for $(a,b)$ in $\rowsub_i(a)$.

In $\rowsub_i(a)$, the tile token source for $j$ and $b$ has label set
$\{\mathrm{replacement}(i,j,a,b),\mathrm{clean}(i,j)\}$ and its local target
has label set $\{\mathrm{clean}(i,j)\}$. We will use these labels to show
that a row tile token moves either to the local target of the column tile
token for the same tile, or to its own local target.

\smallskip
\noindent\textbf{Storage paths.}
The \emph{row-storage path} contains $K$ target vertices. For every
$i\in[K]$, the $i$th target has label set
$\{\mathrm{seal}(\rowsub_i(a)):a\in[N]\}$ and is therefore compatible with
every sealing token in the choice subtree $\rowsub_i$, but with no sealing
token from any other row choice subtree.

Similarly, the \emph{column-storage path} contains $K$ target vertices. For
every $j\in[K]$, the $j$th target has label set
$\{\mathrm{seal}(\colsub_j(b)):b\in[N]\}$ and is therefore compatible with
every sealing token in the choice subtree $\colsub_j$, but with no sealing
token from any other column choice subtree.

\smallskip
\noindent\textbf{Cleanup path.}
The \emph{cleanup path} contains the following source vertices:
\begin{enumerate}
\item for every $(i,j)\in[K]^2$, one source with label set $\{{\mathrm{clean}(i,j)}\}$, the \emph{cell source} for cell $(i,j)$; it can be used to occupy the local target left unoccupied when the corresponding row tile token is moved to a choice subtree for a column;
\item for every $i\in[K]$, the $i$th \emph{row source} has label set $\{{\mathrm{seal}(\rowsub_i(a)):a\in[N]}\}$; it can be used to replace a sealing token that moves from $\rowsub_i$ to the row-storage path;
\item for every $j\in[K]$, the $j$th \emph{column source} has label set $\{{\mathrm{seal}(\colsub_j(b)):b\in[N]}\}$; it can be used to replace a sealing token that moves from $\colsub_j$ to the column-storage path.
\end{enumerate}

\smallskip
\noindent\textbf{Connections.}
All vertices joining these parts are ordinary. In each choice subtree, the
attachment vertices of the $N$ value subtrees are adjacent to a common
ordinary vertex. For the row choice subtrees $\rowsub_1,\ldots,\rowsub_K$,
these common vertices are adjacent, in this order, to distinct vertices of
a path of ordinary vertices, and for the column choice subtrees
$\colsub_1,\ldots,\colsub_K$, to distinct vertices of a second such path.
One end of the first path is adjacent to one end of the second path; the
former end is also adjacent to an end of the row-storage path, and the
latter to an end of the column-storage path. The other end of the second
path is adjacent to the center of the tile subtree, and the other end of
the first path is adjacent to an end of the cleanup path. Along the cleanup
path, starting from that end, the $K^2$ cell sources come first, followed
by the $K$ row sources and then the $K$ column sources.

\begin{figure}[!tp]
\centering
\includegraphics[width=\linewidth]{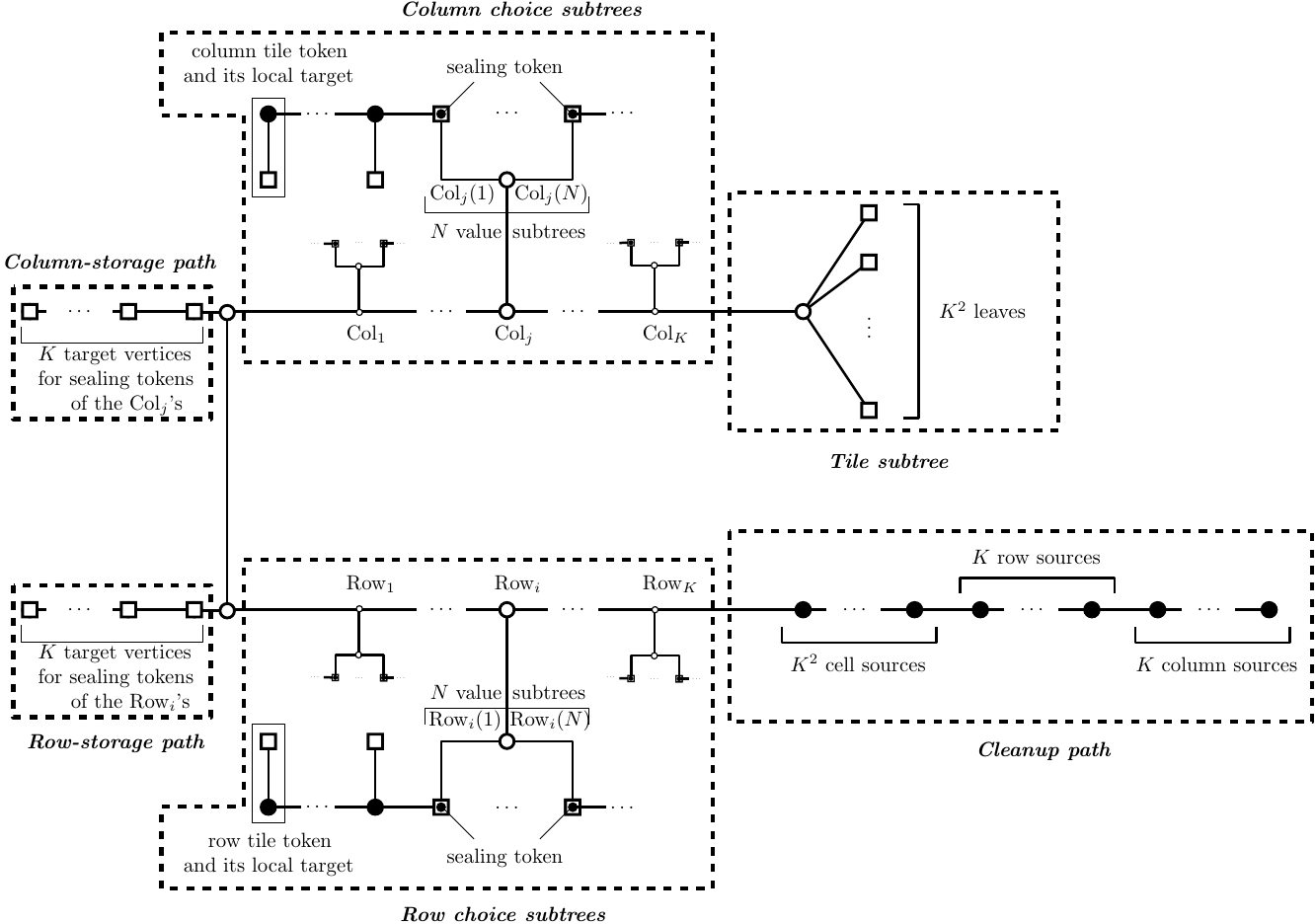}
\caption{The arrangement of the choice subtrees, the tile subtree, the two storage paths, and the cleanup path in the instance constructed from a \textsc{Grid Tiling} instance.}
\label{fig:tree_w1_hard}
\end{figure}

In the constructed tree, the vertices of degree greater than three are the center of the tile subtree, the common neighbors of the attachment vertices in the choice subtrees, and the vertices at which the storage paths and the cleanup path are attached; all of them are ordinary, since every source and every target of the construction lies inside a value subtree, is a leaf of the tile subtree, or lies on a storage path or on the cleanup path. We may therefore apply the operation of Lemma~\ref{lem:subcubic-expansion} to each of them and obtain, in polynomial time, an equivalent instance in which the tree has maximum degree three.

\noindent
For ease of reference, we summarize which sources and targets use each label, as well as their uses:

\begin{description}

\item[Tile labels.] The label $\mathrm{tile}(i,j)$ is associated with cell $(i,j)$. It appears as a source label for the column tile token for $(a,b)$ for each legal tile $(a, b)$ for cell $(i, j)$, and as the target label of a single vertex in the tile subtree.
\item[Sealing labels.] The label $\mathrm{seal}(\rowsub_i(a))$ determines whether the value subtree $\rowsub_i(a)$ is sealed. It labels a vertex in $\rowsub_i(a)$ that is both source and target, the $i$th row source in the cleanup path, and the $i$th target in the row-storage path. Similarly, $\mathrm{seal}(\colsub_j(b))$ labels a vertex in the value subtree $\colsub_j(b)$ that is both source and target, the $j$th column source in the cleanup path, and the $j$th target in the column-storage path.
\item[Replacement labels.] The label $\mathrm{replacement}(i,j,a,b)$, for each cell $(i,j)\in[K]^2$ and each tile $(a,b)\in \mathcal{T}_{i,j}$, appears in the label sets of the tile token source and of the local target in $\colsub_j(b)$ and of the corresponding tile token source in $\rowsub_i(a)$. Thus, either the tile token in $\colsub_j(b)$ or the tile token in $\rowsub_i(a)$ can move to the local target in $\colsub_j(b)$.
\item[Cleanup labels.] The label $\mathrm{clean}(i,j)$ labels a source in the cleanup path and, for every $a,b\in[N]$ with $(a,b)\in\mathcal{T}_{i,j}$, the tile token source for $j$ and $b$ in $\rowsub_i(a)$ together with its local target.
\end{description}

\smallskip
\noindent
Below we present two properties of the construction that are used in the proof.

\begin{observation}\label{obs:value-subtree-attachment}
Every route from a tile token source in a value subtree to a target outside
that value subtree contains the source of the sealing token of that value
subtree as an internal vertex. Moreover, every local target and every target
for a cell in the tile subtree is a leaf of the constructed tree.
\end{observation}

\begin{proof}
By construction, the only edge joining a value subtree to the rest of the tree
is incident with the vertex that is both the source and the target of its
sealing token, and that vertex is not a tile token source; hence every path
leaving the value subtree passes through it, and it is an internal vertex of
such a path. The second statement holds because each local target is
introduced as a leaf neighbor of a tile token source and each target for a
cell is introduced as a leaf of the tile subtree.
\end{proof}

As mentioned above, a value subtree is \emph{unsealed} if its sealing token is matched with the corresponding target on a storage path, and \emph{sealed} otherwise; by the summary of sealing labels, the sealing token of a sealed value subtree is matched with the target at its own source.

Throughout this subsection, the given instance of \textsc{Grid Tiling} is the one fixed above, with parameters $K$ and $N$ and legal-tile sets $\mathcal{T}_{i,j}$, and $J_{\mathrm{GT}}$ denotes the constructed instance of \textsc{Single-Move Labeled Token Routing} after the operation of Lemma~\ref{lem:subcubic-expansion} has been applied to its vertices of degree greater than three, so that its tree has maximum degree three.

\begin{figure}[!tp]
\centering
\includegraphics[width=\linewidth]{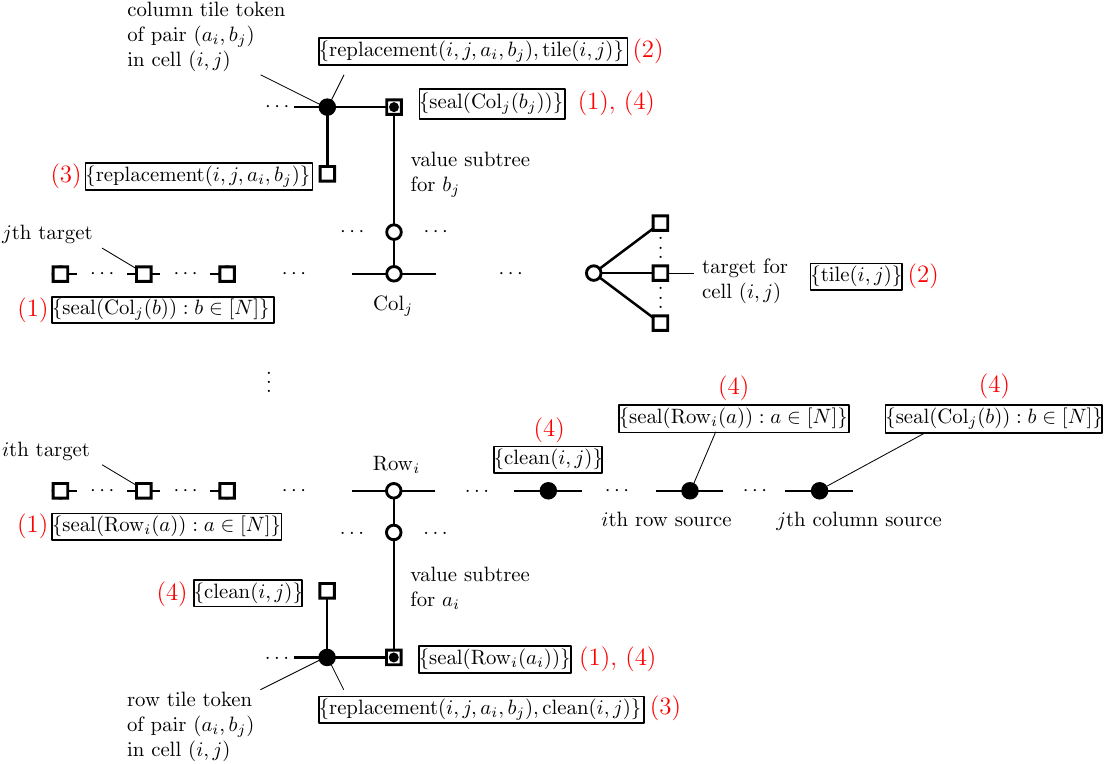}
\caption{The movement order associated with a fixed cell $(i,j)$, assuming that the value $a_i$ is chosen for row $i$ and the value $b_j$ for column $j$. The numbers in parentheses identify the order in which the moves are to be executed, and their positions indicate the source and the target of that move. The label sets of the indicated sources and targets are displayed inside rectangles. Ellipses represent portions of the construction not relevant to the depicted moves: the gadgets retain their relative positions from Figure~\ref{fig:tree_w1_hard}.}
\label{fig:w1_chain_of_moves}
\end{figure}

\begin{lemma}\label{lem:forward-solution-multiplicity}
The constructed \textsc{Single-Move Labeled Token Routing} instance can be computed in polynomial time. Moreover, if the \textsc{Grid Tiling} instance admits a solution, then $J_{\mathrm{GT}}$ admits a solution with solution edge multiplicity at most $3K^2+4K$.
\end{lemma}

\begin{proof}
The construction contains $\mathcal{O}\bigl(KN+\sum_{(i,j)\in[K]^2}|\mathcal{T}_{i,j}|+K^2\bigr)$ vertices and edges. The exhaustive applications of Lemma~\ref{lem:subcubic-expansion} increase the size by at most a polynomial factor. Hence, the constructed instance can be computed in polynomial time.

Suppose that the \textsc{Grid Tiling} instance admits a solution. Let $a_1,\ldots,a_K\in[N]$ and $b_1,\ldots,b_K\in[N]$ satisfy $(a_i,b_j)\in \mathcal{T}_{i,j}$ for every $(i,j)\in[K]^2$. We construct a compatible matching, use the unique tree path for each matched pair as its route, and specify a feasible movement order. The movement order is depicted in Figure~\ref{fig:w1_chain_of_moves}.

\begin{enumerate}
\item \textbf{Sealing tokens ((1) in Figure~\ref{fig:w1_chain_of_moves}).}
For every row $i\in[K]$, match the sealing token in $\rowsub_i(a_i)$ with the $i$th target of the row-storage path. For every $a\in[N]\setminus\{a_i\}$, match the sealing token in $\rowsub_i(a)$ with the target at its own source. Thus, $\rowsub_i(a_i)$ is unsealed, whereas all other value subtrees of $\rowsub_i$ remain sealed.

Similarly, for every column $j\in[K]$, match the sealing token in $\colsub_j(b_j)$ with the $j$th target of the column-storage path, and match all other sealing tokens in $\colsub_j$ with the targets at their own source. Thus, $\colsub_j(b_j)$ is unsealed, whereas all other value subtrees of $\colsub_j$ remain sealed.

Move the sealing tokens matched with storage targets before moving any tile
token, and move those matched with the same storage path in order of
decreasing distance of their storage targets from the vertex at which that
path is attached to the rest of the tree. Then, when a sealing token moves, no
storage target on its route is occupied.

\item \textbf{Column tile tokens ((2) in Figure~\ref{fig:w1_chain_of_moves}).}
For every cell $(i,j)\in[K]^2$, match the column tile token for $(a_i,b_j)$ in $\colsub_j(b_j)$ with the target for cell $(i,j)$ in the tile subtree. These vertices are compatible through the label $\mathrm{tile}(i,j)$.

Match all other column tile tokens with their local targets. Thus, the local target of the column tile token for $(a_i,b_j)$ remains unoccupied, whereas all other local targets of column tile tokens are occupied by the corresponding tile tokens.

In each column value subtree, move the column tile tokens from closest to farthest from the source of the value subtree's sealing token. The routes of the column tile tokens moved to the tile subtree are clear because the relevant value subtrees are unsealed and every tile token source closer to the sealing token's source has already been vacated. Since the targets in the tile subtree and the local targets of column tile tokens are leaves (Observation~\ref{obs:value-subtree-attachment}), occupying one of these targets does not block a later route of a column tile token.

\item \textbf{Row tile tokens ((3) in Figure~\ref{fig:w1_chain_of_moves}).}
For every cell $(i,j)\in[K]^2$, match the row tile token for $(a_i,b_j)$ in $\rowsub_i(a_i)$ with the local target of the column tile token for $(a_i,b_j)$ in $\colsub_j(b_j)$. These vertices are compatible through the label $\mathrm{replacement}(i,j,a_i,b_j)$.

Match all other row tile tokens with their local targets. Thus, for every cell $(i,j)$, the row tile token for $(a_i,b_j)$ occupies the local target left unoccupied by the corresponding column tile token.

After all column tile tokens have moved, move the row tile tokens in every row value subtree from closest to farthest from its sealing token's source. The routes to local targets in column choice subtrees are clear: the value subtrees $\rowsub_i(a_i)$ and $\colsub_j(b_j)$ are unsealed, every row tile token source closer to the sealing token's source has already been vacated, and so have all column tile token sources.

\item \textbf{Cleanup tokens ((4) in Figure~\ref{fig:w1_chain_of_moves}).}
For every cell $(i,j)\in[K]^2$, the local target of the row tile token for $(a_i,b_j)$ remains unoccupied. Match the cleanup token with label $\mathrm{clean}(i,j)$ with this target.

For every row $i\in[K]$, match the $i$th row source of the cleanup path with the target at the source of the sealing token in $\rowsub_i(a_i)$. Similarly, for every column $j\in[K]$, match the $j$th column source of the cleanup path with the target at the source of the sealing token in $\colsub_j(b_j)$.

First move the cleanup tokens with labels of the form $\mathrm{clean}(i,j)$, from closest to farthest from the vertex at which the cleanup path is attached to the rest of the tree. At this point, every row tile token source on the corresponding route has been vacated. Next, move the remaining cleanup tokens, again from closest to farthest from that vertex. Since all tile tokens have already moved, the routes to the sources of the sealing tokens of the unsealed value subtrees are clear.
\end{enumerate}

Every source and every target occurs exactly once in the proposed matching. Every sealing token in a sealed value subtree is matched with the target at its own source, while every sealing token in an unsealed value subtree is replaced by the appropriate cleanup token. Every target for a cell in the tile subtree is matched with one column tile token. For every cell, the local target of this column tile token is matched with the corresponding row tile token, and the local target of that row tile token is matched with the cleanup token carrying the corresponding cleanup label. Every remaining tile token is matched with its local target.

All nontrivial routes are clear when they are used. Therefore, the matching, the unique tree paths between its matched pairs, and the movement order form a solution of $J_{\mathrm{GT}}$.

The only routes that can use an edge shared with another route are the $K^2$ routes from column choice subtrees to the tile subtree, the $K^2$ routes from row choice subtrees to local targets in column choice subtrees, the $K^2$ routes from the cleanup path to local targets in row choice subtrees, the $2K$ routes from unsealed value subtrees to the storage paths, and the $2K$ routes from the cleanup path to the sources of the sealing tokens in unsealed value subtrees. Every remaining route is trivial or consists of the single edge joining a tile token source to its local target; since that local target is a leaf matched with no other source, no other route uses this edge. Hence, every edge, including every edge of a tree introduced by the operation of Lemma~\ref{lem:subcubic-expansion}, is used by at most $3K^2+4K$ routes, so this solution has solution edge multiplicity at most $3K^2+4K$.
\end{proof}

We next establish the properties of an arbitrary solution of $J_{\mathrm{GT}}$ that are needed for the reverse direction.

\begin{lemma}\label{lem:at-most-one-unsealed-value-subtree}
For every row $i\in[K]$, at most one value subtree of $\rowsub_i$ is unsealed. Similarly, for every column $j\in[K]$, at most one value subtree of $\colsub_j$ is unsealed.
\end{lemma}

\begin{proof}
Fix a row $i\in[K]$. As mentioned in the summary of sealing labels, every sealing token in $\rowsub_i$ that is not matched with the target at its own source must be matched with the $i$th target of the row-storage path. Since a matching uses this target at most once, at most one sealing token in $\rowsub_i$ can be matched with the row-storage path. Hence, at most one value subtree of $\rowsub_i$ is unsealed.

The argument for a column choice subtree is symmetric.
\end{proof}

\begin{lemma}\label{lem:cell-verification}
For every cell $(i,j)\in[K]^2$, there exists a legal tile $(a,b)\in \mathcal{T}_{i,j}$ such that the following hold:
\begin{enumerate}
\item the target for cell $(i,j)$ in the tile subtree is matched with the source of the column tile token for $(a,b)$ in $\colsub_j(b)$;
\item this column tile token is moved from $\colsub_j(b)$ to the target for cell $(i,j)$ in the tile subtree;
\item the local target of this column tile token is matched with the source of the row tile token for $(a,b)$ in $\rowsub_i(a)$; and
\item this row tile token is moved from $\rowsub_i(a)$ to the local target of the column tile token.
\end{enumerate}
Consequently, $\colsub_j(b)$ and $\rowsub_i(a)$ are unsealed.
\end{lemma}

\begin{proof}
Fix a cell $(i,j)\in[K]^2$. As discussed in the summary of tile labels, the only sources compatible with the target for cell $(i,j)$ in the tile subtree are the sources of column tile tokens for legal tiles of cell $(i,j)$. Since every target is matched, this target is matched with the source of a column tile token for some legal tile $(a,b)\in \mathcal{T}_{i,j}$ in $\colsub_j(b)$. This proves the first statement.

The target for cell $(i,j)$ lies outside $\colsub_j(b)$, so the corresponding column tile token is moved from $\colsub_j(b)$ to that target, which proves the second statement. By Observation~\ref{obs:value-subtree-attachment}, its route contains the source of the sealing token of $\colsub_j(b)$ as an internal vertex, so that source cannot remain occupied when the column tile token moves. Hence, $\colsub_j(b)$ is unsealed.

As noted in the summary of replacement labels, the local target of this column tile token is compatible only with the source of the column tile token itself and with the source of the corresponding row tile token for $(a,b)$ in $\rowsub_i(a)$. The column tile token is already matched with the target for cell $(i,j)$ in the tile subtree and therefore cannot be matched with its local target. Hence, this local target is matched with the source of the corresponding row tile token. This proves the third statement.

Finally, the local target of the column tile token lies outside $\rowsub_i(a)$. Thus, the corresponding row tile token is moved from $\rowsub_i(a)$ to this target, which proves the fourth statement. By Observation~\ref{obs:value-subtree-attachment}, its route contains the source of the sealing token of $\rowsub_i(a)$ as an internal vertex, so that source cannot remain occupied when the row tile token moves. Hence, $\rowsub_i(a)$ is unsealed.
\end{proof}

\solutionMultiplicityHardness*
\begin{proof}
By Lemma~\ref{lem:forward-solution-multiplicity}, the mapping from the \textsc{Grid Tiling} instance to $J_{\mathrm{GT}}$ can be computed in polynomial time. We show that the former is a yes-instance if and only if the latter is.

\smallskip
\noindent
($\Rightarrow$) Suppose that the \textsc{Grid Tiling} instance is a yes-instance. Lemma~\ref{lem:forward-solution-multiplicity} yields a solution of $J_{\mathrm{GT}}$. Hence, $J_{\mathrm{GT}}$ is a yes-instance.

\smallskip
\noindent
($\Leftarrow$) Suppose that $J_{\mathrm{GT}}$ is a yes-instance.

For every row $i\in[K]$, apply Lemma~\ref{lem:cell-verification} to cell $(i,1)$. This yields a value $a\in[N]$ such that $\rowsub_i(a)$ is unsealed. By Lemma~\ref{lem:at-most-one-unsealed-value-subtree}, no other value subtree of $\rowsub_i$ is unsealed. Write $a_i$ for this value.

Similarly, for every column $j\in[K]$, apply Lemma~\ref{lem:cell-verification} to cell $(1,j)$. This yields a value $b\in[N]$ such that $\colsub_j(b)$ is unsealed. By Lemma~\ref{lem:at-most-one-unsealed-value-subtree}, no other value subtree of $\colsub_j$ is unsealed. Write $b_j$ for this value.

Now fix a cell $(i,j)\in[K]^2$. By Lemma~\ref{lem:cell-verification}, there exists a legal tile $(a,b)\in \mathcal{T}_{i,j}$ such that $\colsub_j(b)$ and $\rowsub_i(a)$ are unsealed. By the definitions of $a_i$ and $b_j$, together with Lemma~\ref{lem:at-most-one-unsealed-value-subtree}, we have $a=a_i$ and $b=b_j$. Therefore, $(a_i,b_j)=(a,b)\in \mathcal{T}_{i,j}$.

Since this holds for every cell $(i,j)\in[K]^2$, the values $a_1,\ldots,a_K$ and $b_1,\ldots,b_K$ form a solution of the \textsc{Grid Tiling} instance.

Thus, the construction is a polynomial-time reduction from \textsc{Grid Tiling}. Since \textsc{Grid Tiling} is \textsf{NP}-hard, \textsc{Single-Move Labeled Token Routing} is \textsf{NP}-hard on trees of maximum degree three. Together with membership in \textsf{NP}, this proves \textsf{NP}-completeness.

For the parameterized statement, map the \textsc{Grid Tiling} instance to
the pair $(J_{\mathrm{GT}},\mu)$ with $\mu=3K^2+4K$; this bound is output on every input,
whether or not the \textsc{Grid Tiling} instance is solvable, and it is
bounded by a computable function of $K$. If the \textsc{Grid Tiling}
instance is solvable, then $J_{\mathrm{GT}}$ admits a solution with solution edge
multiplicity at most $\mu$ by
Lemma~\ref{lem:forward-solution-multiplicity}. Conversely, the backward direction above converts an arbitrary solution of $J_{\mathrm{GT}}$, so in particular
one with solution edge multiplicity at most $\mu$, into a solution of the
\textsc{Grid Tiling} instance. Hence, $(J_{\mathrm{GT}},\mu)$ is a yes-instance of the
parameterized problem if and only if the \textsc{Grid Tiling} instance is a
yes-instance, so the construction is a parameterized reduction from
\textsc{Grid Tiling}, parameterized by $K$. Since \textsc{Grid Tiling} is
\textsf{W}[1]-hard when parameterized by $K$, the claimed
\textsf{W}[1]-hardness follows.
\end{proof}

\smallskip
\noindent\textbf{Prescribed movement order.}
It remains to prove the part of
Theorem~\ref{thm:sol_multiplicity_hardness} concerning the variant in which
the input additionally specifies an ordering of the tokens. Every pair of a
compatible matching contains exactly one source, and every token initially
occupies one source, so an ordering of the tokens induces a movement order
for every compatible matching; a solution of the variant must use the
induced movement order. The same reduction proves both hardness results for
this variant.

Indeed, augment each constructed instance with the ordering of its tokens
obtained by concatenating the following five sequences, in this order.
\begin{enumerate}
\item The sealing tokens; those associated with the same storage path
appear in decreasing order of the distance between their storage targets
and the vertex at which that path is attached to the rest of the tree.
\item The column tile tokens; those of the same value subtree appear in
increasing order of the distance between their sources and the source of
the sealing token of that value subtree.
\item The row tile tokens, ordered in the same way.
\item The cleanup tokens carrying labels of the form $\mathrm{clean}(i,j)$,
in increasing order of the distance between their sources and the vertex at
which the cleanup path is attached to the rest of the tree.
\item The remaining cleanup tokens, ordered in the same way.
\end{enumerate}
Ties are broken by an arbitrary fixed ordering of the vertices.

This ordering is determined solely by the constructed instance and does not depend on the values of a solution of the \textsc{Grid Tiling} instance. The induced movement order is feasible for the route family of the solution constructed in the proof of Lemma~\ref{lem:forward-solution-multiplicity} from any \textsc{Grid Tiling} solution. In both orders, all sealing tokens move before tile tokens, all column tile tokens move before row tile tokens, and all row tile tokens move before cleanup tokens, and the orders within value subtrees and paths are those used in that proof to ensure that the relevant routes are clear. Conversely, every solution of the variant is also a solution of the unrestricted instance, so the reverse direction applies unchanged. The ordering does not affect the forward solution, whose solution edge multiplicity remains at most $3K^2+4K$.

Therefore, even when a movement order is given as part of the input, \textsc{Single-Move Labeled Token Routing} is \textsf{NP}-complete on trees of maximum degree three and \textsf{W}[1]-hard when parameterized by the input bound $\mu$ on the solution edge multiplicity.

\subsection{Candidate edge multiplicity on trees of unbounded degree}
\label{sec:cm-hardness}
Our last hardness result shows that bounding the candidate edge multiplicity
alone does not make the problem tractable: on trees of unbounded degree,
\textsc{Single-Move Labeled Token Routing} remains \textsf{NP}-complete even
when the candidate edge multiplicity is at most eight.

We reduce from \textsc{4-Bounded 3-SAT}. An instance consists of a Boolean
formula $\varphi$ in conjunctive normal form in which every clause contains
exactly three literals of three distinct variables, and every variable
occurs at most four times. The question is whether $\varphi$ admits a
satisfying assignment. This problem is
\textsf{NP}-complete~\cite{DBLP:journals/dam/Tovey84}. We use the
terminology for occurrences of Section~\ref{sec:planar}: an occurrence of a
literal of $x$ or $\bar x$ in a clause $C$ is an occurrence \emph{of} the
variable $x$ and \emph{in} the clause $C$, it is \emph{positive} or
\emph{negative} according to its literal, and it is \emph{true} or
\emph{false} under a truth assignment according to the value of its
literal. Since the three literals of a clause have distinct variables, a
variable has at most one occurrence in a clause and at most four
occurrences in total, and a clause has exactly three occurrences.

To encode $\varphi$ as a tree, we construct one \emph{variable gadget} for
each variable and one \emph{clause gadget} for each clause. All gadgets are
attached to a common ordinary vertex $g$, whose degree is therefore the
number of variables plus the number of clauses. Every vertex other than $g$ has
constant degree, and the candidate edge multiplicity is at most eight, for
two reasons that Lemma~\ref{lem:cm-construction-properties} verifies once
the construction is in place: a gadget has constant size, because a
variable has at most four occurrences and a clause has three, and, for each
occurrence, only two compatible pairs have their source and their target in
different gadgets.

Figure~\ref{fig:para_np_tree} shows a variable gadget and a clause gadget.
We first describe the tokens of each gadget and the targets they may
occupy; the vertices, edges, and labels of the gadgets are specified in the
construction below. The variable gadget for a variable $x$ determines the truth value of $x$. Its tokens are the \emph{assignment token}
$a_x$ and the \emph{complement token} $\bar{a}_x$, and its targets are the
two \emph{assignment targets} $\mathrm{true}_x$ and $\mathrm{false}_x$. The
labels will ensure that, in every compatible matching, $a_x$ is matched
with one assignment target and $\bar a_x$ with the other. Matching $a_x$
with $\mathrm{true}_x$ represents setting $x$ to \texttt{true}; matching it
with $\mathrm{false}_x$ represents setting $x$ to \texttt{false}. The clause gadget for a clause $C$ determines which occurrence in $C$ is
designated to satisfy $C$. Its tokens are a \emph{choice token} $c_o$
for each occurrence $o$ in $C$ and two \emph{filler tokens}, and its
targets are a \emph{selection target} $\mathrm{sel}_o$ for each occurrence
$o$ in $C$ and two \emph{non-selection targets}. The labels will ensure
that, in every compatible matching, exactly one choice token is matched
with its selection target, the other two choice tokens are matched with the
two non-selection targets, and the two filler tokens are matched with the
two selection targets left free. The occurrence whose choice token is
matched with its selection target is \emph{selected}; the other two
occurrences are \emph{unselected}. We will show that a solution exists only
if the selected occurrence of $C$ is true.

Each occurrence $o$ of a variable $x$ in a clause $C$ contributes two
further tokens, a \emph{literal token} $\ell_o$ and a \emph{witness token}
$w_o$, and these are the only tokens whose routes leave a gadget. Their routes are placed so that, as we will show, they impose precedence
constraints between $a_x$ and $c_o$. The token $a_x$ must move before
$\ell_o$, and, if $o$ is selected, $\ell_o$ must move before $c_o$. The
token $c_o$ must move before $w_o$, and, if $o$ is false, $w_o$ must move
before $a_x$. If $o$ is both
selected and false, the four constraints form a cycle, so no feasible
movement order exists. A compatible matching thus determines a truth
assignment and one selected occurrence per clause, and we will show that it
admits a feasible movement order exactly when every selected occurrence is
true.

\smallskip
\noindent\textbf{Construction.}
Tokens are named by the initials of their roles ($a$, $c$, $\ell$, and $w$
for \emph{assignment}, \emph{choice}, \emph{literal}, and \emph{witness},
respectively; the complement token $\bar a_x$ is named after the assignment
token it complements). For each token $z$, we write $\src(z)$ for its
source and, when the label of $z$ is private, so that its target is
determined, $\tgt(z)$ for its target. The other targets share their
compatible sources with another target and are named individually. When
occurrences are numbered $o_1,o_2,\ldots$, as in
Figure~\ref{fig:para_np_tree}, we write $\ell_i$, $w_i$, $c_i$, and
$\mathrm{sel}_i$ for $\ell_{o_i}$, $w_{o_i}$, $c_{o_i}$, and
$\mathrm{sel}_{o_i}$.

\begin{description}
\item[Variable gadgets.] For each variable $x$, we introduce an ordinary
vertex $g_x$ adjacent to $g$ and, adjacent to $g_x$, the sources
$\src(a_x)$ and $\src(\bar{a}_x)$ and the assignment targets
$\mathrm{true}_x$ and $\mathrm{false}_x$. For each occurrence $o$ of $x$,
we introduce the source $\src(\ell_o)$ of the literal token, adjacent to
$\src(a_x)$, and the target $\tgt(w_o)$ of the witness token, adjacent to
$\mathrm{false}_x$ if $o$ is positive and to $\mathrm{true}_x$ if $o$ is
negative. The target of $\ell_o$ and the source of $w_o$ lie in the clause
gadget of the clause of $o$.

\item[Clause gadgets.] For each clause $C$, we introduce an ordinary vertex
$g_C$ adjacent to $g$. For each occurrence $o$ in $C$, we introduce the
source $\src(c_o)$ of the choice token, the source $\src(w_o)$ of the
witness token, the selection target $\mathrm{sel}_o$, and the target
$\tgt(\ell_o)$ of the literal token, forming the two paths
$g_C,\src(c_o),\src(w_o)$ and $g_C,\mathrm{sel}_o,\tgt(\ell_o)$. Finally,
we introduce the sources of the two filler tokens of $C$ and the two
non-selection targets of $C$, all four adjacent to $g_C$.
\end{description}

\begin{figure}[!tp]
\centering
\includegraphics[width=\linewidth]{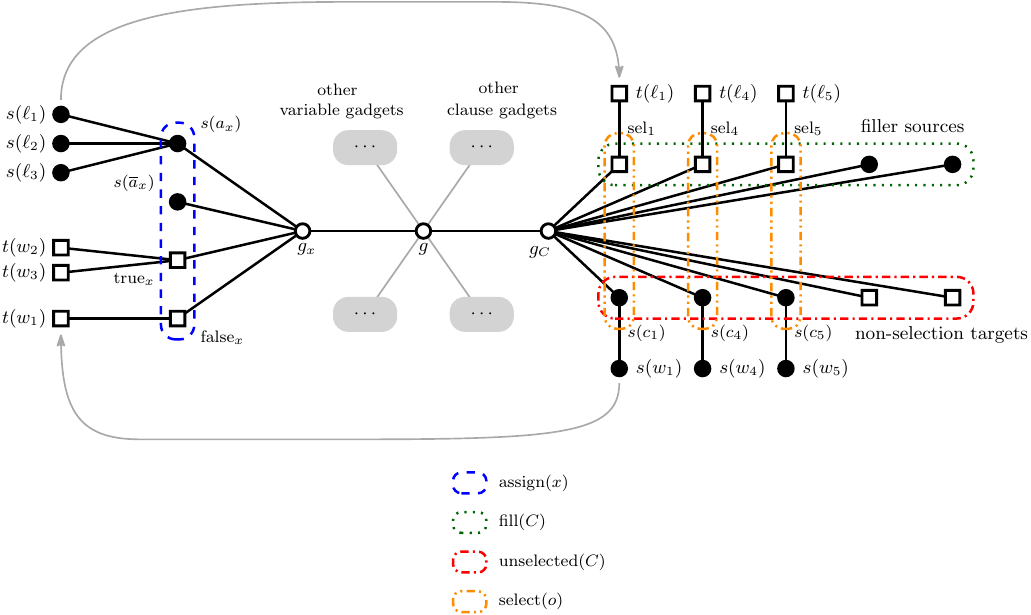}
\caption{Part of the construction of Theorem~\ref{thm:para_np_tree}: the
variable gadget for a variable $x$ with occurrences $o_1$, $o_2$, and $o_3$, of which $o_1$ is positive and lies in the clause $C$ while $o_2$
and $o_3$ are negative and lie in two other clauses, and the clause gadget
for $C$, whose occurrences are $o_1$, $o_4$, and $o_5$. A dashed enclosure
surrounds the sources and targets that carry one shared label, and a
vertex inside two enclosures carries both. The two arrows mark the two
tokens of $o_1$ that move between the gadgets.}
\label{fig:para_np_tree}
\end{figure}

\noindent
For ease of reference, we list, for each label, the sources and targets
that carry it and the compatibilities that it creates; in Figure~\ref{fig:para_np_tree}, the sources and targets carrying one shared label are surrounded by one enclosure. No label is carried by any
source or target other than those listed.

\begin{description}
\item[Assignment labels.] The label $\mathrm{assign}(x)$, for each variable
$x$, labels the sources $\src(a_x)$ and $\src(\bar{a}_x)$ and the targets
$\mathrm{true}_x$ and $\mathrm{false}_x$. Thus, $a_x$ and $\bar{a}_x$ are
each compatible with both assignment targets of $x$ and with no other
target, and each assignment target is compatible only with these two
sources.

\item[Literal labels.] The label $\mathrm{literal}(o)$, for each occurrence
$o$, is private: it labels only $\src(\ell_o)$ and $\tgt(\ell_o)$. It forces
$\ell_o$ to move from the variable gadget of its variable to the clause
gadget of its clause.

\item[Witness labels.] The label $\mathrm{witness}(o)$, for each occurrence
$o$, is likewise private: it labels only $\src(w_o)$ and $\tgt(w_o)$, and
forces $w_o$ to move from the clause gadget to the variable gadget.

\item[Selection labels.] The label $\mathrm{select}(o)$, for each occurrence
$o$, labels $\src(c_o)$ and $\mathrm{sel}_o$.

\item[Filling labels.] The label $\mathrm{fill}(C)$, for each clause $C$,
labels the sources of the two filler tokens of $C$ and the selection target
$\mathrm{sel}_o$ for each occurrence $o$ in $C$. Thus, each filler token of
$C$ is compatible with the three selection targets of $C$, and
$\mathrm{sel}_o$ is compatible with $c_o$ and with the two filler tokens of
$C$.

\item[Non-selection labels.] The label $\mathrm{unselected}(C)$, for each
clause $C$, labels $\src(c_o)$ for each occurrence $o$ in $C$ and the two
non-selection targets of $C$. Thus, each choice token of $C$ is compatible
with its selection target and with the two non-selection targets of $C$,
and each non-selection target is compatible with the three choice tokens of
$C$.
\end{description}
\smallskip
Let $J_{\mathrm{3SAT}}=(G,S,T,L)$ denote the constructed instance. For an occurrence $o$,
we call $(\src(\ell_o),\tgt(\ell_o))$ the \emph{literal pair} and
$(\src(w_o),\tgt(w_o))$ the \emph{witness pair} of $o$. Both pairs are
compatible, and, since their labels are private, both belong to every
compatible matching of $J_{\mathrm{3SAT}}$. We first establish the local properties of $J_{\mathrm{3SAT}}$, among them the bound on its candidate edge multiplicity.

\begin{lemma}\label{lem:cm-construction-properties}
The instance $J_{\mathrm{3SAT}}$ can be computed in polynomial time, $G$ is a tree, and
$\cm(J_{\mathrm{3SAT}})\le 8$. Moreover, every source is compatible with at most three
targets, every target is compatible with at most three sources, and every
vertex other than $g$ has degree at most eleven.
\end{lemma}

\begin{proof}
Each variable, clause, and occurrence contributes a constant number of
vertices, edges, and labels, so $|J_{\mathrm{3SAT}}|$ is linear in $|\varphi|$ and $J_{\mathrm{3SAT}}$ can
be computed in polynomial time. Every gadget is a tree, every gadget is
joined to the rest of $G$ by a single edge incident with $g$, and no edge
joins two gadgets; hence $G$ is a tree.

The compatibility bounds follow from the label list: $a_x$ and $\bar a_x$
are compatible with two targets, $\ell_o$ and $w_o$ with one, and every
choice token and every filler token with three. Symmetrically, every
assignment target is compatible with two sources, $\tgt(\ell_o)$ and
$\tgt(w_o)$ with one, and every selection target and every non-selection
target with three. For the degrees, $g_C$ is adjacent to $g$, to the three
choice sources, the three selection targets, the two filler sources, and the
two non-selection targets of its gadget, and to nothing else, so it has
degree $11$; the vertices $g_x$, $\src(a_x)$, $\mathrm{true}_x$, and
$\mathrm{false}_x$ are adjacent to at most five vertices, since $x$ has at
most four occurrences; every remaining vertex other than $g$ has degree at
most $2$.

We bound the candidate edge multiplicity edge by edge, counting every
compatible pair whose path uses the edge, whether or not the pair belongs
to a solution. Deleting an edge $e$ of $G$ leaves two components, and the
path of a compatible pair uses $e$ exactly when the pair has one endpoint
in each of them; since $g$ lies in one of the two components, the pairs
whose path uses $e$ are those with exactly one endpoint in the component
of $G-e$ that does not contain $g$. By the label list, the only compatible
pairs whose source and target lie in different gadgets are the literal
pair and the witness pair of each occurrence; each remaining compatible
pair joins two vertices of one gadget that are both adjacent to $g_x$ or
both adjacent to $g_C$. For $e=gg_x$, the component of $G-e$ not
containing $g$ is the variable gadget of $x$, which contains an endpoint
of the literal pair and of the witness pair of each of the at most $4$
occurrences of $x$ and no endpoint of any other literal or witness pair,
so $gg_x$ lies on at most $8$ paths; likewise, $gg_C$ lies on $6$ paths.
For an edge $e=g_xv$ with $v$ one of $\src(a_x)$, $\src(\bar{a}_x)$,
$\mathrm{true}_x$, and $\mathrm{false}_x$, the component of $G-e$ not
containing $g$ consists of $v$ and its leaves, so $e$ lies on the paths of
the $2$ compatible pairs that contain $v$ and a vertex of the gadget, and
of at most $4$ literal or witness pairs ending at leaves of $v$, hence on
at most $6$ paths. For an edge $e=g_Cv$ with $v$ another vertex of the
clause gadget, the component of $G-e$ not containing $g$ consists of $v$
and its leaf, if it has one, so $e$ lies on the paths of the at most $3$
compatible pairs that contain $v$ and a vertex of the gadget, and of at
most $1$ literal or witness pair, hence on at most $4$ paths. Every
remaining edge is incident with a leaf and lies on exactly $1$ path.
Therefore, $\cm(J_{\mathrm{3SAT}})\le 8$.
\end{proof}

We now turn to correctness. Since $G$ is a tree, a compatible matching $M$
determines its route family and hence the precedence digraph of
Section~\ref{sec:prelim}, which we denote by $D(M)$. Its vertices are the
pairs of $M$; to simplify notation, we name each pair by its token, so
that an arc from a token to another means that the former must move
before the latter under $M$. By Observation~\ref{obs:precedence-acyclicity},
$M$ is the matching of a solution if and only if $D(M)$ is acyclic.

\begin{lemma}\label{lem:cm-matching-structure}
Let $M$ be a compatible matching of $J_{\mathrm{3SAT}}$.
\begin{enumerate}
\item For each variable $x$, $M$ contains either
$(\src(a_x),\mathrm{true}_x)$ and $(\src(\bar a_x),\mathrm{false}_x)$, or
$(\src(a_x),\mathrm{false}_x)$ and $(\src(\bar a_x),\mathrm{true}_x)$.
\item For each clause $C$, there is exactly one occurrence $o$ in $C$ with
$(\src(c_o),\mathrm{sel}_o)\in M$.
\end{enumerate}
\end{lemma}

\begin{proof}
Fix a variable $x$. By the assignment labels, $\src(a_x)$ and
$\src(\bar a_x)$ are compatible only with $\mathrm{true}_x$ and
$\mathrm{false}_x$, and these two targets are compatible with no other
source. Since $M$ matches all four vertices, its restriction to these four vertices is one of the two stated possibilities.

Fix a clause $C$. By the non-selection labels, the two non-selection
targets of $C$ are compatible only with the three choice tokens of $C$, so
$M$ matches them with two distinct choice tokens. The remaining choice
token $c_o$ is compatible only with $\mathrm{sel}_o$ besides the two
non-selection targets, so $(\src(c_o),\mathrm{sel}_o)\in M$. The two
filler tokens of $C$ are compatible only with the selection targets of
$C$, so $M$ matches them with the two selection targets other than
$\mathrm{sel}_o$, and $o$ is the only occurrence in $C$ whose choice token
is matched with its selection target.
\end{proof}

Let $M$ be a compatible matching. By Lemma~\ref{lem:cm-matching-structure},
every clause $C$ has exactly one occurrence $o$ with
$(\src(c_o),\mathrm{sel}_o)\in M$; we call it the \emph{selected}
occurrence of $C$, and the other two occurrences in $C$
\emph{unselected}. By the same lemma, $M$ matches $\src(a_x)$ with
$\mathrm{true}_x$ or with $\mathrm{false}_x$ for every variable $x$, so $M$
determines a truth assignment $\tau_M$ by setting
$\tau_M(x)=\texttt{true}$ in the first case and
$\tau_M(x)=\texttt{false}$ in the second.

For the reverse direction of the reduction, we will show that if $M$ is
the matching of a solution, then $\tau_M$ satisfies $\varphi$. Since every
clause has a selected occurrence, it suffices to show that $D(M)$ has a
directed cycle whenever some selected occurrence is false under $\tau_M$: by
Observation~\ref{obs:precedence-acyclicity}, the matching of a solution then
has no false selected occurrence, so $\tau_M$ satisfies every clause. The next lemma
gives the arcs of $D(M)$ that $\ell_o$ and $w_o$ create, and shows that,
when $o$ is selected but false under $\tau_M$, four of them form such a
cycle. Two of the four exist in every compatible matching; the other two
depend on the choices made in $M$, one on whether $c_o$ is matched with
$\mathrm{sel}_o$ and one on whether $a_x$ is matched with the assignment
target adjacent to $\tgt(w_o)$.

\begin{lemma}\label{lem:cm-cycle}
Let $M$ be a compatible matching of $J_{\mathrm{3SAT}}$, and let $o$ be an occurrence of a
variable $x$. Then $D(M)$ contains the arcs $(a_x,\ell_o)$ and $(c_o,w_o)$;
it contains the arc $(\ell_o,c_o)$ if $o$ is selected, and the arc
$(w_o,a_x)$ if $o$ is false under $\tau_M$. Consequently, if $o$ is
selected and false under $\tau_M$, then $D(M)$ contains a directed cycle,
and $M$ admits no feasible movement order.
\end{lemma}

\begin{proof}
Let $C$ be the clause of $o$. By construction, the route of $\ell_o$
contains $\src(a_x)$ and $\mathrm{sel}_o$ as internal vertices, and the
route of $w_o$ contains $\src(c_o)$ and the assignment target adjacent to
$\tgt(w_o)$ as internal vertices. By the definition of $D(M)$, the sources
$\src(a_x)$ and $\src(c_o)$ on these routes give the arcs $(a_x,\ell_o)$
and $(c_o,w_o)$; which arcs the two remaining vertices give depends on the tokens that $M$ matches with them.

If $o$ is selected, then $\mathrm{sel}_o$ is the target of $c_o$ and lies
on the route of $\ell_o$, so $D(M)$ contains the arc $(\ell_o,c_o)$. The
assignment target adjacent to $\tgt(w_o)$ is $\mathrm{false}_x$ if $o$ is
positive and $\mathrm{true}_x$ if $o$ is negative, so by
Lemma~\ref{lem:cm-matching-structure} and the definition of $\tau_M$ it is
the target of $a_x$ exactly when $o$ is false under $\tau_M$; since it
lies on the route of $w_o$, $D(M)$ contains the arc $(w_o,a_x)$ in that
case.

If $o$ is selected and false under $\tau_M$, the four arcs form the
directed cycle $a_x,\ell_o,c_o,w_o,a_x$ in $D(M)$, so by
Observation~\ref{obs:precedence-acyclicity} no feasible movement order
exists for $M$.
\end{proof}

We now have all the statements needed to prove this subsection's main result.

\candidateMultiplicityHardness*

\begin{proof}
By Lemma~\ref{lem:cm-construction-properties}, $J_{\mathrm{3SAT}}$ can be computed in
polynomial time, $G$ is a tree, $\cm(J_{\mathrm{3SAT}})\le 8$, every source is compatible
with at most three targets, and every target is compatible with at most
three sources. We show
that $\varphi$ is satisfiable if and only if $J_{\mathrm{3SAT}}$ is a yes-instance.

\smallskip
\noindent
($\Leftarrow$) Suppose that $J_{\mathrm{3SAT}}$ is a yes-instance, and let $M$ be the
compatible matching of a solution. We show that $\tau_M$, the truth
assignment that sets $x$ to \texttt{true} if $M$ matches $\src(a_x)$ with
$\mathrm{true}_x$ and to \texttt{false} if it matches $\src(a_x)$ with
$\mathrm{false}_x$, satisfies $\varphi$. Fix a clause $C$, and let $o$ be
its selected occurrence, which exists by
Lemma~\ref{lem:cm-matching-structure}. If $o$ were false under $\tau_M$,
then by Lemma~\ref{lem:cm-cycle} the matching $M$ would admit no feasible
movement order, contradicting the choice of $M$. Hence, $o$ is true under
$\tau_M$, so $\tau_M$ satisfies $C$. Since this holds for every clause,
$\tau_M$ satisfies $\varphi$.

\smallskip
\noindent
($\Rightarrow$) Suppose that $\varphi$ has a satisfying assignment $\tau$.
Define a compatible matching $M$ as follows. For each variable $x$, include
$(\src(a_x),\mathrm{true}_x)$ and $(\src(\bar a_x),\mathrm{false}_x)$ if
$\tau(x)=\texttt{true}$, and $(\src(a_x),\mathrm{false}_x)$ and
$(\src(\bar a_x),\mathrm{true}_x)$ otherwise. For each occurrence $o$,
include the literal pair and the witness pair of $o$. For each clause $C$,
fix an occurrence $o$ in $C$ that is true under $\tau$; include
$(\src(c_o),\mathrm{sel}_o)$, match the other two choice tokens of $C$
with the two non-selection targets of $C$, and match the two filler tokens
of $C$ with the two remaining selection targets of $C$. Each source and
each target occurs in exactly one pair, and every pair is compatible, so
$M$ is a compatible matching with $\tau_M=\tau$ in which the occurrence
fixed for each clause is selected. In particular, every selected occurrence is true under $\tau$. We now construct a feasible movement order. We move the tokens in
seven stages, moving all tokens of a stage, in any order, before those of
the next; the vertices $g$, $g_x$, and $g_C$ are ordinary and hence never
occupied.

\begin{enumerate}
\item Move every choice token of an unselected occurrence to its
non-selection target, along a route whose only internal vertex is $g_C$.
The route is clear.

\item Move every witness token $w_o$ of an unselected occurrence $o$ that
is false under $\tau$. Its route passes through $\src(c_o)$, vacated in
stage~1, and through the assignment target adjacent to $\tgt(w_o)$, which
is still unoccupied because no token of a variable gadget has moved; its
target $\tgt(w_o)$ is unoccupied. The route is clear.

\item Move every assignment token $a_x$ to the assignment target matched
with it, along a route whose only internal vertex is $g_x$. The witness
targets occupied in stage~2 are leaves adjacent to that assignment target
and do not lie on this route, and the assignment target itself is
unoccupied.

\item Move every literal token $\ell_o$. Its route passes through
$\src(a_x)$, where $x$ is the variable of $o$, vacated in stage~3, and
through $\mathrm{sel}_o$ and $\tgt(\ell_o)$, which are still unoccupied
because no choice token or filler token has reached a selection target.
The route is clear.

\item For each clause $C$ with selected occurrence $o$, move $c_o$ along
$\src(c_o),g_C,\mathrm{sel}_o$; the route ends at $\mathrm{sel}_o$ and does not contain the occupied vertex $\tgt(\ell_o)$. Then move the two
filler tokens of $C$ to the two remaining selection targets of $C$, along
routes whose only internal vertex is $g_C$.

\item Move every witness token $w_o$ that has not yet moved. The
occurrence $o$ is true under $\tau$, so the assignment target adjacent to
$\tgt(w_o)$ is matched with $\bar a_x$, where $x$ is the variable of $o$,
and is still unoccupied; $\src(c_o)$ has been vacated, in stage~1 if $o$
is unselected and in stage~5 if it is selected. The route is clear.

\item Move every complement token $\bar a_x$ to the assignment target not
taken by $a_x$, which is unoccupied, along a route whose only internal
vertex is $g_x$; the witness targets occupied in stage~6 are leaves
adjacent to that assignment target and do not lie on the route.
\end{enumerate}

Every token moves exactly once along the unique path of its pair in $M$.
Hence $M$, together with this movement order, is a solution, and $J_{\mathrm{3SAT}}$ is a
yes-instance.

\smallskip
Thus, the construction is a polynomial-time reduction from
\textsc{4-Bounded 3-SAT}, so \textsc{Single-Move Labeled Token Routing} is
\textsf{NP}-hard under the stated restrictions. Together with membership in
\textsf{NP}, this proves \textsf{NP}-completeness.
\end{proof}

The high degree of the construction is concentrated at $g$, and
Lemma~\ref{lem:subcubic-expansion} cannot remove it without losing the
bound on the candidate edge multiplicity. The lemma preserves feasibility,
but replacing $g$ by a tree $F_g$ of maximum degree three routes the
literal and witness pairs of all occurrences through the edges of $F_g$,
and an edge of $F_g$ then lies on the path of every such pair whose two
gadgets it separates, a number that is not bounded by a constant for an
arbitrary formula. Indeed, by
Theorem~\ref{thm:tree_candidate_multiplicity_fpt}, bounded degree together
with bounded candidate edge multiplicity yields tractability, so any
hardness result with bounded candidate edge multiplicity requires unbounded
maximum degree, unless $\textsf{P}=\textsf{NP}$.

\section{Positive results beyond the all-matching case}
\label{sec:positive}

The hardness results of the previous section show that restricting the
underlying graph to a tree of maximum degree three is not, by itself,
sufficient to obtain tractability. We first show that \textsc{Single-Move
Labeled Token Routing} can be solved efficiently on paths and when the number
of tokens is bounded. We then return to trees under a bound on the candidate
edge multiplicity together with the maximum degree; informally, bounding the
candidate edge multiplicity limits the number of possible matching choices
between the sources and the targets across any edge of the tree. Finally, we
show that the same algorithm yields fixed-parameter tractability when the
problem is parameterized by the candidate vertex multiplicity alone.

\subsection{Paths}
\label{sec:paths}

If an instance on a path admits a solution, then the matching used by
that solution is forced. Thus, it suffices to determine this matching, test
the compatibility of its pairs, and compute a feasible movement order.

\pathTime*

\begin{proof}
Let $P$ be the underlying path, designate one of its endpoints as the left
endpoint, and number the sources from left to right as $s_1,\ldots,s_k$ and the
targets from left to right as $t_1,\ldots,t_k$. The algorithm accepts if and
only if $s_i$ and $t_i$ are compatible for every $i$.

We first show that a solution, if one exists, uses the matching
$\{(s_i,t_i):i\in[k]\}$. Consider the first move of an arbitrary
solution and suppose it takes the token at $s_i$ to $t_j$. If $s_i=t_j$,
the route of this move is the one-vertex path at that vertex, which stays
occupied throughout, so no other route uses it, and each remaining source is
matched with a target on its own side of this vertex. Since $i-1$ sources and $j-1$
targets lie to its left, $i=j$, and the two sides are handled by induction.
Assume henceforth that $s_i\ne t_j$. Once $t_j$ is
occupied, no later route may use it, so every remaining source on one side of
$t_j$ must be routed to a target on the same side. There are exactly $j-1$
targets to the left of $t_j$. If $s_i$ lies to the left of $t_j$, then exactly
$j$ sources lie to the left of $t_j$ before the move, and since the route from
$s_i$ to $t_j$ is clear, $s_i$ is the last of them, so $i=j$. If $s_i$ lies to
the right of $t_j$, then exactly $j-1$ sources lie to the left of $t_j$ and the
clear route forces $s_i$ to be the first source to the right of $t_j$, so again
$i=j$. The move therefore matches $s_i$ to $t_i$, and it leaves all pairs of
smaller index to the left of $t_i$ and all pairs of larger index
to the right of it. Applying the argument independently to the two
sides and inducting yields the claim; consequently, if $s_i$ and $t_i$ are
incompatible for some $i$, no solution exists.

Assume now that $s_i$ and $t_i$ are compatible for every $i$. We construct a
feasible movement order for the matching $\{(s_i,t_i):i\in[k]\}$. First move
the tokens whose target lies to the right of their source, in decreasing
order of their sources; then move the tokens whose target lies to the left of
their source, in increasing order of their sources. We check that each route
is clear when its token moves, writing $u<v$ when the vertex $u$ lies to the
left of the vertex $v$. Let $s_i<t_i$. A source $s_j$ with
$s_i<s_j\le t_i$ has $j>i$, so $t_j>t_i\ge s_j$ and the token at $s_j$ moved
before the token at $s_i$; a target $t_j$ with $s_i<t_j<t_i$ has $j<i$, so
$s_j<s_i<t_j$ and the token at $s_j$ has not moved yet. Now let $t_i<s_i$. A
source $s_j$ with $t_i\le s_j<s_i$ has $j<i$, so $t_j<t_i\le s_j$ and the token at
$s_j$ moved before the token at $s_i$; a target $t_j$ with $t_i<t_j<s_i$ has
$j>i$, so $t_j<s_i<s_j$ and the token at $s_j$ has not moved yet. In both
cases the route from $s_i$ to $t_i$ contains no occupied vertex other than
$s_i$.

It remains to bound the running time. The sources and the targets are found
by a single traversal of $P$ in time $\mathcal{O}(n)$, and the movement order
is produced in time $\mathcal{O}(n)$ by listing the tokens whose source
precedes their target in decreasing order of $i$, followed by the remaining
tokens in increasing order of $i$. What remains is to decide, for every $i$,
whether the label set $L_S(s_i)$ of the source $s_i$ and the label set
$L_T(t_i)$ of the target $t_i$ intersect. Assume, as the encoding allows,
that the labels are integers in $[|\Lambda|]=\{1,\ldots,|\Lambda|\}$, and
initialize a Boolean array indexed by $[|\Lambda|]$ with every entry false,
in time $\mathcal{O}(|I|)$. For each $i$ in turn, set the entry of every
label of $L_S(s_i)$ to true, scan $L_T(t_i)$ for a label whose entry is
true, and reset the entries of $L_S(s_i)$ to false. Each label of each set is
handled a constant number of times, so all $k$ tests take time
$\sum_{i=1}^{k}\bigl(|L_S(s_i)|+|L_T(t_i)|\bigr)=\mathcal{O}(|I|)$, and the
algorithm runs in time $\mathcal{O}(|I|)$.
\end{proof}

\subsection{A bounded number of tokens}
\label{sec:tokens}

By Theorem~\ref{thm:sol_multiplicity_hardness}, the problem remains
\textsf{NP}-complete even when a movement order is given as part of the
input. When the number of tokens is small, however, both the matching and the
movement order can be enumerated, and once both are chosen, finding the routes
reduces to a sequence of reachability tests.

\tokenNumberFpt*

\begin{proof}
Let $I=(G,S,T,L)$ be an instance with $|S|=|T|=k$. Fix a compatible matching
$M$, write $M(s)$ for the target matched with a source $s$, and fix an
ordering $(s_1,\ldots,s_k)$ of $S$. Suppose that the tokens initially placed
at $s_1,\ldots,s_{i-1}$ have been moved to $M(s_1),\ldots,M(s_{i-1})$ and
that no other token has moved. Then the occupied vertices other than $s_i$
are exactly the vertices $M(s_j)$ with $j<i$ and $s_j$ with $j>i$, whatever
routes were used. Hence, there is a clear $s_i$-$M(s_i)$ route at this point
if and only if $M(s_i)$ is reachable from $s_i$ in the graph $G_i$ obtained
from $G$ by deleting these vertices, and $I$ is a yes-instance if and only
if, for some compatible matching $M$ and some ordering of $S$, this holds
for every $i\in[k]$; the routes are then any $s_i$-$M(s_i)$ paths in the
graphs $G_i$.

The algorithm enumerates every bijection $M$ from $S$ to $T$, discards it
unless $L_S(s)\cap L_T(M(s))\neq\emptyset$ for every $s\in S$, and, for
each remaining $M$ and each ordering $(s_1,\ldots,s_k)$ of $S$, tests by
breadth-first search whether $M(s_i)$ is reachable from $s_i$ in $G_i$ for
every $i\in[k]$. It accepts if and only if all $k$ tests succeed for some
$M$ and some ordering, which is correct by the previous paragraph. There are
$k!$ bijections and $k!$ orderings, the compatibility check takes time
polynomial in $|I|$, and each of the $k$ breadth-first searches takes time
$\mathcal{O}(n+m)$, so the running time is
$(k!)^2\cdot|I|^{\mathcal{O}(1)}$, which is of the form
$f(k)\cdot|I|^{\mathcal{O}(1)}$.
\end{proof}

\subsection{Bounded candidate edge multiplicity on bounded-degree trees}
\label{sec:treefpt}

The candidate edge multiplicity bounds, for each edge of the tree, the number of compatible source-target pairs whose unique paths use that edge. After rooting the tree, this bounds the number of pairs of a matching whose paths leave any rooted subtree. Bounded degree then bounds the number of child subtrees that are combined at each vertex. Together, these two bounds yield, via dynamic programming, a fixed-parameter algorithm for instances on trees in which both the candidate edge multiplicity and the maximum degree of the tree are bounded. We first present the notation and describe the algorithm. We then prove its correctness (Lemma~\ref{lem:tree-dp-correctness}) and bound its running time (Lemma~\ref{lem:tree-candidate-multiplicity-running-time}) in separate statements, since Section~\ref{sec:treecvmfpt} reuses the algorithm and its correctness proof but bounds its running time by a different parameter.

\smallskip
\noindent\textbf{Rooted tree, candidate pairs, and partial matchings.}
Let $I=(G,S,T,L)$ be a tree instance, and let $\Delta$ be the maximum degree of $G$. For a compatible pair $q=(s,t)$ we write $P_q$ for the unique $s$-$t$ path in $G$, which is the route of $q$ in every solution whose matching contains $q$, and we say that $q$ \emph{crosses} an edge if $P_q$ uses that edge. Root $G$ at an arbitrary vertex $r$. For each vertex $v$, let $G_v$ denote the subtree rooted at $v$. For each vertex $v$ other than $r$, let $e_v$ be the edge joining $v$ to its parent, called the \emph{boundary edge} of $G_v$. We call $v$ the interface vertex of $G_v$: every path of a pair crossing $e_v$ meets $G_v$ at $v$ before crossing $e_v$. Let $\cand(v)$ be the set of compatible pairs that cross $e_v$, called the \emph{candidate pairs} of $v$. Thus, $\cand(v)$ consists of the compatible pairs with one endpoint in $G_v$ and the other outside $G_v$, and $|\cand(v)|=\cm_I(e_v)\le\cm(I)$ by the definition of candidate edge multiplicity. For the root we set $\cand(r)=\emptyset$. A set of compatible pairs is a \emph{partial matching} if every source occurs in at most one pair and every target occurs in at most one pair.

\smallskip
\noindent\textbf{Description of the entries in the dynamic programming table.}
A subtree interacts with the rest of the tree only through compatible pairs
whose paths cross its boundary edge. The interface vertex is the point at
which all such paths leave the subtree, and the table stores the boundary
pairs together with the precedence relation forced among them. For each
vertex $v$, the table contains one entry $\operatorname{DP}[v,X,\prec]$ for each partial matching $X\subseteq\cand(v)$ and each strict partial order $\prec$ on $X$, that is, each transitive and irreflexive relation on $X$. Every entry is either true or false. The set $X$ stands for the pairs of a partial matching that cross $e_v$, and $p\prec q$ stands for the requirement, imposed by the pairs of that matching inside $G_v$, that the move of $p$ precede the move of $q$.
 
We say that a partial matching $M$ \emph{realizes} the state $(X,\prec)$ at $v$ if it satisfies the following conditions.
 
\begin{enumerate}
\item \textbf{Coverage.} Every source in $S\cap V(G_v)$ appears as the source of exactly one pair of $M$, and every target in $T\cap V(G_v)$ appears as the target of exactly one pair of $M$.
 
\item \textbf{Locality.} Every pair of $M$ has at least one endpoint in $V(G_v)$.
 
\item \textbf{Boundary.} $M\cap\cand(v)=X$. For $v\neq r$, this means that the pairs of $M$ that cross $e_v$ are precisely the pairs of $X$; for $v=r$, it means that $X=\emptyset$.
 
\item \textbf{Precedence.} Let $D_v(M)$ be the digraph with vertex set $M$ in which, for two distinct pairs $p=(s_p,t_p)$ and $q=(s_q,t_q)$ of $M$, there is an arc $(p,q)$, meaning that the move of $p$ must precede the move of $q$, if $s_p\in V(G_v)\cap V(P_q)$ or $t_q\in V(G_v)\cap V(P_p)$; we say that such an arc is \emph{witnessed in $G_v$}. The digraph $D_v(M)$ is acyclic, and, for $p,q\in X$, we have $p\prec q$ if and only if $D_v(M)$ contains a directed path from $p$ to $q$. Informally, $p\prec q$ precisely when the matching inside $G_v$ forces $p$ to be moved before $q$.
\end{enumerate}
 
The intended invariant is that $\operatorname{DP}[v,X,\prec]$ is true if and only if some partial matching realizes $(X,\prec)$ at $v$. We write $\emptyset$ both for the empty set and for the empty order.

\smallskip
\noindent\textbf{Initialization.}
Let $v$ be a vertex with no children. We mark the following entries as true; every state not listed below is set to false.
 
\begin{enumerate}
\item If $v\notin S\cup T$, then $\operatorname{DP}[v,\emptyset,\emptyset]$ is true.
 
\item If $v\in S\setminus T$, then, for each pair $(v,t)\in\cand(v)$, the entry $\operatorname{DP}[v,\{(v,t)\},\emptyset]$ is true.
 
\item If $v\in T\setminus S$, then, for each pair $(s,v)\in\cand(v)$, the entry $\operatorname{DP}[v,\{(s,v)\},\emptyset]$ is true.
 
\item Suppose that $v\in S\cap T$. If $(v,v)$ is compatible, then $\operatorname{DP}[v,\emptyset,\emptyset]$ is true. In addition, for every two pairs $(v,t),(s,v)\in\cand(v)$ that together form a partial matching, the entry $\operatorname{DP}[v,X,\prec]$
with $X=\{(v,t),(s,v)\}$ and $(v,t)\prec(s,v)$ is true, the token at $v$
having to leave before the token assigned to $v$ arrives.
\end{enumerate}

\smallskip
\noindent\textbf{Combining child states.}
Let $v$ be a vertex with at least one child. For each child $u$ of $v$,
choose a true entry $\operatorname{DP}[u,X_u,\prec_u]$; choose also a
partial matching $X\subseteq\cand(v)$. If $v\in S\cap T$, we additionally
choose either the compatible self-pair $(v,v)$ or no self-pair. Let
$Y=X\cup\bigcup_u X_u$, where $u$ ranges over the children of $v$. These
are the pairs whose routes use an edge incident with $v$. We check the
following three conditions on $Y$. If all three hold, we compute an order
$\prec$ on $X$ and mark $\operatorname{DP}[v,X,\prec]$ as true.
  
\begin{enumerate}
\item \textbf{Endpoint conditions.} The set $Y$ is a partial matching. Moreover, for each child $u$ of $v$, the set $X_u$ consists precisely of the pairs in $Y$ having exactly one endpoint in $V(G_u)$. If $v$ is not the root, then $X$ consists precisely of the pairs in $Y$ having exactly one endpoint in $V(G_v)$.
 
\item \textbf{Requirements at $v$.} If $v\notin S\cup T$, then no pair in $Y$ has source or target $v$. If $v\in S\setminus T$, then exactly one pair in $Y$ has source $v$, and no pair in $Y$ has target $v$. Symmetrically, if $v\in T\setminus S$, then exactly one pair in $Y$ has target $v$, and no pair in $Y$ has source $v$.
 
If $v\in S\cap T$, then there are two possibilities. If the self-pair $(v,v)$ is chosen, then no pair in $Y$ has source or target $v$, and no pair in $Y$ uses $v$ as an internal vertex of its path. Otherwise, exactly one pair in $Y$ has source $v$, and exactly one pair in $Y$ has target $v$.
 
\item \textbf{Acyclicity.} Let $D$ be the digraph with vertex set $Y$ whose
arcs are the following. For each child $u$ of $v$ and all $p,q\in X_u$
with $p\prec_u q$, the pair $(p,q)$ is an arc. For all distinct
$p=(s_p,t_p)$ and $q=(s_q,t_q)$ in $Y$, the pair $(p,q)$ is an arc whenever
$s_p=v$ and $v\in V(P_q)$, or $t_q=v$ and $v\in V(P_p)$; we say that such
an arc is \emph{witnessed at $v$}. The digraph $D$ is acyclic.
\end{enumerate}
 
We define $\prec$ on $X$ by setting $p\prec q$ if and only if $D$ contains
a directed path from $p$ to $q$; since $D$ is acyclic, $\prec$ is a strict
partial order. Every entry $\operatorname{DP}[v,X,\prec]$ not marked
true by any choice is false. Note that the entry
$\operatorname{DP}[v,X,\prec]$ keeps only $X$ and not all of $Y$. By the
endpoint conditions, each pair of $Y\setminus X$ has both endpoints in
$G_v$, so no later combination can interact with that pair except through
the precedence relation that it induces on the boundary pairs. The
computation above $v$ therefore needs only $X$ and $\prec$.

Figure~\ref{fig:dp-combine} illustrates one combination step, at a vertex
$v$ with children $u$ and $w$. The pairs $p_1$ and $p_4$ have an endpoint
outside $G_v$ and therefore cross $e_v$, the pair $p_3$ joins the two child
subtrees, the pair $p_2$ has source $v$, and the pair $p_5$ lies inside
$G_u$. In $G_u$, the route of $p_1$ contains $\src(p_5)$ and the route of
$p_3$ contains $\tgt(p_5)$, so the arcs witnessed in $G_u$ are $(p_5,p_1)$
and $(p_3,p_5)$. The pair $p_5$ does not cross $e_u$, so the entry chosen
at $u$ has $X_u=\{p_1,p_3\}$ and $p_3\prec_u p_1$. In $G_w$, the route of
$p_3$ contains $\src(p_4)$, so the entry chosen at $w$ has $p_4\prec_w p_3$.
Since $v$ is the source of $p_2$ and lies on the routes of $p_1$, $p_3$, and $p_4$, the arcs $(p_2,p_1)$, $(p_2,p_3)$, and $(p_2,p_4)$ are witnessed at
$v$. The digraph $D$ consists of these three arcs together with $(p_3,p_1)$
and $(p_4,p_3)$, and it is acyclic. Only $p_1$ and $p_4$ cross $e_v$, so
$X=\{p_1,p_4\}$, and the directed path $p_4,p_3,p_1$ of $D$ gives
$p_4\prec p_1$, an order contained in neither child entry.
 
\begin{figure}[h]
\centering
\includegraphics[width=\linewidth]{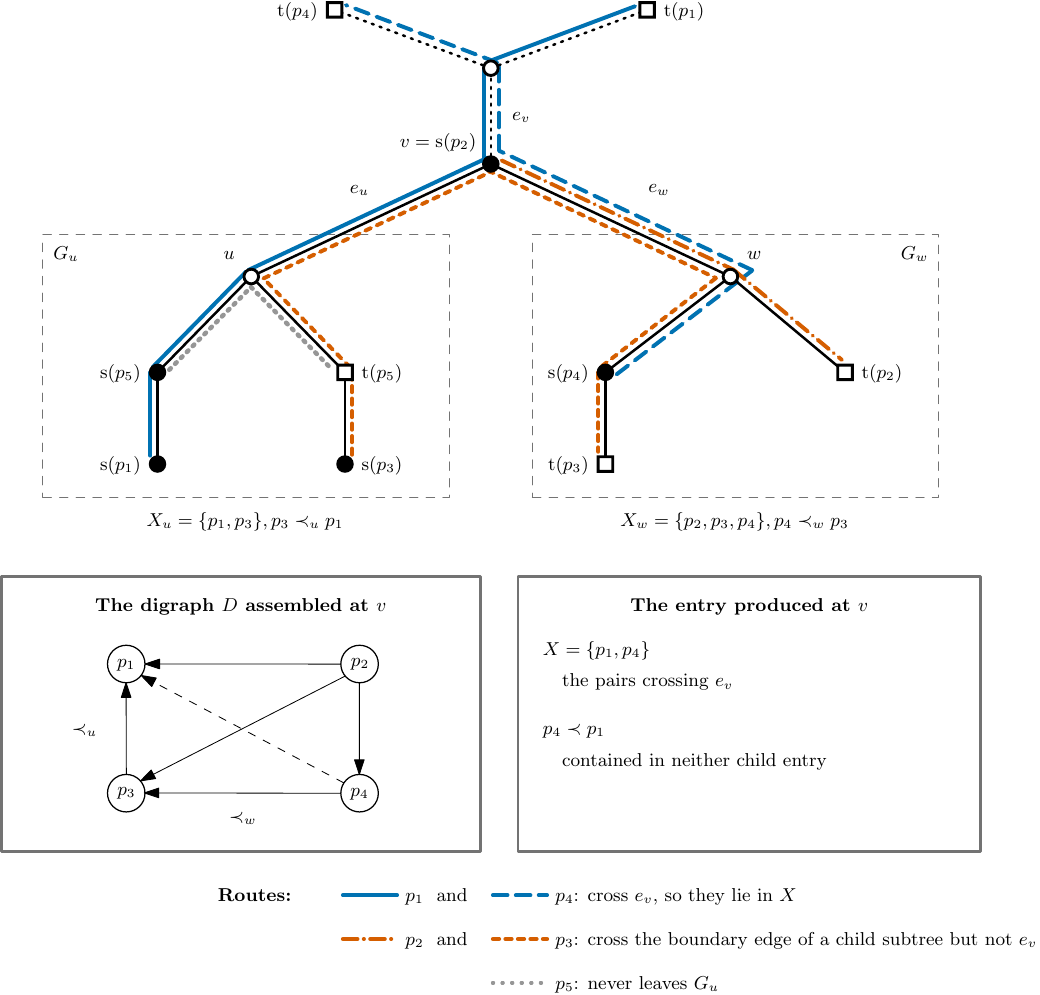}
\caption{One combination step at a vertex $v$ with children $u$ and $w$,
described in the text. The route of each pair of the matching is drawn
beside the tree edges it uses; as the legend shows, its color and dash
pattern indicate whether the pair crosses $e_v$, stops at a child boundary,
or stays inside $G_u$. Below the tree, the entries chosen at $u$ and $w$
are listed under $G_u$ and $G_w$, the left panel shows the digraph $D$, and
the right panel shows the entry produced at $v$; the dashed arrow is the
order $p_4\prec p_1$.}
\label{fig:dp-combine}
\end{figure}

\smallskip
\noindent\textbf{Order of computation and acceptance.}
The algorithm first computes the path $P_q$ of each compatible pair $q$ and the set $\cand(v)$ of each vertex $v$, by enumerating the compatible pairs and, for each of them, the edges of its path; this takes time polynomial in $|I|$. It then processes the vertices in a bottom-up order, and computes the entries $\operatorname{DP}[v,X,\prec]$ of each vertex $v$ by the initialization if $v$ has no children and by combining child states otherwise. Finally, it accepts if and only if $\operatorname{DP}[r,\emptyset,\emptyset]$ is true; since $\cand(r)=\emptyset$, this is the only entry of the form $\operatorname{DP}[r,X,\prec]$.

\smallskip
\noindent\textbf{Correctness.}
We first show that the entries computed by the algorithm have the intended meaning.

\begin{lemma}\label{lem:tree-candidate-multiplicity-table-correctness}
For each vertex $v$, after the algorithm has computed the entries $\operatorname{DP}[v,X,\prec]$, such an entry is true if and only if there exists a partial matching that realizes $(X,\prec)$ at $v$.
\end{lemma}
 
\begin{proof}
We prove the statement by induction on the height of $G_v$.
 
\smallskip
\noindent\textbf{Base case.}
Suppose that $v$ has no children.
 
\begin{enumerate}
\item \textbf{Case $v\notin S\cup T$.}
The empty matching is the only partial matching satisfying the coverage and locality conditions, and it realizes $(\emptyset,\emptyset)$.
 
\item \textbf{Case $v\in S\setminus T$.}
The coverage condition requires exactly one pair with source $v$. Since every pair of the matching must have an endpoint in $G_v$, this pair crosses the boundary edge $e_v$ and therefore belongs to $\cand(v)$. Hence, the realizable states are precisely the states $(\{(v,t)\},\emptyset)$ created in the second initialization case.
 
\item \textbf{Case $v\in T\setminus S$.}
This case is symmetric to the case $v\in S\setminus T$.
 
\item \textbf{Case $v\in S\cap T$.}
A realizing matching either contains the self-pair $(v,v)$, which realizes $(\emptyset,\emptyset)$, or contains two pairs crossing $e_v$, one with source $v$ and one with target $v$. In the latter case, the pair with source $v$ is forced to precede the pair with target $v$. Thus, the realizable states are exactly the entries created in the fourth initialization case.
\end{enumerate}
 
This proves the statement for vertices with no children.
 
\smallskip
\noindent\textbf{Inductive step.}
Let $v$ be a vertex with at least one child, and assume that the statement holds for every child of $v$.
 
\smallskip
\noindent\textbf{Entries produced by the algorithm.}
Consider a choice of child entries and a set $X$ that passes all conditions of the algorithm. By the inductive hypothesis, for each child $u$ of $v$, there is a partial matching $M_u$ that realizes $(X_u,\prec_u)$ at $u$.
 
Construct a set $M$ by taking the union of the matchings $M_u$ over all children $u$ of $v$, treating equal pairs as a single pair, adding every pair in $X$ that is not already present, and adding $(v,v)$ when the self-pair is chosen.
 
The endpoint conditions ensure that every pair appears in the appropriate parts of this union. A pair with endpoints in two distinct child subtrees occurs in the matchings of the corresponding two children. A pair with one endpoint in a child subtree and one endpoint at $v$ occurs only in the matching of that child. A pair with one endpoint in a child subtree and the other outside $G_v$ occurs in the matching of that child and in $X$. Finally, a pair with one endpoint at $v$ and the other outside $G_v$ occurs only in $X$.
 
Since every $M_u$ is a partial matching, since the child subtrees are disjoint, and since $Y$ is a partial matching, the set $M$ is a partial matching. Every source and target in a child subtree is covered by the corresponding child matching. The requirements at $v$ ensure that the source and target roles of $v$, when present, are also covered exactly once. Thus, $M$ satisfies the coverage condition.
 
Every pair inherited from $M_u$ has an endpoint in $G_u$, every pair in $X$ has an endpoint in $G_v$, and the only possible additional pair is $(v,v)$. Hence, $M$ satisfies the locality condition. The endpoint conditions also imply that $M\cap\cand(v)=X$, which is the boundary condition.
 
It remains to verify the precedence condition.
 
\begin{claim}\label{clm:prec-cstr}
The digraph $D_v(M)$ is acyclic if and only if $D$ is acyclic. Moreover, when both are acyclic, for all $p,q\in X$, the digraph $D_v(M)$ contains a directed path from $p$ to $q$ if and only if $D$ does.
\end{claim}
 
\begin{proof}
Every arc of $D_v(M)$ is witnessed either at $v$ itself or in one child subtree. The arcs of $D_v(M)$ witnessed at $v$ are precisely the arcs of $D$ witnessed at $v$.
 
Now fix a child $u$ of $v$. Any pair whose two endpoints belong to $G_u$ has its entire path inside $G_u$. Hence, such a pair cannot impose a precedence constraint outside $G_u$. The only effect of the matching inside $G_u$ on the rest of the tree is therefore the order that it forces among the pairs that cross the boundary edge of $G_u$, namely the pairs in $X_u$. Since $M_u$ realizes $(X_u,\prec_u)$, this order is exactly $\prec_u$.
 
Thus, after replacing the part of the matching inside each child subtree by
the corresponding order $\prec_u$, the digraph $D$ contains all precedence
requirements witnessed in $G_v$ that can affect pairs outside the child
subtrees. Any directed path that enters a child subtree through a boundary
pair and leaves it again can be replaced by the corresponding relation
$\prec_u$; a path that has both endpoints inside the child subtree is handled
entirely by the child state. Consequently, $D_v(M)$ is acyclic if and only
if $D$ is acyclic. Moreover, for $p,q\in X$, the digraph $D_v(M)$ contains
a directed path from $p$ to $q$ if and only if $D$ does. \end{proof}
 
By Claim~\ref{clm:prec-cstr}, the algorithm rejects the current choice precisely when $M$ violates the precedence condition. Otherwise, the order $\prec$ defined by the algorithm is exactly the order required in the definition of a state realized by $M$. Hence, every entry produced by the algorithm is realized by some partial matching.
 
\smallskip
\noindent\textbf{Every realizable state is produced.}
Conversely, suppose that a partial matching $M$ realizes $(X,\prec)$ at $v$. For each child $u$ of $v$, let $M_u$ consist of the pairs of $M$ having at least one endpoint in $G_u$. Let $X_u=M_u\cap\cand(u)$ be the set of pairs of $M_u$ that cross the boundary edge of $G_u$, and, for $p,q\in X_u$, let $p\prec_u q$ if and only if $D_u(M_u)$ contains a directed path from $p$ to $q$.
 
The matching $M_u$ satisfies the coverage and the locality conditions for $G_u$, its pairs crossing the boundary edge of $G_u$ are exactly the pairs in $X_u$, and $D_u(M_u)$ is acyclic, since it is a subdigraph of $D_v(M)$. Thus, $M_u$ realizes $(X_u,\prec_u)$ at $u$. By the inductive hypothesis, $\operatorname{DP}[u,X_u,\prec_u]$ is true for each child $u$ of $v$.
 
Choose these child entries and the set $X$ in the algorithm. Since $M$ is a partial matching, the set $Y$, the union of $X$ and of the sets $X_u$, is also a partial matching. By the definition of the sets $X_u$ and $X$, the endpoint conditions hold.
 
The requirements at $v$ follow from the coverage condition. If $(v,v)\in M$, then no other pair of $M$ has source or target $v$, because $M$ is a partial matching. Moreover, no pair of $M$ uses $v$ as an internal vertex of its path: otherwise, the self-pair and that pair would form a precedence cycle witnessed at $v$. Hence, the self-pair case is valid. If $(v,v)\notin M$, then the coverage condition gives exactly the required number of pairs with source $v$ and of pairs with target $v$.
 
Finally, Claim~\ref{clm:prec-cstr} shows that the digraph $D$ constructed by the algorithm is acyclic and that, for $p,q\in X$, it contains a directed path from $p$ to $q$ if and only if $p\prec q$. Therefore, the algorithm sets $\operatorname{DP}[v,X,\prec]$ to true. This proves the inductive step.
\end{proof}

Applying Lemma~\ref{lem:tree-candidate-multiplicity-table-correctness} at the root yields the correctness of the algorithm.

\begin{lemma}\label{lem:tree-dp-correctness}
The algorithm accepts if and only if $I$ is a yes-instance.
\end{lemma}

\begin{proof}
The algorithm accepts if and only if $\operatorname{DP}[r,\emptyset,\emptyset]$ is true. By Lemma~\ref{lem:tree-candidate-multiplicity-table-correctness}, this holds if and only if there exists a partial matching $M$ that realizes $(\emptyset,\emptyset)$ at $r$. As $G_r=G$, the coverage condition implies that $M$ is a compatible matching of all sources and targets. Moreover, the digraph $D_r(M)$ is the precedence digraph of $M$ and of the route family formed by the unique paths $P_q$, and the precedence condition states that this digraph is acyclic. By Observation~\ref{obs:precedence-acyclicity}, such a compatible matching exists if and only if $I$ admits a solution. Therefore, the algorithm accepts if and only if $I$ is a yes-instance.
\end{proof}

\smallskip
\noindent\textbf{Running time.}
We bound the running time in two steps. Lemma~\ref{lem:tree-dp-vertex-time} bounds the number of entries $\operatorname{DP}[v,X,\prec]$ of a vertex $v$ and the time needed to compute them in terms of the numbers of candidate pairs of $v$ and of its children, without reference to the parameters, and Lemma~\ref{lem:tree-candidate-multiplicity-running-time} bounds these numbers by the candidate edge multiplicity and the maximum degree.

\begin{lemma}\label{lem:tree-dp-vertex-time}
Let $v$ be a vertex of $G$, and let $u$ range over the children of $v$. There are at most $2^{|\cand(v)|^2+|\cand(v)|}$ entries $\operatorname{DP}[v,X,\prec]$, and the algorithm computes them in time $2^{|\cand(v)|+\sum_u(|\cand(u)|^2+|\cand(u)|)}\cdot|I|^{\mathcal{O}(1)}$.
\end{lemma}

\begin{proof}
An entry $\operatorname{DP}[v,X,\prec]$ is indexed by a partial matching $X\subseteq\cand(v)$ and a strict partial order $\prec$ on $X$. There are at most $2^{|\cand(v)|}$ subsets $X$ of $\cand(v)$ and, for each of them, at most $2^{|X|^2}\le 2^{|\cand(v)|^2}$ relations on $X$, since a relation on $X$ is a subset of $X\times X$. This counts all relations and therefore upper-bounds the number of strict partial orders stored by the algorithm. This proves the first statement.

If $v$ has no children, the initialization examines every pair of $\cand(v)$ and, when $v\in S\cap T$, every two pairs of $\cand(v)$, which takes time polynomial in $|I|$. Otherwise, the algorithm examines every choice of one true entry $\operatorname{DP}[u,X_u,\prec_u]$ for each child $u$ of $v$, of a partial matching $X\subseteq\cand(v)$, and, when $v\in S\cap T$, of whether the self-pair $(v,v)$ is chosen. By the first statement applied to each child $u$, the child entries can be chosen in at most $\prod_u 2^{|\cand(u)|^2+|\cand(u)|}=2^{\sum_u(|\cand(u)|^2+|\cand(u)|)}$ ways, the set $X$ in at most $2^{|\cand(v)|}$ ways, and the self-pair in at most two ways. Every choice is processed in time polynomial in $|I|$. Indeed, $v$ has fewer than $n$ children, and every entry consists of at most $|S|\cdot|T|$ compatible pairs and a relation on them, so the chosen entries are read and the set $Y$ is formed in time polynomial in $|I|$. The endpoint conditions and the requirements at $v$ are then checked, and the digraph $D$ on $Y$, its acyclicity, and the order $\prec$ on $X$ are computed, in time polynomial in $|I|$, since $|Y|\le|S|\cdot|T|$ and the paths of the compatible pairs have been computed beforehand. Multiplying the three counts by the polynomial time per choice yields the second statement.
\end{proof}

\begin{lemma}\label{lem:tree-candidate-multiplicity-running-time}
The algorithm runs in time $2^{\mathcal{O}(\Delta\cdot\cm(I)^2)}\cdot|I|^{\mathcal{O}(1)}$.
\end{lemma}

\begin{proof}
The paths of the compatible pairs and the sets $\cand(v)$ are computed in time polynomial in $|I|$. Fix a vertex $v$ of $G$, and let $u$ range over the children of $v$. By the definition of candidate pairs, $|\cand(w)|\le\cm(I)$ for every vertex $w$ of $G$, and $v$ has at most $\Delta$ children, so $|\cand(v)|+\sum_u(|\cand(u)|^2+|\cand(u)|)\le\cm(I)+\Delta\bigl(\cm(I)^2+\cm(I)\bigr)$, which is $\mathcal{O}(\Delta\cdot\cm(I)^2)$. Hence, by Lemma~\ref{lem:tree-dp-vertex-time}, the entries $\operatorname{DP}[v,X,\prec]$ are computed in time $2^{\mathcal{O}(\Delta\cdot\cm(I)^2)}\cdot|I|^{\mathcal{O}(1)}$. Summing over the $n$ vertices of $G$ yields the claimed bound.
\end{proof}

\treeCandidateMultiplicityFpt*

\begin{proof}
By Lemma~\ref{lem:tree-dp-correctness}, the algorithm accepts if and only if $I$ is a yes-instance. By Lemma~\ref{lem:tree-candidate-multiplicity-running-time}, it runs in time $2^{\mathcal{O}(\Delta\cdot\cm(I)^2)}\cdot|I|^{\mathcal{O}(1)}$. Since $\Delta\cdot\cm(I)^2\le(\Delta+\cm(I))^3$, this running time is of the form $f(\Delta+\cm(I))\cdot|I|^{\mathcal{O}(1)}$ for a computable function $f$, so the problem is fixed-parameter tractable when parameterized by $\Delta+\cm(I)$.
\end{proof}

The two parameters of Theorem~\ref{thm:tree_candidate_multiplicity_fpt} enter the proof of Lemma~\ref{lem:tree-candidate-multiplicity-running-time} in different ways. The candidate edge multiplicity bounds every number $|\cand(v)|$ and $|\cand(u)|$ in the exponent of Lemma~\ref{lem:tree-dp-vertex-time}, whereas the maximum degree bounds only the number of children $u$ of $v$, that is, the number of terms $|\cand(u)|^2+|\cand(u)|$ of that exponent. By Theorems~\ref{thm:sol_multiplicity_hardness} and~\ref{thm:para_np_tree}, respectively, neither parameter can be dropped. Section~\ref{sec:treecvmfpt} nevertheless bounds the exponent of Lemma~\ref{lem:tree-dp-vertex-time} without the maximum degree, by a single parameter, at least as large as the candidate edge multiplicity, that counts, for each vertex, the compatible pairs whose paths contain that vertex.

\subsection{Bounded candidate vertex multiplicity}
\label{sec:treecvmfpt}

The maximum degree enters the running time of the algorithm of Section~\ref{sec:treefpt} only through the number of terms of the sum in the exponent of Lemma~\ref{lem:tree-dp-vertex-time}, one for each child of the vertex $v$ of that lemma. In this subsection we bound that exponent without the maximum degree, by replacing the candidate edge multiplicity, which counts for each edge the compatible pairs whose unique paths use that edge, with the candidate vertex multiplicity of Section~\ref{sec:prelim}, which counts for each vertex the compatible pairs whose unique paths contain that vertex. Throughout this subsection, $I=(G,S,T,L)$ is a tree instance, $\Delta$ is the maximum degree of $G$, and the root $r$, the subtrees $G_v$, the boundary edges $e_v$, and the sets $\cand(v)$ of candidate pairs are those of Section~\ref{sec:treefpt}.

We first relate the candidate vertex multiplicity to the parameters of Theorem~\ref{thm:tree_candidate_multiplicity_fpt}.

\begin{observation}\label{obs:cvm-vs-cm}
For every tree instance $I$ on a tree of maximum degree $\Delta$, $\cm(I)\le\cvm(I)\le\Delta\cdot\cm(I)+1$.
\end{observation}

\begin{proof}
Every path that uses an edge contains both of its endpoints, so, for each edge $e=vw$ of $G$, every compatible pair counted in $\cm_I(e)$ is counted in $\cvm_I(v)$, and $\cm_I(e)\le\cvm_I(v)\le\cvm(I)$. This proves the first inequality. For the second, fix a vertex $v$ and consider the compatible pairs whose paths contain $v$. At most one of them has a trivial path, namely the pair $(v,v)$, since a trivial path containing $v$ consists of the vertex $v$ alone. The path of each of the remaining pairs uses at least one edge incident with $v$; assign to each of these pairs one such edge. A pair assigned an edge $e$ has $e$ on its path, so at most $\cm(I)$ pairs are assigned $e$, and at most $\Delta$ edges are incident with $v$. Hence, $\cvm_I(v)\le\Delta\cdot\cm(I)+1$.
\end{proof}

By the second inequality of Observation~\ref{obs:cvm-vs-cm}, bounding the maximum degree together with the candidate edge multiplicity bounds the candidate vertex multiplicity, so Theorem~\ref{thm:tree_cvm_fpt} below implies the fixed-parameter tractability stated in Theorem~\ref{thm:tree_candidate_multiplicity_fpt}. The algorithm is that of Section~\ref{sec:treefpt}, unchanged, so by Lemma~\ref{lem:tree-dp-correctness} the algorithm accepts if and only if $I$ is a yes-instance. Only the bound on its running time changes. Lemma~\ref{lem:tree-dp-vertex-time} bounds the time needed to compute the entries $\operatorname{DP}[v,X,\prec]$ of a vertex $v$ in terms of $|\cand(v)|$ and of $|\cand(u)|$ for the children $u$ of $v$, and the proof of Lemma~\ref{lem:tree-candidate-multiplicity-running-time} bounds these numbers by $\cm(I)$ and the number of children by $\Delta$. Lemma~\ref{lem:cvm-candidate-bounds} below bounds $|\cand(v)|$ and the sum of the numbers $|\cand(u)|$ by the candidate vertex multiplicity of $v$ alone, however many children $v$ has. Its proof uses only that the boundary edge $e_v$, when $v\neq r$, and the boundary edges $e_u$ of the children $u$ of $v$ are all incident with $v$, so that the path of every pair of $\cand(v)$ and of every pair of $\cand(u)$ contains $v$.

\begin{lemma}\label{lem:cvm-candidate-bounds}
Let $v$ be a vertex of $G$, and let $u$ range over the children of $v$. Then $|\cand(v)|\le\cvm_I(v)$ and $\sum_u|\cand(u)|\le 2\cvm_I(v)$.
\end{lemma}

\begin{proof}
If $v=r$, then $\cand(v)=\emptyset$, and the first inequality holds. Otherwise, the boundary edge $e_v$ is incident with $v$, so the path of every pair of $\cand(v)$ uses $e_v$ and therefore contains $v$. Hence, $\cand(v)$ is a subset of the set of compatible pairs whose paths contain $v$, and this set has $\cvm_I(v)$ elements.

For each child $u$ of $v$, the boundary edge $e_u$ joins $u$ to $v$, so the path of every pair of $\cand(u)$ contains $v$ as well. The path of a compatible pair is simple and therefore uses at most two edges incident with $v$, so every compatible pair belongs to $\cand(u)$ for at most two children $u$ of $v$. Hence, the sum $\sum_u|\cand(u)|$ counts only compatible pairs whose paths contain $v$, of which there are $\cvm_I(v)$, and counts each of them at most twice, which proves the second inequality.
\end{proof}

\treeCandidateVertexMultiplicityFpt*

\begin{proof}
We run the algorithm of Section~\ref{sec:treefpt}, which accepts if and only if $I$ is a yes-instance by Lemma~\ref{lem:tree-dp-correctness}, and we bound its running time. The paths of the compatible pairs and the sets $\cand(v)$ are computed in time polynomial in $|I|$. Fix a vertex $v$ of $G$, and let $u$ range over the children of $v$. By Lemma~\ref{lem:tree-dp-vertex-time}, the entries $\operatorname{DP}[v,X,\prec]$ are computed in time $2^{|\cand(v)|+\sum_u(|\cand(u)|^2+|\cand(u)|)}\cdot|I|^{\mathcal{O}(1)}$. By Lemma~\ref{lem:cvm-candidate-bounds} applied to $v$, we have $|\cand(v)|\le\cvm_I(v)\le\cvm(I)$ and $\sum_u|\cand(u)|\le 2\cvm_I(v)\le 2\cvm(I)$, and by Lemma~\ref{lem:cvm-candidate-bounds} applied to each child $u$ of $v$, we have $|\cand(u)|\le\cvm_I(u)\le\cvm(I)$. Therefore, $\sum_u|\cand(u)|^2\le\bigl(\max_u|\cand(u)|\bigr)\cdot\sum_u|\cand(u)|\le 2\cvm(I)^2$, and the exponent in the bound of Lemma~\ref{lem:tree-dp-vertex-time} is at most $2\cvm(I)^2+3\cvm(I)$. Hence, these entries are computed in time $2^{\mathcal{O}(\cvm(I)^2)}\cdot|I|^{\mathcal{O}(1)}$, and summing over the $n$ vertices of $G$ yields the claimed running time.
\end{proof}

Although Theorem~\ref{thm:tree_cvm_fpt} implies the fixed-parameter tractability stated in Theorem~\ref{thm:tree_candidate_multiplicity_fpt}, the two running times are incomparable. Substituting the bound $\cvm(I)\le\Delta\cdot\cm(I)+1$ of Observation~\ref{obs:cvm-vs-cm} into the running time of Theorem~\ref{thm:tree_cvm_fpt} yields $2^{\mathcal{O}(\Delta^2\cdot\cm(I)^2)}\cdot|I|^{\mathcal{O}(1)}$, which is worse than the running time of Theorem~\ref{thm:tree_candidate_multiplicity_fpt}, whereas on trees of unbounded degree and bounded candidate vertex multiplicity only Theorem~\ref{thm:tree_cvm_fpt} applies. By the first inequality of Observation~\ref{obs:cvm-vs-cm}, bounding the candidate vertex multiplicity bounds the candidate edge multiplicity. By Theorem~\ref{thm:tree_cvm_fpt}, no \textsf{NP}-hardness reduction to instances on trees can produce only instances of constant candidate vertex multiplicity, unless $\textsf{P}=\textsf{NP}$. The instances of Theorem~\ref{thm:para_np_tree} indeed have candidate edge multiplicity at most eight but not constant candidate vertex multiplicity: the literal pair and the witness pair of every occurrence have paths through the vertex $g$ of that construction, so the candidate vertex multiplicity of these instances is at least twice the number of occurrences.

\section{Conclusion and open problems}
\label{sec:conclusion}
 
We introduced \textsc{Single-Move Labeled Token Routing}, in which every
token moves at most once from its source to a compatible target. On
unique-matching instances the problem is the discrete form of
\textsc{Monotone Sliding Reconfiguration}~\cite{GeftHalperinNakar2024}, and we proved
that it is \textsf{NP}-complete on planar graphs and on grids even for instances
that, when solvable, admit a solution with pairwise edge-disjoint routes.
Outside the unique-matching case, three hardness results hold
on trees: \textsf{NP}-completeness with two colors on trees of maximum degree
three, proved by Geft~\cite{Geft2026Trees}; \textsf{W}[1]-hardness on trees of maximum
degree three parameterized by an input bound on the solution edge
multiplicity, which holds even when a movement order is given as part of
the input; and \textsf{NP}-completeness on trees of unbounded degree whose
candidate edge multiplicity is at most eight. The last two are proved in
this paper, together with a linear-time algorithm on paths, a
fixed-parameter algorithm parameterized by the number of tokens, and a
fixed-parameter algorithm on trees, parameterized either by the maximum
degree together with the candidate edge multiplicity or by the candidate
vertex multiplicity alone; by the \textsf{NP}-completeness on trees of
maximum degree three and on trees of candidate edge multiplicity at most
eight, neither the maximum degree nor the candidate edge multiplicity can
be omitted.
 
The two extreme compatibility relations remain the subject of open
problems in robotics. Abrahamsen and Halperin list among ten open
problems in geometric robotics the monotone reconfiguration of unit
discs in the plane in which any disc may fill any target and a disc may
initially overlap a target, a geometric counterpart of our all-matching
case~\cite{AbrahamsenHalperin2024Geobotics}.
At the other extreme, Geft, Halperin, and Nakar prove that unique-matching instances on
$3 \times w$ grids remain \textsf{NP}-hard even when all sources lie
to the left of all targets~\cite{GeftHalperinNakar2024}. As soon as there
are more than three tokens, their instances admit no solution with
pairwise edge-disjoint routes, since every route uses one of the three
edges that cross the line separating the sources from the targets.
Theorem~\ref{thm:grid_and_subgrid_hardness} gives hard instances on grids
that admit such a solution, but its proof places a token whose source is
its target at every vertex missing from the grid graph. Whether grids
remain hard for instances that admit a solution with pairwise
edge-disjoint routes and in which no vertex is both a source and a
target is open. A second feature of the
construction of Section~\ref{sec:planar}, which Geft and Halperin
identify as a possible source of hardness in their own
constructions~\cite{GeftHalperin2022Refined}, is that a route may pass
through unboundedly many sources and targets of other tokens: the route
of an assignment token contains the two vertices created for every
positive occurrence of its variable, or for every negative occurrence.
Whether Theorem~\ref{thm:grid_and_subgrid_hardness} holds when every
route contains at most a constant number of sources and targets is open
as well, whereas the construction of Theorem~\ref{thm:para_np_tree} has
this property, since every path of a compatible pair contains at most
two sources or targets other than its endpoints.
 
Further questions concern which measures of the structure of the graph
and of the compatible pairs, alone or together, yield tractability. On
trees, it remains to determine which restrictions on the compatible pairs
yield tractability when one seeks a solution with small solution edge
multiplicity. Beyond trees, bounding the pathwidth does not suffice, as
the problem is \textsf{NP}-complete on $3 \times w$ grid
graphs~\cite{GeftHalperinNakar2024}, whose
pathwidth is at most three, and a graph of small treewidth may have
arbitrarily many compatible pairs. On general graphs it would therefore be
of interest to combine a measure of the structure of the graph with a
measure of the compatibility relation, such as the number of compatible
pairs separated by a small vertex set. Independently of such a
combination, one may ask whether some structural parameter other than a
measure of the graph's width yields fixed-parameter tractability by
itself.
 
The compatibility relation itself raises two further directions. Call
the compatibility relation a \emph{coloring} when every source and every
target carries exactly one label. The \textsf{NP}-completeness on trees
of maximum degree three holds for two colors~\cite{Geft2026Trees},
whereas neither construction of this paper on trees is a coloring: a
column tile token of Section~\ref{sec:tree-hardness} may move either to the tile
subtree or to its local target, a choice token of
Section~\ref{sec:cm-hardness} either to its selection target or to a
non-selection target, and the second alternative is expressed by a second
label. For Theorem~\ref{thm:para_np_tree} this label can be avoided:
giving the three choice tokens of a clause one common color, merging
their selection targets into one target of that color, giving the two
non-selection targets the same color, and removing the filler tokens
yields a coloring for which the reduction goes through unchanged.
Whether Theorem~\ref{thm:sol_multiplicity_hardness} holds for colorings,
and whether Theorems~\ref{thm:sol_multiplicity_hardness}
and~\ref{thm:para_np_tree} hold for a bounded number of labels, is open.
Conversely, is there a graph class or a parameter for which the problem
is polynomial-time solvable or fixed-parameter tractable for colorings
but not for arbitrary compatibility relations? A different restriction
is to constrain the routes rather than the compatible pairs, by giving
for each compatible pair a family of allowed routes, as in
compatible-ordering problems~\cite{DBLP:conf/sofsem/BousquetSMN26}.
 
Finally, two assumptions of the problem can be relaxed. If each token may
move more than once, then a solution must also specify intermediate
positions. On grid graphs, the unique-matching case remains
\textsf{NP}-hard when each token may move at most
twice~\cite{VanDerHorst2026Squares}; we do not know whether
Theorem~\ref{thm:token_number_fpt} extends to this setting. If
$|S| > |T|$, as in atom reconfiguration, then a solution consists of a set of compatible pairs in
which every target occurs in exactly one pair and every source in at most
one pair, a route family, and a movement order; a token whose source is in
no pair never moves, so the move of a pair is feasible only if its route
contains no such source. Our hardness results carry over, since the
instances with $|S| = |T|$ are exactly those in which every source is in a
pair. Our algorithms carry over as well, since each of them constructs or
enumerates the pairs of a solution and can leave a source out of every
pair. The leaf-removal procedure of Theorem~\ref{thm:allMatching} leaves
out a leaf that is a source but not a target whenever the current tree
contains more sources than targets, the dynamic programming of
Section~\ref{sec:treefpt} may leave out a source that is not a target when
no path of a pair examined at that vertex contains it, without changing
the running times of Theorems~\ref{thm:tree_candidate_multiplicity_fpt}
and~\ref{thm:tree_cvm_fpt}, and the algorithm of
Theorem~\ref{thm:token_number_fpt} also enumerates the sources left out.
On a path, the order of the tokens never changes, so the sources left out
can be chosen by dynamic programming along the path in polynomial time,
although we do not know whether the linear running time of
Theorem~\ref{thm:path_time} can be preserved.

\section*{Acknowledgements}
Nicolas Bousquet was partly supported by ANR project ENEDISC
(ANR-24-CE48-7768-01). Remy El Sabeh and Naomi Nishimura's research was
funded by the Natural Sciences and Engineering Research Council of Canada.

\bibliographystyle{unsrt}
\bibliography{main}

\begin{thebibliography}{10}

\bibitem{Calinescu2008Reconfigurations}
Gruia C{\u a}linescu, Adrian Dumitrescu, and J{\'a}nos Pach.
\newblock Reconfigurations in graphs and grids.
\newblock {\em SIAM Journal on Discrete Mathematics}, 22(1):124--138, 2008.

\bibitem{GeftHalperinNakar2024}
Tzvika Geft, Dan Halperin, and Yonatan Nakar.
\newblock Tractability frontiers in multi-robot coordination and geometric
  reconfiguration.
\newblock In {\em Algorithmic Foundations of Robotics {XVI}}, volume~37 of {\em
  Springer Proceedings in Advanced Robotics}, pages 247--267. Springer, 2026.
\newblock Proceedings of WAFR 2024.

\bibitem{Geft2026Trees}
Tzvika Geft.
\newblock {Hardness of Multi-Agent Path Finding on Trees: A Unified Approach}.
\newblock In Philip Bille, Seth Pettie, and Sabine Storandt, editors, {\em 34th
  Annual European Symposium on Algorithms (ESA 2026)}, volume 388 of {\em
  Leibniz International Proceedings in Informatics (LIPIcs)}, pages
  159:1--159:15, Dagstuhl, Germany, 2026. Schloss Dagstuhl -- Leibniz-Zentrum
  f{\"u}r Informatik.

\bibitem{Endres2016AtomAssembly}
Manuel Endres, Hannes Bernien, Alexander Keesling, Harry Levine, Eric~R.
  Anschuetz, Alexandre Krajenbrink, Crystal Senko, Vladan Vuletic, Markus
  Greiner, and Mikhail~D. Lukin.
\newblock Atom-by-atom assembly of defect-free one-dimensional cold atom
  arrays.
\newblock {\em Science}, 354(6315):1024--1027, 2016.

\bibitem{Barredo2016AtomAssembler}
Daniel Barredo, Sylvain de~L{\'e}s{\'e}leuc, Vincent Lienhard, Thierry Lahaye,
  and Antoine Browaeys.
\newblock An atom-by-atom assembler of defect-free arbitrary two-dimensional
  atomic arrays.
\newblock {\em Science}, 354(6315):1021--1023, 2016.

\bibitem{OhlDeMello2019DefectFree}
Daniel~Ohl de~Mello, Daniel Schaffner, Jonathan~P. Knudsen, Thomas~G. Walker,
  and David~S. Weiss.
\newblock Defect-free assembly of {2D} clusters of more than 100 single-atom
  quantum systems.
\newblock {\em Physical Review Letters}, 122(20):203601, 2019.

\bibitem{Schymik2020EnhancedAssembly}
Kai-Niklas Schymik, Vincent Lienhard, Daniel Barredo, Pascal Scholl, Hannah
  Williams, Antoine Browaeys, and Thierry Lahaye.
\newblock Enhanced atom-by-atom assembly of arbitrary tweezer arrays.
\newblock {\em Physical Review A}, 102(6):063107, 2020.

\bibitem{ElSabeh2023ARO}
Remy El~Sabeh, Jessica Bohm, Zhiqian Ding, Stephanie Maaz, Naomi Nishimura,
  Izzat El~Hajj, Amer~E. Mouawad, and Alexandre Cooper.
\newblock Efficient algorithms to solve atom reconfiguration problems. {II.}
  {The} assignment-rerouting-ordering (aro) algorithm.
\newblock {\em Physical Review A}, 108(2):023108, 2023.

\bibitem{Cimring2023RedRec}
Barry Cimring, Remy El~Sabeh, Marc Bacvanski, Stephanie Maaz, Izzat El~Hajj,
  Naomi Nishimura, Amer~E. Mouawad, and Alexandre Cooper.
\newblock Efficient algorithms to solve atom reconfiguration problems. {I.}
  {The} redistribution-reconfiguration (red-rec) algorithm.
\newblock {\em Physical Review A}, 108(2):023107, 2023.

\bibitem{Afiouni2025BirdBatching}
Fouad Afiouni, Remy El~Sabeh, Naomi Nishimura, Izzat El~Hajj, Amer~E Mouawad,
  and Alexandre Cooper.
\newblock Efficient algorithms to solve atom reconfiguration problems. {III.
  The bird} and batching algorithms and other parallel implementations on
  {GPUs}.
\newblock {\em Physical Review A}, 112(2):023109, 2025.

\bibitem{Anand2024DualSpeciesRydberg}
Shraddha Anand, Conor~E. Bradley, Ryan White, Vikram Ramesh, Kevin Singh, and
  Hannes Bernien.
\newblock A dual-species {R}ydberg array.
\newblock {\em Nature Physics}, 20(11):1744--1750, Nov 2024.

\bibitem{Fang2025Interleaved}
Chengyu Fang, Jared Miles, Jonathan Goldwin, Martin Lichtman, Matthew Gillette,
  Michael Bergdolt, Sanket Deshpande, Sam~A. Norrell, Preston Huft, Mikhail~A.
  Kats, and Mark Saffman.
\newblock Interleaved dual-species arrays of single atoms using a passive
  optical element and one trapping laser.
\newblock {\em Science Advances}, 11(29):eadw4166, 2025.

\bibitem{Sheng2022MixedSpecies}
Cheng Sheng, Jiayi Hou, Xiaodong He, Kunpeng Wang, Ruijun Guo, Jun Zhuang,
  Bahtiyar Mamat, Peng Xu, Min Liu, Jin Wang, and Mingsheng Zhan.
\newblock Defect-free arbitrary-geometry assembly of mixed-species atom arrays.
\newblock {\em Physical Review Letters}, 128(8):083202, 2022.

\bibitem{Wei2026MixedSpecies}
Ming-Rui Wei, Kun-Peng Wang, Jia-Yi Hou, Yi~Chen, Peng Xu, Jun Zhuang, Rui-Jun
  Guo, Min Liu, Jin Wang, Xiao-Dong He, and Ming-Sheng Zhan.
\newblock Enhanced atom-by-atom assembly of defect-free two-dimensional
  mixed-species atomic arrays.
\newblock {\em Physical Review Applied}, 25(3):034009, 2026.

\bibitem{Ruttley2023MergingTweezers}
Daniel~K. Ruttley, Alexander Guttridge, Stefan Spence, Robert~C. Bird,
  C.~Ruth~Le Sueur, Jeremy~M. Hutson, and Simon~L. Cornish.
\newblock Formation of ultracold molecules by merging optical tweezers.
\newblock {\em Physical Review Letters}, 130(22):223401, 2023.

\bibitem{Ruttley2024MoleculeArrays}
Daniel~K. Ruttley, Alexander Guttridge, Tom~R. Hepworth, and Simon~L. Cornish.
\newblock Enhanced quantum control of individual ultracold molecules using
  optical tweezer arrays.
\newblock {\em PRX Quantum}, 5(2):020333, 2024.

\bibitem{Cooper2024ParameterizedAtoms}
Alexandre Cooper, Stephanie Maaz, Amer~E. Mouawad, and Naomi Nishimura.
\newblock Parameterized complexity of reconfiguration of atoms.
\newblock {\em Algorithmica}, 2024.

\bibitem{Barredo2018Synthetic3D}
Daniel Barredo, Vincent Lienhard, Sylvain de~L{\'e}s{\'e}leuc, Thierry Lahaye,
  and Antoine Browaeys.
\newblock Synthetic three-dimensional atomic structures assembled atom by atom.
\newblock {\em Nature}, 561:79--82, 2018.

\bibitem{Lee2017Hungarian}
Woojun Lee, Hyosub Kim, and Jaewook Ahn.
\newblock Defect-free atomic array formation using the {Hungarian} matching
  algorithm.
\newblock {\em Physical Review A}, 95(5):053424, 2017.

\bibitem{Tao2022DualSpeciesAlgorithm}
Zhi-Jin Tao, Li-Geng Yu, Peng Xu, Jia-Yi Hou, Xiao-Dong He, and Ming-Sheng
  Zhan.
\newblock Efficient two-dimensional defect-free dual-species atom arrays
  rearrangement algorithm with near-fewest atom moves.
\newblock {\em Chinese Physics Letters}, 39(8):083701, 2022.

\bibitem{Harle2026Atommovr}
Nikhil~K. Harle, Bo-Yu Chen, Bob Bao, and Hannes Bernien.
\newblock atommovr: An open-source simulation framework for rearrangement in
  atomic arrays.
\newblock {\em Quantum}, 10:2177, 2026.

\bibitem{Kornhauser1984PebbleMotion}
Daniel Kornhauser, Gary~L. Miller, and Paul~G. Spirakis.
\newblock Coordinating pebble motion on graphs, the diameter of permutation
  groups, and applications.
\newblock In {\em Proceedings of the 25th Annual Symposium on Foundations of
  Computer Science (FOCS)}, pages 241--250, 1984.

\bibitem{Auletta1999PebbleTrees}
Vincenzo Auletta, Angelo Monti, Mimmo Parente, and Pino Persiano.
\newblock A linear-time algorithm for the feasibility of pebble motion on
  trees.
\newblock {\em Algorithmica}, 23(3):223--245, 1999.

\bibitem{Yu2013LinearPebble}
Jingjin Yu.
\newblock A linear time algorithm for the feasibility of pebble motion on
  graphs, 2013.

\bibitem{Goraly2010MulticolorPebble}
Gilad Goraly and Refael Hassin.
\newblock Multi-color pebble motion on graphs.
\newblock {\em Algorithmica}, 58(3):610--636, 2010.

\bibitem{YuLaValle2016OptimalMPP}
Jingjin Yu and Steven~M. LaValle.
\newblock Optimal multi-robot path planning on graphs: Complete algorithms and
  effective heuristics.
\newblock {\em IEEE Transactions on Robotics}, 32(5):1163--1177, 2016.

\bibitem{Stern2019MAPFDefinitions}
Roni Stern, Nathan Sturtevant, Ariel Felner, Sven Koenig, Hang Ma, Thayne
  Walker, Jiaoyang Li, Dor Atzmon, Liron Cohen, T.~K.~Satish Kumar, Eli
  Boyarski, and Roman Bartak.
\newblock Multi-agent pathfinding: Definitions, variants, and benchmarks, 2019.

\bibitem{Ma2022GraphMAPF}
Hang Ma.
\newblock Graph-based multi-robot path finding and planning.
\newblock {\em Current Robotics Reports}, 2022.

\bibitem{Kedia2022PCMAPF}
Kushal Kedia, Rajat~Kumar Jenamani, Aritra Hazra, and Partha~Pratim
  Chakrabarti.
\newblock Optimal multi-agent path finding for precedence constrained planning
  tasks, 2022.

\bibitem{Surynek2019CapacityMAPF}
Pavel Surynek, T.~K.~Satish Kumar, and Sven Koenig.
\newblock Multi-agent path finding with capacity constraints, 2019.

\bibitem{Chong2024TAPFPTC}
Yu~Quan Chong, Jiaoyang Li, and Katia Sycara.
\newblock Optimal task assignment and path planning using conflict-based search
  with precedence and temporal constraints, 2024.

\bibitem{Yamanaka2015SwappingLabeled}
Katsuhisa Yamanaka, Erik~D. Demaine, Takashi Horiyama, Takehiro Ito, Jun
  Kawahara, Masashi Kiyomi, Yoshio Okamoto, Toshiki Saitoh, Akira Suzuki,
  Ryuhei Uehara, and Takeaki Uno.
\newblock Swapping labeled tokens on graphs.
\newblock {\em Theoretical Computer Science}, 586:81--94, 2015.

\bibitem{Miltzow2016Approximation}
Tillmann Miltzow, Lothar Narins, Yoshio Okamoto, G{\"u}nter Rote, Antonis
  Thomas, and Takeaki Uno.
\newblock Approximation and hardness of token swapping.
\newblock In {\em 24th Annual European Symposium on Algorithms (ESA 2016)},
  volume~57 of {\em Leibniz International Proceedings in Informatics (LIPIcs)},
  pages 66:1--66:15. Schloss Dagstuhl -- Leibniz-Zentrum f{\"u}r Informatik,
  2016.

\bibitem{Alon1994RoutingMatchings}
Noga Alon, Fan R.~K. Chung, and Ronald~L. Graham.
\newblock Routing permutations on graphs via matchings.
\newblock {\em {SIAM} J. Discret. Math.}, 7(3):513--530, 1994.

\bibitem{Kawahara2019PermutationRouting}
Jun Kawahara, Toshiki Saitoh, and Ryo Yoshinaka.
\newblock The time complexity of permutation routing via matching, token
  swapping and a variant.
\newblock {\em Journal of Graph Algorithms and Applications}, 23(1):29--70,
  2019.

\bibitem{MaKoenig2016TAPF}
Hang Ma and Sven Koenig.
\newblock Optimal target assignment and path finding for teams of agents.
\newblock AAMAS '16, page 1144–1152, Richland, SC, 2016. International
  Foundation for Autonomous Agents and Multiagent Systems.

\bibitem{Bonnet2017Complexity}
{\'{E}}douard Bonnet, Tillmann Miltzow, and Pawe{\l} Rzazewski.
\newblock Complexity of token swapping and its variants.
\newblock In {\em 34th Symposium on Theoretical Aspects of Computer Science
  (STACS 2017)}, volume~66 of {\em Leibniz International Proceedings in
  Informatics (LIPIcs)}, pages 16:1--16:14. Schloss Dagstuhl -- Leibniz-Zentrum
  f{\"u}r Informatik, 2017.

\bibitem{DBLP:journals/comgeo/AbellanasBHORT06}
Manuel Abellanas, Sergey Bereg, Ferran Hurtado, Alfredo~Garc{\'{\i}}a Olaverri,
  David Rappaport, and Javier Tejel.
\newblock Moving coins.
\newblock {\em Comput. Geom.}, 34(1):35--48, 2006.

\bibitem{StilmanKuffner2008}
Mike Stilman and James~J. Kuffner.
\newblock {\em Planning Among Movable Obstacles with Artificial Constraints},
  pages 119--135.
\newblock Springer Berlin Heidelberg, Berlin, Heidelberg, 2008.

\bibitem{DBLP:conf/icra/Buckley89}
Stephen~J. Buckley.
\newblock Fast motion planning for multiple moving robots.
\newblock In {\em Proceedings of the 1989 {IEEE} International Conference on
  Robotics and Automation, Scottsdale, Arizona, USA, May 14-19, 1989}, pages
  322--326. {IEEE} Computer Society, 1989.

\bibitem{Diestel}
Reinhard Diestel.
\newblock {\em Graph Theory}, volume 173 of {\em Graduate Texts in
  Mathematics}.
\newblock Springer, 4 edition, 2010.

\bibitem{Tamassia1987}
Roberto Tamassia.
\newblock On embedding a graph in the grid with the minimum number of bends.
\newblock {\em SIAM Journal on Computing}, 16(3):421--444, 1987.

\bibitem{TamassiaTollisVitter2000}
Roberto Tamassia, Ioannis~G. Tollis, and Jeffrey~Scott Vitter.
\newblock A parallel algorithm for planar orthogonal grid drawings.
\newblock {\em Parallel Process. Lett.}, 10(1):141--150, 2000.

\bibitem{cygan2015parameterized}
Marek Cygan, Fedor~V Fomin, {\L}ukasz Kowalik, Daniel Lokshtanov, D{\'a}niel
  Marx, Marcin Pilipczuk, Micha{\l} Pilipczuk, and Saket Saurabh.
\newblock {\em Parameterized algorithms}, volume~5.
\newblock Springer, 2015.

\bibitem{DBLP:conf/atal/AbrahamsenGHU23}
Mikkel Abrahamsen, Tzvika Geft, Dan Halperin, and Barak Ugav.
\newblock Coordination of multiple robots along given paths with bounded
  junction complexity.
\newblock In {\em Proceedings of the 2023 International Conference on
  Autonomous Agents and Multiagent Systems, {AAMAS} 2023}, pages 932--940.
  {ACM}, 2023.

\bibitem{deBergKhosravi2012}
Mark de~Berg and Amirali Khosravi.
\newblock Optimal binary space partitions for segments in the plane.
\newblock {\em International Journal of Computational Geometry \&
  Applications}, 22(3):187--206, 2012.

\bibitem{Marx2007GT}
D{\'a}niel Marx.
\newblock On the optimality of planar and geometric approximation schemes.
\newblock In {\em FOCS 2007}, pages 338--348, 2007.

\bibitem{DBLP:journals/dam/Tovey84}
Craig~A. Tovey.
\newblock A simplified np-complete satisfiability problem.
\newblock {\em Discret. Appl. Math.}, 8(1):85--89, 1984.

\bibitem{AbrahamsenHalperin2024Geobotics}
Mikkel Abrahamsen and Dan Halperin.
\newblock Ten problems in geobotics.
\newblock {\em CoRR}, abs/2408.12657, 2024.

\bibitem{GeftHalperin2022Refined}
Tzvika Geft and Dan Halperin.
\newblock Refined hardness of distance-optimal multi-agent path finding.
\newblock In {\em Proceedings of the 21st International Conference on
  Autonomous Agents and Multiagent Systems, {AAMAS} 2022}, pages 481--488.
  International Foundation for Autonomous Agents and Multiagent Systems
  ({IFAAMAS}), 2022.

\bibitem{DBLP:conf/sofsem/BousquetSMN26}
Nicolas Bousquet, Remy El~Sabeh, Amer~E. Mouawad, and Naomi Nishimura.
\newblock On the complexity of constrained reconfiguration and motion planning.
\newblock In Jakub Kozik and Alexander Wolff, editors, {\em {SOFSEM} 2026:
  Theory and Practice of Computer Science - 51st International Conference on
  Current Trends in Theory and Practice of Computer Science, {SOFSEM} 2026,
  Krak{\'{o}}w, Poland, February 9-13, 2026, Proceedings}, Lecture Notes in
  Computer Science, pages 1--15. Springer, 2026.

\bibitem{VanDerHorst2026Squares}
Thijs van~der Horst, Maarten L{\"o}ffler, Tim Ophelders, and Tom Peters.
\newblock Reconfiguration of squares using a constant number of moves each.
\newblock {\em CoRR}, abs/2603.05203, 2026.

\end{thebibliography}

\end{document}